\documentclass[journal,10pt,letterpaper,onecolumn]{IEEEtran}

\usepackage[utf8]{inputenc}
\usepackage[english]{babel}
\usepackage[babel]{csquotes}

\usepackage{amscd}
\usepackage{amsfonts}
\usepackage{amsmath}
\usepackage{amssymb}
\usepackage[usenames,dvipsnames]{color}
\usepackage{array}
\usepackage{longtable}
\usepackage{xspace}
\usepackage{paralist}
\usepackage{listings}
\usepackage{extarrows}
\usepackage{xifthen}
\usepackage{url}
\usepackage{amsthm}
\newtheorem{theorem}{Theorem}					
\newtheorem{proposition}{Proposition}
\newtheorem{example}{Example}
\newtheorem{lemma}{Lemma}
\newtheorem{remark}{Remark}
\newtheorem{definition}{Definition}

\newtheorem{corollary}{Corollary}

\newcommand{\smalllb}{\\[-0.25cm]}
\let\exampleOrig\endexample
\def\endexample{\hspace*{0pt}\hfill$\triangleleft$\smalllb\exampleOrig}
\let\remarkOrig\endremark
\def\endremark{\hspace*{0pt}\hfill$\triangleleft$\smalllb\remarkOrig}

\usepackage{pgf}
\usepackage{tikz}
\usetikzlibrary{trees,decorations,arrows,automata,shadows,positioning,plotmarks,backgrounds,shapes}
\usetikzlibrary{calc,matrix,fit,petri,decorations.markings,decorations.pathmorphing,patterns,intersections,decorations.text}

 \pgfdeclarelayer{foreground}
 \pgfdeclarelayer{background}
 \pgfsetlayers{background,main,foreground}
\usetikzlibrary{decorations.pathreplacing}
 
\newcommand{\REFlem}[1]{\text{Lem.~\ref{#1}}}

\newcommand{\REFthm}[1]{\text{Thm.~\ref{#1}}}

\newcommand{\REFdef}[1]{Def.~\ref{#1}}

\newcommand{\REFrem}[1]{Rem.~\ref{#1}}

\newcommand{\REFsec}[1]{Sec.~\ref{#1}}
\newcommand{\REFfig}[1]{Fig.~\ref{#1}}
\newcommand{\REFprop}[1]{Prop.~\ref{#1}}

\newcommand{\deff}{:=}
\newcommand{\BR}[1]{\left( #1 \right)}

\newcommand{\ON}[1]{\operatorname{#1}}

\def\clap#1{\hbox to 0pt{\hss#1\hss}}

\newcommand{\val}[1]{\ensuremath{\mathsf{#1}}}

\newcommand{\DiCases}[4]{\ensuremath{\begin{cases}%
#1&,~#2\\%
#3&,~#4%
\end{cases}}}%
\newcommand{\TriCases}[6]{\ensuremath{\begin{cases}%
#1&,~#2\\%
#3&,~#4\\%
#5&,~#6%
\end{cases}}}%
\makeatletter 
\newif\ifFIRST
\newif\ifSECOND
\let\LISTOP\relax
\newcommand{\List}[4][\;]{#3#1%
	\FIRSTtrue
	\@for\i:=#2\do{%
	\ifFIRST\LISTOP{\i}\FIRSTfalse\else,\LISTOP{\i}\fi%
	}%
	#1#4%
	\let\LISTOP\relax
}
\makeatother
\newcounter{DINGLIST}
\stepcounter{DINGLIST}
\newcommand{\markD}[3][\;\;]{\text{\ding{\the\numexpr171+\theDINGLIST}\stepcounter{DINGLIST}}#1#3}

\newcommand{\ZZ}{\textsl{Zu Zeigen:}~\@ifstar\ZZStar\ZZNoStar}
\newcommand{\ZZStar}[1]{\begin{align*}#1\end{align*}}
\newcommand{\ZZNoStar}[1]{\ensuremath{#1}}
\newcommand{\THATIS}{i.e.\xspace}
\newcommand{\SUCHTHAT}{s.t.\xspace}

\newcommand{\SHOW}[2][]{Show \ifthenelse{\isempty{#1}}{}{#1, \THATIS, }\ensuremath{#2}:}

\newcommand{\SORRY}[1]{\emph{\color{red}Sorry: #1}\@latex@warning{SORRY: #1}}

\makeatletter

\newcommand{\propNeg}{\@ifstar\propNegStar\propNegNoStar}
\newcommand{\propNegStar}[1]{\ensuremath{\left(\propNegNoStar{#1}\right)}}
\newcommand{\propNegNoStar}[2][\cdot]{\ensuremath{\neg\ifthenelse{\isempty{#2}}{#1}{#2}}}

\newcommand{\propConj}{\@ifstar\propConjStar\propConjNoStar}
\newcommand{\propConjStar}[2]{\ensuremath{\left(\propConjNoStar{#1}{#2}\right)}}
\newcommand{\propConjNoStar}[3][\cdot]{\ensuremath{\ifthenelse{\isempty{#2}}{#1}{#2}\wedge\ifthenelse{\isempty{#3}}{#1}{#3}}}

\newcommand{\propDisj}{\@ifstar\propDisjStar\propDisjNoStar}
\newcommand{\propDisjStar}[2]{\ensuremath{\left(\propDisjNoStar{#1}{#2}\right)}}
\newcommand{\propDisjNoStar}[3][\cdot]{\ensuremath{\ifthenelse{\isempty{#2}}{#1}{#2}\vee\ifthenelse{\isempty{#3}}{#1}{#3}}}

\newcommand{\propImp}{\@ifstar\propImpStar\propImpNoStar}
\newcommand{\propImpStar}[2]{\ensuremath{\left(\propImpNoStar{#1}{#2}\right)}}
\newcommand{\propImpNoStar}[3][\cdot]{\ensuremath{\ifthenelse{\isempty{#2}}{#1}{#2}\Rightarrow\ifthenelse{\isempty{#3}}{#1}{#3}}}

\newcommand{\propAequ}{\@ifstar\propAequStar\propAequNoStar}
\newcommand{\propAequStar}[2]{\ensuremath{\left(\propAequNoStar{#1}{#2}\right)}}
\newcommand{\propAequNoStar}[3][\cdot]{\ensuremath{\ifthenelse{\isempty{#2}}{#1}{#2}\Leftrightarrow\ifthenelse{\isempty{#3}}{#1}{#3}}}

\newcommand{\propXOR}{\@ifstar\propXORStar\propXORNoStar}
\newcommand{\propXORStar}[2]{\ensuremath{\left(\propXORNoStar{#1}{#2}\right)}}
\newcommand{\propXORNoStar}[3][\cdot]{\ensuremath{\ifthenelse{\isempty{#2}}{#1}{#2}\oplus\ifthenelse{\isempty{#3}}{#1}{#3}}}

\newcommand{\AllQ}{\@ifstar\AllQStar\AllQNoStar}
\newcommand{\AllQStar}[3][\;]{\ensuremath{\left(\forall #2#1.#1#3\right)}}
\newcommand{\AllQNoStar}[3][\;]{\ensuremath{\forall #2#1.#1#3}}
\newcommand{\AllQu}{\@ifstar\AllQuStar\AllQuNoStar}
\newcommand{\AllQuStar}[3][\;]{\ensuremath{\left(\forall^{\infty} #2#1.#1#3\right)}}
\newcommand{\AllQuNoStar}[3][\;]{\ensuremath{\forall^{\infty} #2#1.#1#3}}

\newcommand{\ExQ}{\@ifstar\ExQStar\ExQNoStar}
\newcommand{\ExQStar}[3][\;]{\ensuremath{\left(\exists #2#1.#1#3\right)}}
\newcommand{\ExQNoStar}[3][\;]{\ensuremath{\exists #2#1.#1#3}}

\newcommand{\NExQ}{\@ifstar\NExQStar\NExQNoStar}
\newcommand{\NExQStar}[3][\;]{\ensuremath{\left(\nexists #2#1.#1#3\right)}}
\newcommand{\NExQNoStar}[3][\;]{\ensuremath{\nexists #2#1.#1#3}}

\newcommand{\UniqueQ}{\@ifstar\UniqueQStar\UniqueQNoStar}
\newcommand{\UniqueQStar}[3][\;]{\ensuremath{\left(\exists! #2#1.#1#3\right)}}
\newcommand{\UniqueQNoStar}[3][\;]{\ensuremath{\exists! #2#1.#1#3}}

\newcommand{\Set}[2][]{\List[#1]{#2}{\{}{\}}}
\newcommand{\VSet}[2][]{\let\LISTOP\val\List[#1]{#2}{\{}{\}}}

\newcommand{\Tuple}[2][]{\List[#1]{#2}{(}{)}}

\newcommand{\VTuple}[2][]{\let\LISTOP\val\List[#1]{#2}{(}{)}}

\newcommand{\StateLabelYt}[2]{\scriptstyle\langle\kern-1pt{#1}\kern-1pt{#2}\kern-1pt\rangle} 
\newcommand{\StateLabelYo}[1]{\ye_{#1}} 
\newcommand{\StateLabelYtb}[2]{\scriptstyle\{\kern-1pt\langle\kern-1pt{#1}\kern-1pt{#2}\kern-1pt\rangle\kern-1pt\}} 
\newcommand{\StateLabelYob}[1]{\scriptstyle\{\kern-1pt{\ye_{#1}}\kern-1pt\}} 
\newcommand{\EdgeLabelYt}[2]{\scriptstyle(\kern-1pt{#1}\kern-1pt{,}\kern-1pt{#2}\kern-1pt)} 
\newcommand{\SetComp}[3][]{\{#1#2#1\mid#1#3#1\}}
\newcommand{\SetCompX}[3][]{\left\{#1#2#1\middle\vert#1#3#1\right\}}

\newcommand{\POWERSET}{\@ifstar\POWERSETStar\POWERSETNoStar}
\newcommand{\POWERSETStar}[1]{\ensuremath{\ON{2}^{\ifthenelse{\isempty{#1}}{\cdot}{#1}}}}
\newcommand{\POWERSETNoStar}[1]{\ensuremath{\ON{2}^{\ifthenelse{\isempty{#1}}{\cdot}{#1}}}}

\newcommand{\FINPOWERSET}{\@ifstar\FINPOWERSETStar\FINPOWERSETNoStar}
\newcommand{\FINPOWERSETStar}[1]{\ensuremath{\mathcal{P}_{\ON{fin}}(\ifthenelse{\isempty{#1}}{\cdot}{#1})}}
\newcommand{\FINPOWERSETNoStar}[1]{\ensuremath{\mathcal{P}_{\ON{fin}}\left(\ifthenelse{\isempty{#1}}{\cdot}{#1}\right)}}

\newcommand{\UNION}{\@ifstar\UNIONStar\UNIONNoStar}
\newcommand{\UNIONStar}[2]{\ensuremath{\left(\UNIONNoStar{#1}{#2}\right)}}
\newcommand{\UNIONNoStar}[2]{\ensuremath{\ifthenelse{\isempty{#1}}{\cdot}{#1}\cup\ifthenelse{\isempty{#2}}{\cdot}{#2}}}

\newcommand{\UNIOND}{\@ifstar\UNIONDStar\UNIONDNoStar}
\newcommand{\UNIONDStar}[2]{\ensuremath{\left(\UNIONDNoStar{#1}{#2}\right)}}
\newcommand{\UNIONDNoStar}[2]{\ensuremath{\ifthenelse{\isempty{#1}}{\cdot}{#1}\uplus\ifthenelse{\isempty{#2}}{\cdot}{#2}}}

\newcommand{\SETMINUS}{\@ifstar\SETMINUSStar\SETMINUSNoStar}
\newcommand{\SETMINUSStar}[2]{\ensuremath{\left(\SETMINUSNoStar{#1}{#2}\right)}}
\newcommand{\SETMINUSNoStar}[2]{\ensuremath{\ifthenelse{\isempty{#1}}{\cdot}{#1}\setminus\ifthenelse{\isempty{#2}}{\cdot}{#2}}}

\newcommand{\INTERSECT}{\@ifstar\INTERSECTStar\INTERSECTNoStar}
\newcommand{\INTERSECTStar}[2]{\ensuremath{\left(\INTERSECTNoStar{#1}{#2}\right)}}
\newcommand{\INTERSECTNoStar}[2]{\ensuremath{\ifthenelse{\isempty{#1}}{\cdot}{#1}\cap\ifthenelse{\isempty{#2}}{\cdot}{#2}}}

\newcommand{\CARTPROD}{\@ifstar\CARTPRODStar\CARTPRODNoStar}
\newcommand{\CARTPRODStar}[2]{\ensuremath{\left(\CARTPRODNoStar{#1}{#2}\right)}}
\newcommand{\CARTPRODNoStar}[2]{\ensuremath{\ifthenelse{\isempty{#1}}{\cdot}{#1}\times\ifthenelse{\isempty{#2}}{\cdot}{#2}}}

\newcommand{\FINCOUNT}{\@ifstar\FinCountStar\FinCountNoStar}
\newcommand{\FinCountStar}[1]{\ensuremath{\#(\ifthenelse{\isempty{#1}}{\cdot}{#1})}}
\newcommand{\FinCountNoStar}[1]{\ensuremath{\#\left(\ifthenelse{\isempty{#1}}{\cdot}{#1}\right)}}

\makeatother 

\newcommand{\sconc}{\cdot}
\newcommand{\sconcps}{\hspace{-0.1cm}\cdot\hspace{-0.1cm}}

\newcommand{\fun}{\ensuremath{\ON{\rightarrow}}}

\tikzstyle{Sstate}=[state,inner sep=1pt,minimum size=0pt]
\tikzstyle{istate}=[state,initial,inner sep=1pt,minimum size=0pt,initial text=]
\tikzstyle{fistate}=[state,accepting,initial,initial text=]
\tikzstyle{fistateA}=[state,accepting,initial,initial text=,initial where=above]
\tikzstyle{fistateB}=[state,accepting,initial,initial text=,initial where=below]
\tikzstyle{fistateL}=[state,accepting,initial,initial text=,initial where=left]
\tikzstyle{fistateR}=[state,accepting,initial,initial text=,initial where=right]
\tikzstyle{ifstate}=[state,accepting,initial,initial text=]
\tikzstyle{ifstateA}=[state,accepting,initial,initial text=,initial where=above]
\tikzstyle{ifstateB}=[state,accepting,initial,initial text=,initial where=below]
\tikzstyle{ifstateL}=[state,accepting,initial,initial text=,initial where=left]
\tikzstyle{ifstateR}=[state,accepting,initial,initial text=,initial where=right]
\tikzstyle{istateA}=[state,initial,initial text=,initial where=above]
\tikzstyle{istateB}=[state,initial,initial text=,initial where=below]
\tikzstyle{istateL}=[state,initial,initial text=,initial where=left]
\tikzstyle{istateR}=[state,initial,initial text=,initial where=right]
\tikzstyle{fstate}=[state,inner sep=1pt,minimum size=0pt,accepting]
\tikzstyle{SFSautomat}=[->,>=stealth',shorten >=1pt,auto,node distance=2cm,on grid,semithick,inner sep=1pt,bend angle=45]
\newcommand{\SFSAutomatEdge}[5]{\draw[->](#1) to[#4] node[#5] {\ensuremath{#2}} (#3);}

\newenvironment{propConjA}{\left(\def\unionAtest{1}\begin{array}{@{\if\unionAtest1\gdef\unionAtest{0}\phantom{\wedge}\else\wedge\fi}l@{}}}{\end{array}\right)}

\newenvironment{propDisjA}{\left(\def\unionAtest{1}\begin{array}{@{\if\unionAtest1\gdef\unionAtest{0}\phantom{\vee}\else\vee\fi}l@{}}}{\end{array}\right)}

\makeatletter
  \newlength{\SFS@HEIGHT}
  \newlength{\SFS@WIDTH}
  \newcommand{\SplitX}[2]{
	    \settoheight{\SFS@HEIGHT}{$#2$}
	    \settowidth{\SFS@WIDTH}{$#2$}
	    \mbox{\begin{tikzpicture}[baseline=(current bounding box.center)]
	    \node[] (E) at (0,0) {$#1$};
	    \node[inner sep=0pt] (F) at ($(E.south west)+(1ex,-1ex)+(3ex+.5\SFS@WIDTH,-\SFS@HEIGHT)$) {$#2$};
	    \node[] (E) at (0,0) {\phantom{$#1$}};
	    \draw[fill] ($(E.east)+(0ex,0ex)$) circle (.2ex);
	    \draw[-] ($(E.east)+(0ex,0ex)$) -- ($(E.south east)+(0ex,-0.5ex)$) -- ($(E.south west)+(1ex,-0.5ex)$) -- ($(E.south west)+(1ex,-1ex)-(0,\SFS@HEIGHT)$) -- ($(E.south west)+(2.5ex,-1ex)-(0,\SFS@HEIGHT)$);
	    \draw[fill] ($(E.south west)+(2.5ex,-1ex)-(0,\SFS@HEIGHT)$) circle (.2ex);
	    \end{tikzpicture}}}
  \newcommand{\SplitS}[2]{
	    \settoheight{\SFS@HEIGHT}{$#2$}
	    \settowidth{\SFS@WIDTH}{$#2$}
	    \mbox{\begin{tikzpicture}[baseline=(current bounding box.center)]
	    \node[] (E) at (0,0) {$#1$};
	    \node[inner sep=0pt] (F) at ($(E.south west)+(1ex,0.5ex)+(3ex+.5\SFS@WIDTH,-\SFS@HEIGHT)$) {$#2$};
	    \end{tikzpicture}}}	    
  \newcommand{\AllQSplit}[2]{\SplitX{\forall\;#1\;.}{#2}}

  \newcommand{\ExQSplit}[2]{\SplitX{\exists\;#1\;.}{#2}}

  \newcommand{\propImpSplit}[2]{\SplitX{#1\;\Rightarrow\;}{#2}}

\makeatother

\providecommand{\length}[1]{\lvert#1\rvert}

\newcommand{\trivialN}[1]{\text{trivial}\xspace}

\newcommand{\Nb}{\ensuremath{\mathbb{N}}} 
\newcommand{\Zb}{\ensuremath{\mathbb{Z}}}
\newcommand{\Nbn}{\ensuremath{\mathbb{N}_{0}}}

\newcommand{\twoup}[1]{\ensuremath{2^{#1}}} 

\newcommand{\extern}[1]{#1}
\newcommand{\abst}[1]{\widehat{#1}}

\newcommand{\lsup}{l}
\newcommand{\lsupm}{\Ilm}
\newcommand{\lsupp}{l^+}
\newcommand{\qsup}{\ensuremath{\triangledown}}
\newcommand{\qsupb}{\ensuremath{\triangledown}}

\newcommand{\Sys}{\ensuremath{\mathcal{S}}} 
\newcommand{\Qsys}{\ensuremath{\mathcal{Q}}} 
\newcommand{\Qsyse}{\ensuremath{\extern{\mathcal{Q}}}}
\newcommand{\Qsysa}{\ensuremath{\abst{\mathcal{Q}}}}
\newcommand{\Qsysal}{\ensuremath{\abst{\mathcal{Q}}^{\lsupm}}}

\newcommand{\Qsysaq}{\ensuremath{\abst{\mathcal{Q}}^\qsup}}
\newcommand{\Qsysaql}{\ensuremath{\abst{\mathcal{Q}}^{l\qsup}}}
\newcommand{\Qsysaqbl}{\ensuremath{\abst{\mathcal{Q}}^{l\qsupb}}}

\newcommand{\Xt}[2]{\ensuremath{\ifthenelse{\isempty{#2}}{X_{\T,#1}}{X_{#2,{\T_{#2}},#1}}}} 
\newcommand{\XT}[1]{\ensuremath{\ifthenelse{\isempty{#1}}{X_\T}{X_{#1,{\T_{#1}}}}}} 
\newcommand{\Xk}[2]{\ensuremath{\ifthenelse{\isempty{#2}}{X_{\TE,#1}}{X_{#2,\TE,#1}}}} 
\newcommand{\XK}[1]{\ensuremath{\ifthenelse{\isempty{#1}}{X_{\TE}}{X_{#1,\TE}}}} 
\newcommand{\Zk}[2]{\ensuremath{\ifthenelse{\isempty{#2}}{Z_{\TE,#1}}{Z_{#2,\TE,#1}}}}

\newcommand{\Zt}[1]{\ensuremath{\ifthenelse{\isempty{#1}}{Z_t}{Z_{#1,t}}}} 
\newcommand{\ZT}[1]{\ensuremath{\ifthenelse{\isempty{#1}}{Z_T}{Z_{#1,T}}}} 
 
\newcommand{\ZPit}[1]{\ensuremath{\ifthenelse{\isempty{#1}}{Z_t}{Z_{#1,t}}}} 
\newcommand{\ZPiT}[1]{\ensuremath{\ifthenelse{\isempty{#1}}{Z_T}{Z_{#1,T}}}} 
 
\newcommand{\ZTit}[1]{\ensuremath{\ifthenelse{\isempty{#1}}{\breve{Z}_t}{\breve{Z}_{#1,t}}}} 
\newcommand{\ZTiT}[1]{\ensuremath{\ifthenelse{\isempty{#1}}{\breve{Z}_T}{\breve{Z}_{#1,T}}}} 

\renewcommand{\v}{\ensuremath{v}} 
\renewcommand{\d}{\ensuremath{d}}

\newcommand{\tr}{\ensuremath{\delta}}

\newcommand{\T}{\ensuremath{T}}

\newcommand{\Interval}{\ensuremath{\mathcal{I}}} 
\newcommand{\Ilm}{\ensuremath{{\Interval_m^l}}}
\newcommand{\Ill}{\ensuremath{{\Interval_l^l}}}
\newcommand{\Ilmp}{\ensuremath{{\Interval_{m\plps 1}^l}}}

\newcommand{\Ilpm}{\ensuremath{{\Interval_{m}^{l\plps 1}}}}

\renewcommand{\ll}[1]{\ensuremath{|_{[#1]}}} 
 
\newcommand{\llr}[1]{\ensuremath{|_{[#1)}}}

\newcommand{\f}[1]{\ensuremath{f\ifthenelse{\isempty{#1}}{}{\Tuple{#1}}}} 
\newcommand{\g}[1]{\ensuremath{g\ifthenelse{\isempty{#1}}{}{\Tuple{#1}}}} 
\newcommand{\h}[1]{\ensuremath{h\ifthenelse{\isempty{#1}}{}{\Tuple{#1}}}} 
\newcommand{\fs}[2]{\ensuremath{f_{#2}\ifthenelse{\isempty{#1}}{}{\Tuple{#1}}}} 
\newcommand{\gs}[2]{\ensuremath{g_{#2}\ifthenelse{\isempty{#1}}{}{\Tuple{#1}}}}  
\newcommand{\hs}[2]{\ensuremath{h_{#2}\ifthenelse{\isempty{#1}}{}{\Tuple{#1}}}}  
\newcommand{\fe}[2]{\ensuremath{\extern{f}_{#1}\ifthenelse{\isempty{#2}}{}{\Tuple{#2}}}} 
\newcommand{\he}[2]{\ensuremath{\extern{h}_{#1}\ifthenelse{\isempty{#2}}{}{\Tuple{#2}}}} 
\newcommand{\fa}[2]{\ensuremath{\abst{f}_{#1}\ifthenelse{\isempty{#2}}{}{\Tuple{#2}}}} 
\newcommand{\ha}[2]{\ensuremath{\abst{h}_{#1}\ifthenelse{\isempty{#2}}{}{\Tuple{#2}}}} 
\newcommand{\fal}[2]{\ensuremath{\abst{f}_{#1}^{\lsupm}\ifthenelse{\isempty{#2}}{}{\Tuple{#2}}}} 
\newcommand{\hal}[2]{\ensuremath{\abst{h}_{#1}^{\lsupm}\ifthenelse{\isempty{#2}}{}{\Tuple{#2}}}} 
\newcommand{\falp}[2]{\ensuremath{\abst{f}_{#1}^{\lsupp}\ifthenelse{\isempty{#2}}{}{\Tuple{#2}}}} 
\newcommand{\halp}[2]{\ensuremath{\abst{h}_{#1}^{\lsupp}\ifthenelse{\isempty{#2}}{}{\Tuple{#2}}}} 
\newcommand{\faq}[2]{\ensuremath{\abst{f}_{#1}^{\qsup}\ifthenelse{\isempty{#2}}{}{\Tuple{#2}}}} 
\newcommand{\haq}[2]{\ensuremath{\abst{h}_{#1}^{\qsup}\ifthenelse{\isempty{#2}}{}{\Tuple{#2}}}}

\newcommand{\lnc}{\ensuremath{l_t}}

\newcommand{\Beh}{\ensuremath{\mathcal{B}}}
\newcommand{\Behf}{\ensuremath{\mathcal{B}_f}}
\newcommand{\BeheQ}{\ensuremath{{\extern{\mathcal{B}}(\Qsyse)}}}
\newcommand{\BeheSQ}{\ensuremath{{\extern{\mathcal{B}}_S(\Qsyse)}}}
\newcommand{\Behe}{\ensuremath{\extern{\mathcal{B}}}}

\newcommand{\kg}[1]{\ensuremath{\xspace\preceq_{#1}\xspace}}
\newcommand{\hg}[1]{\ensuremath{\xspace\cong_{#1}\xspace}}
\newcommand{\kgl}[2]{\ensuremath{\xspace\preceq_{#1}^{#2}\xspace}}
\newcommand{\hgl}[2]{\ensuremath{\xspace\cong_{#1}^{#2}\xspace}}

\newcommand{\Ds}[2]{\ensuremath{\Pi^{#1}_{#2}(\BeheQ)}}

\newcommand{\ESn}[1]{\ensuremath{\ifthenelse{\isempty{#1}}{\Sigma^+_{S}}{\Sigma^+_{S,#1}}}}

\newcommand{\BehS}[1]{\ensuremath{\ifthenelse{\isempty{#1}}{\Beh_{S}}{\Beh_{S,#1}}}}
\newcommand{\BehSe}[1]{\ensuremath{\ifthenelse{\isempty{#1}}{\extern{\Beh}_{S}}{\extern{\Beh}_{S,#1}}}}
\newcommand{\BehE}[1]{\ensuremath{\ifthenelse{\isempty{#1}}{\Beh_{E}}{\Beh_{E,#1}}}}

\newcommand{\W}{\ensuremath{W}}
\newcommand{\D}{\ensuremath{D}}
\newcommand{\V}{\ensuremath{V}}
\newcommand{\statemap}[3]{\ifthenelse{\isempty{#2#3}}{\psi_{#1}}{\psi_{#1}(#2,#3)}}
\newcommand{\statemapPi}[3]{\ifthenelse{\isempty{#2#3}}{\varphi_{#1}}{\varphi_{#1}(#2,#3)}}
\newcommand{\statemapTi}[3]{\ifthenelse{\isempty{#2#3}}{\psi_{#1}}{\psi_{#1}(#2,#3)}}

\newcommand{\CONCAT}[4]{#1\wedge^{#2}_{#3}#4}

\newcommand{\Xx}[2]{\ensuremath{\chi^{-}_{#1}\ifthenelse{\isempty{#2}}{}{(#2)}}}
\newcommand{\Xxp}[2]{\ensuremath{\chi^{+}_{#1}\ifthenelse{\isempty{#2}}{}{(#2)}}}
\newcommand{\Xxr}[3]{\ensuremath{\chi_{#1}^{#2^-}\ifthenelse{\isempty{#3}}{}{\hspace{-0.1cm}(#3)}}}
\newcommand{\Xxrp}[3]{\ensuremath{\chi_{#1}^{#2^+}\ifthenelse{\isempty{#3}}{}{(#3)}}}

\newcommand{\EnabY}[2]{\ON{H}_{#1}\ifthenelse{\isempty{#2}}{}{(#2)}} 
\newcommand{\EnabYr}[2]{\ON{H}^{-1}_{#1}\ifthenelse{\isempty{#2}}{}{(#2)}} 
\newcommand{\EnabT}[2]{\ON{T}_{#1}\ifthenelse{\isempty{#2}}{}{(#2)}} 
\newcommand{\EnabF}[2]{\ON{F}_{#1}\ifthenelse{\isempty{#2}}{}{(#2)}} 
\newcommand{\EnabTr}[2]{\ON{T}^{-1}_{#1}\ifthenelse{\isempty{#2}}{}{(#2)}} 

\newcommand{\EnabWl}[3]{\ON{E}_{#1}^{#2}\ifthenelse{\isempty{#3}}{}{(#3)}}
\newcommand{\EnabWlr}[3]{\BR{\ON{E}_{#1}^{#2}}^{-1}\hspace{-0.4cm}\ifthenelse{\isempty{#3}}{}{(#3)}}
\newcommand{\EnabWlp}[3]{\ON{E}_{#1}^{#2^+}\ifthenelse{\isempty{#3}}{}{(#3)}}
 \newcommand{\EnabU}[2]{\ON{U}_{#1}(#2)}
 \newcommand{\EnabUY}[2]{\mathcal{A}_{\ueS\times\yeS}^{#1}\ifthenelse{\isempty{#2}}{}{(#2)}}

\newcommand{\R}{\ensuremath{\mathcal{R}}}

\newcommand{\Rn}{\ensuremath{\mathcal{R}^{\lsupm}}}

\newcommand{\Rq}{\ensuremath{\mathcal{R}^{\qsup}}}

\newcommand{\SR}[4]{\ensuremath{\mathfrak{R}_{#1}^{#2}(#3,#4)}}

\newcommand{\lext}[1]{\vphantom{#1}^\diamond\kern-\scriptspace#1}

\newcommand{\signalmap}{\phi}

\newcommand{\TE}{\ensuremath{{T_E}}}

\newcommand{\E}{\ensuremath{\Sigma}}
\newcommand{\Ee}{\ensuremath{\extern{\Sigma}}}
\newcommand{\Ea}{\ensuremath{\abst{\Sigma}}}
\newcommand{\Ep}{\ensuremath{\Sigma^{\signalmap}}}
\newcommand{\Epa}{\ensuremath{\abst{\Sigma}^{\signalmap}}}
\newcommand{\ES}[1]{\ensuremath{\ifthenelse{\isempty{#1}}{\Sigma_{S}}{\Sigma_{S,#1}}}}
\newcommand{\ESe}[1]{\ensuremath{\ifthenelse{\isempty{#1}}{\extern{\Sigma}_{S}}{\extern{\Sigma}_{S,#1}}}}
\newcommand{\EpS}[1]{\ensuremath{\ifthenelse{\isempty{#1}}{\Ep_{S}}{\Sigma^{\signalmap_{#1}}_{S,#1}}}}
\newcommand{\EpSa}[1]{\ensuremath{\ifthenelse{\isempty{#1}}{\Epa_{S}}{\abst{\Sigma}^{\signalmap_{#1}}_{S,#1}}}}
\newcommand{\EE}[1]{\ensuremath{\ifthenelse{\isempty{#1}}{\Sigma_{E}}{\Sigma_{E,#1}}}}
\newcommand{\ESm}[1]{\ensuremath{\ifthenelse{\isempty{#1}}{\Sigma_{\psi}}{\Sigma_{\psi,#1}}}}

\newcommand{\ElaMax}{\ensuremath{\Sigma^{l^\uparrow}}}
\newcommand{\Eal}{\ensuremath{\abst{\Sigma}^{\lsup}}}

\newcommand{\EplMaxS}[1]{\ensuremath{\ifthenelse{\isempty{#1}}{\Sigma^{\phi,l^\uparrow}_{S}}{\Sigma^{\phi_{#1},l^\uparrow}_{S,#1}}}}
\newcommand{\EplMaxSp}[1]{\ensuremath{\ifthenelse{\isempty{#1}}{\overline{\Sigma}^{\phi,l^\uparrow}_{S}}{\overline{\Sigma}^{\phi_{#1},l^\uparrow}_{S,#1}}}}
\newcommand{\EplncMaxS}[1]{\ensuremath{\ifthenelse{\isempty{#1}}{\Sigma^{\phi,\lnc^\uparrow}_{S}}{\Sigma^{\phi_{#1},\lnc^\uparrow}_{S,#1}}}}
\newcommand{\EplsMaxS}[2]{\ensuremath{\ifthenelse{\isempty{#1}}{\Sigma^{\phi,{#2}^\uparrow}_{S}}{\Sigma^{\phi_{#1},{#2}^\uparrow}_{S,#1}}}}

\newcommand{\ElMaxS}[1]{\ensuremath{\ifthenelse{\isempty{#1}}{\Sigma^{l^\Uparrow}_{S}}{\Sigma^{l^\Uparrow}_{S,#1}}}}
\newcommand{\ElaMaxS}[1]{\ensuremath{\ifthenelse{\isempty{#1}}{\Sigma^{l^\uparrow}_{S}}{\Sigma^{l^\uparrow}_{S,#1}}}}
\newcommand{\ElMaxSp}[1]{\ensuremath{\ifthenelse{\isempty{#1}}{\overline{\Sigma}^{l^\uparrow}_{S}}{\overline{\Sigma}^{l^\uparrow}_{S,#1}}}}
\newcommand{\ElncMaxS}[1]{\ensuremath{\ifthenelse{\isempty{#1}}{\Sigma^{\lnc^\uparrow}_{S}}{\Sigma^{\lnc^\uparrow}_{S,#1}}}}
\newcommand{\ElsMaxS}[2]{\ensuremath{\ifthenelse{\isempty{#1}}{\Sigma^{{#2}^\uparrow}_{S}}{\Sigma^{{#2}^\uparrow}_{S,#1}}}}

\newcommand{\ElE}[1]{\ensuremath{\ifthenelse{\isempty{#1}}{\Sigma^{l}_{E}}{\Sigma^{l}_{E,#1}}}}
\newcommand{\ElEp}[1]{\ensuremath{\ifthenelse{\isempty{#1}}{\Sigma^{l}_{E}}{\overline{\Sigma}^{l}_{E,#1}}}}
\newcommand{\ElncE}[1]{\ensuremath{\ifthenelse{\isempty{#1}}{\Sigma^{\lnc}_{E}}{\Sigma^{\lnc}_{E,#1}}}}
\newcommand{\ElMaxE}[1]{\ensuremath{\ifthenelse{\isempty{#1}}{\Sigma^{l^\uparrow}_{E}}{\Sigma^{l^\uparrow}_{E,#1}}}}

\newcommand{\ElMaxEp}[1]{\ensuremath{\ifthenelse{\isempty{#1}}{\overline{\Sigma}^{l^\uparrow}_{E}}{\overline{\Sigma}^{l^\uparrow}_{E,#1}}}}
\newcommand{\ElncMaxE}[1]{\ensuremath{\ifthenelse{\isempty{#1}}{\Sigma^{\lnc^\uparrow}_{E}}{\Sigma^{\lnc^\uparrow}_{E,#1}}}}
\newcommand{\ElsMaxE}[2]{\ensuremath{\ifthenelse{\isempty{#1}}{\Sigma^{{#2}^\uparrow}_{E}}{\Sigma^{{#2}^\uparrow}_{E,#1}}}}

\newcommand{\Behal}{\ensuremath{\abst{\Beh}^{\lsup}}}
\newcommand{\Beha}{\ensuremath{\abst{\Beh}}}

\newcommand{\BehlMaxS}[1]{\ensuremath{\ifthenelse{\isempty{#1}}{\Beh^{l^\Uparrow}_{S}}{\Beh^{l^\Uparrow}_{S,#1}}}}
\newcommand{\BehlaMaxS}[1]{\ensuremath{\ifthenelse{\isempty{#1}}{\Beh^{l^\uparrow}_{S}}{\Beh^{l^\uparrow}_{S,#1}}}}
\newcommand{\BehlMaxSp}[1]{\ensuremath{\ifthenelse{\isempty{#1}}{\overline{\Beh}^{l^\uparrow}_{S}}{\overline{\Beh}^{l^\uparrow}_{S,#1}}}}
\newcommand{\BehlncMaxS}[1]{\ensuremath{\ifthenelse{\isempty{#1}}{\Beh^{\lnc^\uparrow}_{S}}{\Beh^{\lnc^\uparrow}_{S,#1}}}}
\newcommand{\BehlsMaxS}[2]{\ensuremath{\ifthenelse{\isempty{#1}}{\Beh^{{#2}^\uparrow}_{S}}{\Beh^{{#2}^\uparrow}_{S,#1}}}}

\newcommand{\BehlE}[1]{\ensuremath{\ifthenelse{\isempty{#1}}{\Beh^{l}_{E}}{\Beh^{l}_{E,#1}}}}
\newcommand{\BehlncE}[1]{\ensuremath{\ifthenelse{\isempty{#1}}{\Beh^{\lnc}_{E}}{\Beh^{\lnc}_{E,#1}}}}
\newcommand{\BehlMaxE}[1]{\ensuremath{\ifthenelse{\isempty{#1}}{\Beh^{l^\uparrow}_{E}}{\Beh^{l^\uparrow}_{E,#1}}}}
\newcommand{\BehlMaxEp}[1]{\ensuremath{\ifthenelse{\isempty{#1}}{\overline{\Beh}^{l^\uparrow}_{E}}{\overline{\Beh}^{l^\uparrow}_{E,#1}}}}
\newcommand{\BehlncMaxE}[1]{\ensuremath{\ifthenelse{\isempty{#1}}{\Beh^{\lnc^\uparrow}_{E}}{\Beh^{\lnc^\uparrow}_{E,#1}}}}
\newcommand{\BehlsMaxE}[2]{\ensuremath{\ifthenelse{\isempty{#1}}{\Beh^{{#2}^\uparrow}_{E}}{\Beh^{{#2}^\uparrow}_{E,#1}}}}

\newcommand{\BehVlMaxS}[1]{\ensuremath{\ifthenelse{\isempty{#1}}{\projState{\V}{\Beh^{l^\uparrow}_{S}}}{\projState{\V}{\Beh^{l^\uparrow}_{S,#1}}}}}

\newcommand{\BehDlMaxS}[1]{\ensuremath{\ifthenelse{\isempty{#1}}{\projState{\D}{\Beh^{l^\uparrow}_{S}}}{\projState{\D}{\Beh^{l^\uparrow}_{S,#1}}}}}

\newcommand{\EoMaxS}[1]{\ensuremath{\ifthenelse{\isempty{#1}}{\Sigma^{1^\uparrow}_{S}}{\Sigma^{1^\uparrow}_{S,#1}}}}

\newcommand{\projState}[2]{\pi_{#1}(#2)}

\newcommand{\x}{\ensuremath{x}}
\newcommand{\xS}{\ensuremath{X}}
\newcommand{\xSo}[1]{\ensuremath{\xS_{#1 0}}}
\newcommand{\xG}{\ensuremath{\xi}} 
\newcommand{\xe}{\ensuremath{\extern{\x}}}
\newcommand{\xeS}{\ensuremath{\extern{\xS}}}
\newcommand{\xeSo}[1]{\ensuremath{\xeS_{#1 0}}}
\newcommand{\xeG}{\ensuremath{\extern{\xG}}}
\newcommand{\xa}{\ensuremath{\abst{\x}}}
\newcommand{\xaS}{\ensuremath{\abst{\xS}}}
\newcommand{\xaSo}[1]{\ensuremath{\xaS_{#1 0}}}
\newcommand{\xaG}{\ensuremath{\abst{\xG}}}

\newcommand{\xaqlS}{\ensuremath{\abst{\xS}^{l\qsup}}}

\newcommand{\xaqlSo}[1]{\ensuremath{\xaS_{#1 0}^{l\qsup}}}
\newcommand{\xalS}{\ensuremath{\abst{\xS}^{\lsupm}}}

\newcommand{\xalSo}[1]{\ensuremath{\xaS_{#1 0}^{\lsupm}}}

\renewcommand{\u}{\ensuremath{u}}
\newcommand{\uS}{\ensuremath{U}}
\newcommand{\uG}{\ensuremath{\mu}} 
\newcommand{\ue}{\ensuremath{\extern{\u}}}
\newcommand{\ueS}{\ensuremath{\extern{\uS}}}
\newcommand{\ueG}{\ensuremath{\extern{\uG}}}

\newcommand{\uaG}{\ensuremath{\abst{\uG}}}
\newcommand{\uQ}{\ensuremath{ \mathfrak{r}}}

\renewcommand{\v}{\ensuremath{v}}

\newcommand{\y}{\ensuremath{y}}
\newcommand{\yS}{\ensuremath{Y}}
\newcommand{\yG}{\ensuremath{\nu}} 
\newcommand{\ye}{\ensuremath{\extern{\y}}}
\newcommand{\yeS}{\ensuremath{\extern{\yS}}}
\newcommand{\yeG}{\ensuremath{\extern{\yG}}}
\newcommand{\ya}{\ensuremath{\abst{\y}}}
\newcommand{\yaS}{\ensuremath{\abst{\yS}}}
\newcommand{\yaG}{\ensuremath{\abst{\yG}}}
\newcommand{\yQ}{\ensuremath{ \mathfrak{d}}}
\newcommand{\yqlS}{\ensuremath{\yaS^{l}}}

\newcommand{\w}{\ensuremath{w}}
\newcommand{\wS}{\ensuremath{W}}
\newcommand{\wG}{\ensuremath{\omega}} 
\newcommand{\we}{\ensuremath{\extern{\w}}}
\newcommand{\weS}{\ensuremath{\extern{\wS}}}
\newcommand{\weG}{\ensuremath{\extern{\wG}}}
\newcommand{\waG}{\ensuremath{\abst{\wG}}}

\newcommand{\z}{\ensuremath{z}}
\newcommand{\zS}{\ensuremath{Z}}
\newcommand{\zG}{\ensuremath{\zeta}} 

\newcommand{\zeS}{\ensuremath{\extern{\zS}}}
\newcommand{\zeG}{\ensuremath{\extern{\zG}}}

\newcommand{\tre}{\ensuremath{\extern{\delta}}}
\newcommand{\tra}{\ensuremath{\abst{\delta}}}

\newcommand{\traql}{\ensuremath{\abst{\delta}^{l\qsup}}}
\newcommand{\traqbl}{\ensuremath{\abst{\delta}^{l\qsupb}}}

\newcommand{\tral}{\ensuremath{\abst{\delta}^{\lsupm}}}

\newcommand{\QBA}{QBA\xspace}
\newcommand{\SlA}{S$l$CA\xspace}
\newcommand{\SAlA}{SA$l$CA\xspace}

\newcommand{\timesps}{\hspace{-0.1cm}\times\hspace{-0.1cm}}
\newcommand{\inps}{\hspace{-0.1cm}\in\hspace{-0.1cm}}
\newcommand{\eqps}{\hspace{-0.1cm}=\hspace{-0.1cm}}
\newcommand{\plps}{\hspace{-0.05cm}+\hspace{-0.05cm}}
\newcommand{\mips}{\hspace{-0.05cm}-\hspace{-0.05cm}}

\renewcommand{\b}{\textcolor{blue}}

\begin{document}
\title{Comparing {Asynchronous $l$-Complete Approximations} and {Quotient Based Abstractions}}

\author{Anne-Kathrin Schmuck, Paulo Tabuada, J\"org Raisch
\thanks{A.-K. Schmuck is with the Max Planck Institute for Software Systems, Kaiserslautern, Germany. {\tt\small akschmuck@mpi-sws.org}}
\thanks{P. Tabuada is with the UCLA Electrical Engineering Department, Los Angeles, USA.
{\tt\small tabuada@ee.ucla.edu}}
\thanks{J. Raisch is with the Control Systems Group, Technical University of Berlin, Germany and the Max Planck Institute
for Dynamics of Complex Technical Systems, Magdeburg, Germany. {\tt\small raisch@control.tu-berlin.de}}
}

\maketitle

\begin{abstract}
This paper is concerned with a detailed comparison of two different abstraction techniques for the construction of finite state symbolic models for controller synthesis of hybrid systems. Namely, we compare quotient based abstractions (\QBA), e.g., described in \cite[Part~II]{TabuadaBook} with different realizations of strongest (asynchronous) $l$-complete approximations (\SAlA) from \cite{MoorRaisch1999,SchmuckRaisch2014_ControlLetters}.
Even though the idea behind their construction is very similar, we show that they are generally incomparable both in terms of behavioral inclusion and similarity relations.
We therefore derive necessary and sufficient conditions 
for \QBA to coincide with particular realizations of \SAlA.
Depending on the original system, either \QBA or \SAlA can be a tighter abstraction.
\end{abstract}


\section{Introduction}

The increasing interconnection of physical components and digital hardware in today's engineering systems causes challenges that have been investigated by both the control and the computer science community.
Although some efforts have been made to bring these parallel advances together, there are still considerable gaps between concepts in both fields addressing very similar questions. 
In this paper, we provide a step towards connecting two methods for the construction of finite state symbolic abstractions
inspired by these two communities.

Systems where digital hardware is connected to physical components usually lead to \emph{hybrid system models}. Control synthesis for hybrid systems is a difficult problem, and one common approach to this problem is, first, to simplify a given hybrid control problem by generating a symbolic abstraction of the system to be controlled and, second, to design a symbolic controller using existing synthesis techniques. 
%
%
This controller synthesis approach is usually used in two different settings.

In the first setting a system should obey a specification given in terms of a linear temporal logic (LTL) or computational tree logic (CTL) formula over a finite set of symbols, e.g., \enquote{always eventually visit region A}, which can only be enforced by symbolic controller synthesis techniques.
Inspired by the computer science community, this line of research applies techniques developed for verification and synthesis of software processes, as e.g. in \cite{AlurHenzingerLaffarrierePappas2000,TabuadaPappas2003b,Tabuada2007} and summarized in \cite[Part~II]{TabuadaBook}. In that work a symbolic abstraction is constructed by 
partitioning the original state space into a finite number of cells, such that this partition allows for a bisimulation relation between the original state space model and its abstraction. The set of equivalence classes of this partition is used to define the outputs as well as the states of the constructed abstraction. 
This abstraction method is often referred to as \emph{quotient based abstraction} (\QBA), a terminology we adopt in this paper. 

Contrary to this viewpoint, another class of abstractions is tailored to handle systems where the available interface for control is symbolic. Hence, the construction of a symbolic abstraction is motivated by limited sensing (e.g., a sensor that can only detect threshold crossings) and/or limited actuation (e.g., a valve that can only be fully opened or closed). 
This implies that the set of input and output symbols is predefined and cannot be used to adjust the abstraction accuracy. The \emph{Strongest $l$-complete approximation} (\SlA) \cite{MoorRaisch1999} is one concept explicitly addressing this issue, which was recently generalized to the strongest \emph{asynchronous} $l$-complete approximation (\SAlA) \cite{SchmuckRaisch2014_ControlLetters}. Here, the  accuracy of the abstraction is adjusted by changing the number $l$ of past input and output symbols considered in the construction of the abstract state space.

The idea of using $l$-long strings of symbols as abstract states was recently revisited in \cite{LeCorroncGirardGoessler2013,ZamaniTkachevAbate2014,TarrafEspinosa2011}. 
Interestingly, the abstractions in \cite{LeCorroncGirardGoessler2013,ZamaniTkachevAbate2014} are based on (approximated versions of) \QBA but employ ideas from \SlA without assuming a symbolic controller interface. In  \cite{LeCorroncGirardGoessler2013} and \cite{ZamaniTkachevAbate2014} $l$-long sequences of modes of, respectively, incrementally stable switched systems and stochastic systems are used as abstract states rather than input and output symbols. 

In this paper we formally compare \QBA and \SAlA to point out their conceptual differences which are mostly due to the different scenarios they are tailored to. This, of course, also has an influence on the construction of symbolic controllers based on those abstractions. While we do not provide a formal comparison of the controller synthesis step, an insightful discussion of this step in both scenarios is given in \REFsec{sec:Compare:Control}. 

Apart from this additional discussion, this paper furthermore extends the results in \cite{SchmuckTabuadaRaisch_2015} by providing proofs for all results and several detailed examples illustrating the paper's contents. 

\section{Preliminaries}\label{sec:prelim}
In this section, we first review necessary notation from behavioral systems theory (e.g., \cite{Willems1991}) in \REFsec{sec:notation} and derive a model of the original system in \REFsec{sec:Qsyse}. To compare the resulting \QBA and \SAlA of this system we introduce the notion of simulation relations in \REFsec{sec:2:SimulationRelations}.

\subsection{Notation}\label{sec:notation}

In the behavioral framework, a \textit{dynamical system} is given by ${\E=\Tuple{\T,\wS,\Beh}}$, consisting of the time axis ${\T}$, the signal space $\wS$, and the behavior of the system, ${\Beh\subseteq\BR{\wS}^\T}$, where
$\BR{\wS}^\T\deff\SetComp{\wG}{\wG:\T\fun\wS}$ is the set of all 
signals evolving on $\T$ and
taking values in $\wS$.
In this paper we only consider dynamical systems evolving on the discrete time axis $\T=\Nbn$. However, to simplify notation, we extend the time axis of a behavior $\Beh\subseteq \BR{\wS}^{\Nbn}$ from $\Nbn$ to $\Zb$ by pre-appending each $\wG\in\Beh$ with the special symbol $\diamond$, i.e., $\wG=\w_0\w_1\w_2\hdots\in\Beh$ is transformed to $\hdots\diamond\diamond\diamond\w_0\w_1\w_2\hdots\subseteq\UNION*{\wS}{\Set{\diamond}}^\Zb$. 
Hence, the notation ${\E=\Tuple{\Nbn,\wS,\Beh}}$ refers to a system with behavior $\Beh\subseteq\UNION*{\wS}{\Set{\diamond}}^\Zb$ s.t.\footnote{Throughout this paper we use the notation "$\AllQ{}{}$", meaning that all statements after the dot hold for all variables quantified before the dot. "$\ExQ{}{}$" is interpreted analogously. } $\AllQ{\weG\in\Beh,k<0}{\weG(k)=\diamond}$.\\
%
%
For any $l\in\Nbn$, $\BR{\W}^l\deff\SetComp{\omega}{\omega:[0,l-1]\fun\W}$ denotes the set of strings $\wG$ with length $l$
 and elements in $\W$. Now let $\Interval=[t_1,t_2]$ be a bounded interval on $\Zb$ with length $\length{\Interval}=t_2-t_1+1$. Then
 $\wG|_{\Interval}=\wG(t_1)\hdots\wG(t_2)\in\BR{\wS}^{\length{\Interval}}$ is the result of  \textit{restricting} the map $\wG:\Zb\fun\wS$ to the domain $\Interval$ and disregarding absolute time information, i.e., $\wG|_{\Interval}\in\wS^{\length{\Interval}}$ instead of $\wG|_{\Interval}\in\wS^\Interval$. Similarly, $\Beh|_{\Interval}$ results from restricting all trajectories in $\Beh$ to $\Interval$ and disregarding absolute time information. 
For $t_1< t_2$ we define $\wG\ll{t_2,t_1}:=\lambda$, where $\lambda$ denotes the \textit{empty string}. \\
Now let $\wS,V$ and $\tilde{V}$ be sets. Then the \textit{projection} of the set $\wS$ and the symbol $\w\in\wS$ to $ V$ is defined by
\begin{equation*}
 \projState{ V}{\wS}\deff\TriCases{ V\hspace{-0.2cm}}{\wS\eqps V\timesps\tilde{ V}}{\wS\hspace{-0.2cm}}{\wS\eqps V}{\emptyset\hspace{-0.2cm}}{\text{else}}\quad
 \projState{ V}{\w}\deff\TriCases{\v\hspace{-0.2cm}}{\w\eqps\Tuple{\v,\tilde{\v}}
 }{\w\hspace{-0.2cm}}{\wS\eqps V}{\lambda\hspace{-0.2cm}}{\text{else},}
\end{equation*}
respectively. With this, the projection of a signal $\wG\in\wS^{\T}$ to $ V$ is given by
 $\projState{ V}{\wG}\deff\SetComp{ v\in V^{\T}}{\AllQ{t\in\T}{ v(t)=\projState{ V}{\wG(t)}}}$
and $\projState{ V}{\Beh}$ denotes the projection of all signals in the behavior $\Beh$ to $ V$. 
%
The \textit{concatenation} of two strings $\wG_1\in\BR{\wS}^{t_1},~\wG_2\in\BR{\wS}^{t_2}, t_1,t_2\in\Nbn$ is denoted by ${\wG_1\sconc \wG_2}$ (meaning that $\wG_2$ is appended to $\wG_1$).
\subsection{Modelling the Original System}\label{sec:Qsyse}

The common starting point of methods generating finite state abstractions of a (possibly continuous) dynamical system is the definition of a finite external signal space $\weS$. 
In the context of \SAlA, $\weS=\ueS\times\yeS$ is assumed to be predefined by the system to be abstracted, where $\ueS$ is a finite set of control symbols and $\yeS$ a finite set of measurement symbols.
In contrast, the work on \QBA usually assumes full sensing and actuating capabilities but defines the finite output set $\yeS$ based on a specification that the subsequently to be designed controller should guarantee. Therefore, the choice of $\weS=\yeS$ is already part of the construction of \QBA. 
In both cases, prior to the abstraction process, a state model of the system to be abstracted is required. 

\begin{definition}\label{def:EL}
 A \emph{state machine} is a tuple $\Qsyse=(\xeS,\ueS,\yeS,\tre,\xeSo{})$, where $\xeS$ is the set of states, $\xeSo{}$ is the set of initial states, $\ueS$ is the set of inputs, $\yeS$ is the set of outputs, and $\tre\subseteq\xeS\times\ueS\times\yeS\times\xeS$ is a next state relation.\\
 The set of \emph{admissible outputs} of a state $\xe\in\xeS$ is defined by 
\begin{subequations}
 \begin{align}
  \EnabY{\tr}{\x}&\deff\SetCompX{\ye\inps\yeS}{\ExQ{\ue\inps\ueS,\xe'\inps\xeS}{\Tuple{\xe,\ue,\ye,\xe'}\inps\tre}} \label{equ:EnabY}
 \end{align}
and $\Qsyse$ is said to be 
\emph{output deterministic} %
if 
\begin{equation}\label{equ:OutputDeterministic}
 \AllQ{\xe\in\xeS}{\propImp{\EnabY{\tr}{\x}\neq\emptyset}{\length{\EnabY{\tr}{\x}}=1}}.
\end{equation}
\end{subequations}
Furthermore,
\begin{subequations}\label{equ:TransStruct:all}
 \begin{align}
  \EnabF{\tr}{\x,\u}&\deff\SetCompX{\x'\inps\xS}{\ExQ{\y\inps\EnabY{\tr}{\x}}{\Tuple{\x,\u,\y,\x'}\inps\tr}},\\
  \EnabT{\tr}{\x}&\deff\SetCompX{\x'\inps\xS}{\ExQ{\u\in\uS}{\xe'\in\EnabF{\tr}{\x,\u}}}\label{equ:EnabT},
 \end{align}
 \end{subequations}
 are the sets of \emph{post-states} of a state-input pair $\Tuple{\xe,\ue}$ and a state $\xe$, respectively.
\end{definition}

%
%
If the state evolution and the output generation of a transition $\Tuple{\xe,\ue,\ye,\xe'}\inps\tre$ can be separated in $\Qsyse$ s.t.
\begin{equation}\label{equ:TransStruct}
\AllQ{\xe\inps\xeS,\ue\in\ueS}{\propAequ{\Tuple{\xe,\ue,\ye,\xe'}\inps\tre}{
\begin{propConjA}
 \xe'\in\EnabF{\tr}{\x,\u}\\
 \y\in\EnabY{\tr}{\x}
\end{propConjA}}},
\end{equation}
a state machine can be equivalently defined by the six-tuple $\Tuple{\xeS,\xeSo{},\ueS,\EnabF{\tr}{},\yeS,\EnabY{\tr}{}}$, which usually defines a transition system.\\
%
%
Using a state machine $\Qsyse$ to model the original system, its \emph{full behavior}, i.e., the set of infinite input, state, and output sequences compatible with its dynamics, is defined as follows.

\begin{definition}
 Let $\Qsyse$ be a state machine as in \REFdef{def:EL}. Then the \emph{full behavior of $\Qsys$} is defined by
\begin{equation}\label{equ:Behtr}
 \Behf(\Qsyse):=\SetCompX{\Tuple{\ueG,\yeG,\xeG}\in(\ueS\times\yeS\times\xeS)^{\Nbn}}{
\begin{propConjA}
\xeG(0)\inps\xeSo{}\\
 \AllQ{k\inps\Nbn}{\Tuple{\xeG(k),\ueG(k),\yeG(k),\xeG(k\plps1)}\inps\tre}
\end{propConjA}\hspace{-0.1cm}}\hspace{-0.1cm}.
\end{equation}
Furthermore, if
\begin{subequations}\label{equ:ReachLive}
 \begin{align}
&\AllQ{\x\in\xSo{}}{\ExQ{\Tuple{\uG,\yG,\xG}\in\Behf(\Qsyse)}{\xG(0)=\x}}~\text{and}\label{equ:ReachLive:0}\\
&\AllQ{\x\in\xS}{\ExQ{\Tuple{\uG,\yG,\xG}\in\Behf(\Qsyse),k\in\Nbn}{\xG(k)=\x}}\label{equ:ReachLive:k}
\end{align}
\end{subequations}
$\Qsyse$ is called \emph{live and reachable}.
\end{definition}

Whenever $\Qsyse$ is live and reachable, the dynamics of $\Qsyse$ can be equivalently described by its full behavior. As \SAlA are typically constructed from $\Behf(\Qsyse)$ instead of $\Qsyse$ we restrict attention to state machines that are live and reachable. Furthermore, since \QBA are usually constructed from transition systems which coincide with state machines if \eqref{equ:TransStruct} holds, we consider the following setup in this paper.\\[0.25cm]
%
\textit{
 Given a dynamical system $\Sys$, we assume that its external dynamics can be modeled by a state machine
\begin{subequations}\label{equ:Prelim}
 \begin{equation}\label{equ:Prelim:Q}
 \Qsyse=\Tuple{\xeS,\ueS,\yeS,\tre,\xeSo{}},~\SUCHTHAT~\text{\eqref{equ:TransStruct} and \eqref{equ:ReachLive} holds}
\end{equation}
and the external signal space
\begin{equation}\label{equ:Prelim:W}
\weS\in\Set{\ueS\times\yeS,\yeS}~\text{is finite.}
\end{equation}
\end{subequations}}
In the remainder of this paper we introduce two methods to construct a finite state abstraction of $\Qsyse$ in \eqref{equ:Prelim}, namely \emph{asynchronous $l$-complete approximations} (\SAlA) (from \cite{SchmuckRaisch2014_ControlLetters}) in \REFsec{sec:3:SAlA} and \emph{quotient based abstractions} (\QBA) (from \cite[part II]{TabuadaBook}) in \REFsec{sec:3:QuotientbasedAbstractions}. To provide a formal comparison of the resulting models in \REFsec{sec:4:Comparison}, we first introduce the notion of simulation relations.


\subsection{Simulation Relations}\label{sec:2:SimulationRelations}
Simulation relations are commonly used to compare system models in a step-by-step fashion. The idea is to investigate, if there exists a relation between the state spaces of two systems which ensures that trajectories 
of the first can be mimicked by the second system, such that only related states are visited and equivalent external symbols are generated by both systems.
%
To incorporate all possible choices of external signal spaces $\weS$ as in \eqref{equ:Prelim:W}, we slightly modify the usual definition of simulation relations for transition systems (e.g. \cite[Def. 4.7]{TabuadaBook}) as follows. 

\begin{definition}\label{def:SimRel_EL}
 Let $\Qsys_i=(\xS_i,\uS_i,\yS_i,\tr_i,\xSo{i}),~i\in\Set{1,2}$, be state machines and $\wS$ a set s.t. $\projState{\wS}{\uS_1\times\yS_1}=\projState{\wS}{\uS_2\times\yS_2}\neq\emptyset$.
 Then $\R\subseteq\xS_1\times\xS_2$ s.t.
 \begin{subequations}\label{equ:SimRel_EL} \allowdisplaybreaks
 \begin{align}
&\AllQ{\x_1\inps\xSo{1}}{\ExQ*{\x_2\inps\xSo{2}}{\Tuple{\x_1,\x_2}\inps\R}}\quad\text{and}\label{equ:SimRel_EL:a}\\
&\AllQSplit{\Tuple{\x_1,\x_2}\in\R, \u_1\in\uS_1,\y_1\in\yS_1,\xe_1'\in\xS_1}{
 \propImp{\Tuple{\x_1,\u_1,\y_1,\x_1'}\in\tr_1}{
 \ExQSplit{\u_2\in\uS_2,\y_2\in\yS_2,\x_2'\inps\xS_2}{
\begin{propConjA}
\Tuple{\x_2,\u_2,\y_2,\x_2'}\inps\tr_2\\
\Tuple{\x_1',\x_2'}\in\R\\
\projState{\wS}{\u_1,\y_1}=\projState{\wS}{\u_2,\y_2}
\end{propConjA}}}}\label{equ:SimRel_EL:b}
\end{align}
\end{subequations}
is a simulation relation from $\Qsys_1$ to $\Qsys_2$ w.r.t. $\wS$, denoted by $\R\in\SR{\wS}{}{\Qsys_1}{\Qsys_2}$.
\end{definition}

Using \REFdef{def:SimRel_EL} we can formally define an ordering on the set of state machines in the usual way.

\begin{definition}\label{def:SimRel_Symbols}
 Given the premises of \REFdef{def:SimRel_EL}
 , a state machine $\Qsys_1$ is  simulated by $\Qsys_2$ w.r.t. $\wS$, denoted by $\Qsys_1\kgl{\wS}{}\Qsys_2$, if there exists a relation $\R\in\SR{\wS}{}{\Qsys_1}{\Qsys_2}$. Furthermore, $\Qsys_1$ and $\Qsys_2$ are  bisimilar w.r.t. $\wS$, denoted by $\Qsys_1\hgl{\wS}{}\Qsys_2$,  if there exists a relation $\R\in\SR{\wS}{}{\Qsys_1}{\Qsys_2}$ also satisfying\footnote{As usual, $\R^{-1}:=\SetComp{\Tuple{\xe_2,\xe_1}}{\Tuple{\xe_1,\xe_2}\in\R}$.} $\R^{-1}\in\SR{\wS}{}{\Qsys_2}{\Qsys_1}$.
\end{definition}

\section{Strongest Asynchronous $l$-Complete Approximations (\SAlA)}\label{sec:3:SAlA}

%
The idea of \SAlA is to exactly mimic the external behavior of $\Qsyse$ in \eqref{equ:Prelim} over finite time intervals of length $l+1$. 
We therefore consider the behavioral system $\Ee=\Tuple{\Nbn,\weS,\BeheQ}$, where $\BeheQ$ is the extension of $\projState{\weS}{\Behf(\Qsyse)}$ to $\Zb$ as discussed in \REFsec{sec:prelim}. 
All finite strings of external symbols of length $l$ which are consistent with the dynamics of $\Qsyse$ are given by
\begin{align}
&\Pi_{l}(\BeheQ):=\bigcup_{k\in\Nbn}\BeheQ\ll{k-l+1,k}.\label{equ:Ds}
\end{align}
Now consider the following gedankenexperiment: assume playing a sophisticated domino game where $\Pi_{l+1}(\BeheQ)$ is the set of dominos. Pick the first domino to be $\BeheQ\ll{-l,0}$ (i.e., a domino with only diamonds except for the last symbol) and append any domino from the set $\Pi_{l+1}(\BeheQ)$ if the last $l$ symbols of the first domino are the same as the first $l$ symbols of the second domino (see Figure~\ref{fig:DominoGame} (left) for an example).
Playing the domino game arbitrarily long and with all possible initial conditions and domino combinations results in the largest, in the sense of set inclusion, behavior $\Behal$ satisfying 
\begin{subequations}\label{equ:BehalW}
\begin{align}
 &\Behal\ll{-l,0}=\BeheQ\ll{-l,0}~\text{and}~\label{equ:BehalW:a}\\
 &\Pi_{l+1}(\Behal)=\Pi_{l+1}(\BeheQ), \label{equ:BehalW:b}
 \end{align}
\end{subequations}
%
%
%
defining the behavioral system ${\Eal=\Tuple{\Nbn,\weS,\Behal}}$.
%
%
%
Observe that the smaller $l$, the less information in the domino game is used, which generates more freedom in constructing signals, implying $\Behal\supseteq\allowbreak\abst{\Beh{}}^{l+1}\supseteq\allowbreak\BeheQ$ for all $l\in\Nbn$. 
This motivates the use of $\Behal$ as an over-approximation of the behavior $\BeheQ$.
Obviously, equality ${\Beha^r=\BeheQ}$ holds for all $r\geq l$ if $\BeheQ$ is itself the largest behavior satisfying \eqref{equ:BehalW}. In \cite{SchmuckRaisch2014_ControlLetters}, a system $\E=\Tuple{\Nbn,\weS,\BeheQ}$ for which the latter is true was called \emph{asynchronously $l$-complete} which inspired the name of \SAlA. 
Following \cite{SchmuckRaisch2014_ControlLetters}, $\Eal$ constructed in the outlined domino game is the unique \SAlA of $\Ee=\Tuple{\Nbn,\weS,\BeheQ}$.  However, we are usually interested in a state machine \emph{realizing} its step by step evolution.

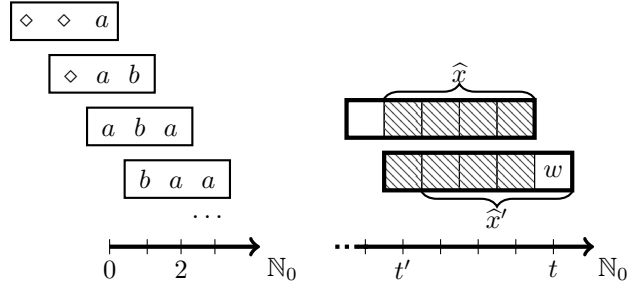
\begin{figure}
\begin{center}
\begin{tikzpicture}[auto, node distance=0.5cm,scale=1]

  \tikzstyle{myblock} = [draw,rectangle, text centered, minimum height=0.5cm, minimum width=1.4cm,thick]
 \node[myblock] (b1) at (0,0) {$\diamond~~\diamond~~a$};
 \node[myblock,below of=b1,xshift=0.5cm,node distance=0.7cm] (b2) {$\diamond~~a~~b$};
 \node[myblock,below of=b2,xshift=0.5cm,node distance=0.7cm] (b3) {$a~~b~~a$};
  \node[myblock,below of=b3,xshift=0.5cm,node distance=0.7cm] (b4) {$b~~a~~a$};
   \node[below of=b4,xshift=0.4cm,node distance=0.5cm] (b5) {$\hdots$}; 
 \draw [draw,->,ultra thick] (0.6,-3cm) -- ++ (2cm,0);
\foreach \x in {0.6,1.1,1.55,2}
\draw (\x cm,-2.9cm) -- (\x cm,-3.1cm);
\draw (0.6,-3.3cm) node {$ 0 $};
\draw (1.55,-3.3cm) node {$ 2 $};
\draw (2.9,-3.3cm) node {$\Nbn$};

  \tikzstyle{myblock} = [draw,rectangle, text centered, minimum height=0.5cm, minimum width=0.5cm]
   \tikzstyle{myblocks} = [draw,rectangle, text centered, minimum height=0.5cm, minimum width=0.5cm,postaction={pattern color=black!70,pattern=north west lines}] 
   

   \def\n{4} \def\c{1.3cm}
   
  \node[myblock] (b1) at (\n,-\c) {};
  \node[myblocks, right of=b1] (b2)  {};
   \node[myblocks, right of=b2] (b3) {};
   \node[myblocks, right of=b3] (b4) {};
   \node[myblocks, right of=b4] (b5) {}; 
   \node[myblocks,below of=b2,node distance=0.7cm] (a1) {};
  \node[myblocks, right of=a1] (a2) {};
   \node[myblocks, right of=a2] (a3) {};
   \node[myblocks, right of=a3] (a4) {};
   \node[myblock, right of=a4] (a5) {$\w$};   
   \node[draw,rectangle, text centered, minimum height=0.5cm, minimum width=2.5cm,ultra thick] () at (b3) {};
    \node[draw,rectangle, text centered, minimum height=0.5cm, minimum width=2.5cm,ultra thick] () at (a3) {};  
   \draw [decorate,decoration={brace,amplitude=6pt},thick] (b2.north west) -- (b5.north east) node [black,midway,yshift=0.1cm] (bc){$\xa$};
   \draw [decorate,decoration={brace,amplitude=6pt},thick] (a5.south east) -- (a2.south west) node [black,midway,yshift=-0.1cm] (bc){$\xa'$};
   \draw [draw,-,ultra thick,dotted] (\n-0.4,-1.7cm-\c) -- ++ (\n+0.3cm,0);
 \draw [draw,->,ultra thick] (\n-0.1,-1.7cm-\c) -- ++ (\n+3cm,0);
 \def\d{0.5}
\foreach \i in {0,...,5}
\draw (\n+\i*\d,-1.6cm-\c) -- (\n+\i*\d,-1.8cm-\c);
\draw (\n+0.5,-2cm-\c) node {$t' $};
\draw (\n+2.5,-2cm-\c) node {$ t$};
\draw (\n+3.3,-2cm-\c) node {$\Nbn$};

%

  \end{tikzpicture}
  \end{center}
  \caption{Example of a domino game for $l=2$ (left) and an illustration of the usual choice $\xaS^l$ in \REFprop{prop:SforLcomplete} for $t>l=4$ with $t'=t-l$ (right).}\label{fig:DominoGame}
\end{figure}
\begin{definition}\label{def:SAlA_Realization}
 Given \eqref{equ:Prelim} and \eqref{equ:BehalW}, the dynamical system ${\Eal=\Tuple{\Nbn,\weS,\Behal}}$ is the \SAlA of $\Ee=\Tuple{\Nbn,\weS,\BeheQ}$. A state machine $\Qsysa$ is a realization of $\Eal$ 
 iff\footnote{As before,  $\Behe(\Qsysa)$ denotes the extension of $\projState{\weS}{\Behf(\Qsysa)}$ to $\Zb$.} $\Behal=\Beh(\Qsysa)$.
\end{definition}
%
In the work on \SlA and \SAlA the state space $\xaS$ to construct the realization $\Qsysa$ of the abstraction $\Eal$ is usually chosen such that the state represents the \enquote{recent past} of length $l$ of the external signal. 
Recalling the gedankenexperiment, this choice of $\xaS$ is motivated by the fact that the  next feasible domino of length $l+1$ is determined by the last $l$ symbols of the previous domino (see \REFfig{fig:DominoGame} (right) for an illustration).
Using this state space, the standard state machine realization of \SAlA, denoted by $\Qsysa^l$ in this paper, is defined as follows.

\begin{proposition}[\cite{SchmuckRaisch2014_ControlLetters}, Thm.4]\label{prop:SforLcomplete}
Let $\Eal=\Tuple{\Nbn,\weS,\Behal}$ be the \SAlA of $\Ee$ and define
\begin{subequations}\label{equ:Qsysal_old}
\begin{align}\allowdisplaybreaks
 \xaS^l:=&\Set{\diamond}^l\cup\Pi_l(\Behal),\\
 \xaSo{}^l:=&\Set{\diamond}^l,~\text{and}\\
 \tra^l:=&\SetCompX{\Tuple{\xa,\we,\BR{\xa\sconc\we}\ll{1,l}}}{\xa\sconc\we\in\Pi_{l+1}(\Behal)}.
\end{align}
\end{subequations}
Then $\Eal$ is realized by $\Qsysa^l=\Tuple{\xaS^l,\weS,\tra^l,\xaSo{}^l}$.
\end{proposition}

Summarizing the abstraction procedure outlined above, constructing the finite state abstraction $\Qsysa^l$ in \REFprop{prop:SforLcomplete} using \SAlA only requires knowledge about the set $\Pi_{l+1}(\BeheQ)$.
However, if $\Qsyse$ is available, we can construct $\Qsysa^l$ from $\Qsyse$ directly,
as shown in the following section. 


\subsection{Some State Machine Realizations of \SAlA}\label{sec:SAlA_SM}

Recall from \REFprop{prop:SforLcomplete} that the set of external sequences of length $l$, given by $\Pi_l(\Behal)=\Pi_l(\BeheQ)$ (from \eqref{equ:BehalW:b}), is finite. We now investigate how to use this set as a state space in the construction of different state machine realizations of the \SAlA of a system $\E$. This is be done on the basis of a state machine realization $\Qsyse$ of $\Ee$ satisfying \eqref{equ:Prelim}. For this,
we first investigate how a string $\zeG\in\Pi_l(\BeheQ)$ can correspond to a state $\xe\in\xeS$ of $\Qsyse$. Observe that $\zeG$ is a string of length $l$ and $\xe$ is a state reached at a particular time $k\in\Nbn$. We consider the cases where $\zeG$ is generated by $\Qsyse$ immediately \emph{before}, immediately \emph{after} or \emph{while} $\xe$ was reached. This leads us to a set of intervals 
\begin{equation}\label{equ:Prelim:I}
 \Ilm=[m-l,m-1]\quad\SUCHTHAT~ l,m\in\Nbn,~\text{and}~m\leq l,
\end{equation}

where\footnote{The addition of two intervals is interpreted in the usual sense, i.e., $[a,b]+[c,d]=[a+c,b+d]$.} $[k,k]+\Interval^l_0=[k-l,k-1]$ corresponds to the first,  $[k,k]+\Interval^l_l=[k,k+l-1]$ corresponds to the second, and for all other choices of $m$, $[k,k]+\Ilm$ corresponds to the third case. Based on \eqref{equ:Prelim:I} the sets of compatible states are introduced in \REFdef{def:Xxr} and illustrated in \REFfig{fig:Enabl}.%
\begin{figure}
\begin{center}
\begin{tikzpicture}[auto,scale=1]
 
 \tikzset{
>=stealth',
help lines/.style={dashed, thick},
axis/.style={<->},
important line/.style={thick},
connection/.style={thick, dotted},
}

\def\y{0} \def\yd{0.5} \def\x{0} \def\xd{1} \def\xmax{10} \def\ymax{2}

\coordinate (zero) at (\x,\y);
\coordinate (top) at (\x,\y+5);
\coordinate (start) at (\x,\y+2);

\foreach \nx/\ny/\a in {1/0/a,2/1/b,3/2/c,4/1.5/c,5/0.5/b,6/0.8/b}{
\draw[help lines] (\x+\nx*\xd,\y) -- (\x+\nx*\xd,\y+\ymax-0.5);
\coordinate (x\nx) at start+(\nx*\xd,\ny*\yd);
\fill [black] (x\nx) node {}  circle [radius=3pt];
\draw (\nx*\xd,\ymax-0.2) node {$ \a $};
}
\draw (\x+0.2,\ymax-0.2) node {$ \hdots $};
\draw (\x-0.5,\ymax-0.2) node {$ \yeG$:};
\draw (6.8*\xd,\ymax-0.2) node {$ \hdots $};
\draw (x4)+(0.3,0) node {$x$};

\foreach \a/\b/\c/\z in {1/3/0/0,2/4/0.2/1,3/5/0.4/2,4/6/0.6/3}{
\draw[important line] (\a*\xd-0.1,\ymax+0.05+\c) -- (\a*\xd-0.1,\ymax+0.15+\c) --(\b*\xd+0.1,\ymax+0.15+\c) --(\b*\xd+0.1,\ymax+0.05+\c);
\draw (\a*\xd+0.4,\ymax+\c+0.35) node {$\zeG_{\z}$};
}

\draw[important line] (x1) to[out=0,in=-140] (x2) to[out=40,in=-180] (x3) to[out=0,in=-210] (x4) to[out=-30,in=-190] (x5) to[out=-10,in=-150] (x6);
\draw[black, dotted,thick] (\x,\y+\yd) to[out=-10,in=-180] (x1);
\draw[black, dotted,thick] (x6) to[out=20,in=-150] (\x+7*\xd,\y+2*\yd);
\draw (\x-0.5,\y+\yd) node {$ \xeG $:};



\end{tikzpicture}
  \end{center}
 \caption{Illustration of corresponding external sequences $\zeG_m\in\ON{E}^{\Interval^{3}_{m}}(\xe),~m\in\{0,\hdots,3\}$ for state $\xe=\xeG(k)$ where $\weS=\yeS=\Set{a,b,c}$ and $\yeG(k)\in\ON{H}_{\tre}(\xe)$ for some $k\in\Nbn$.}\label{fig:Enabl}
\end{figure}
\begin{definition}\label{def:Xxr}
Given \eqref{equ:Prelim} and \eqref{equ:Prelim:I}, let $\ESe{}=(\Nbn,\weS\times\xeS,\BeheSQ)$ be a dynamical system, where $\BeheSQ$ is the extension of $\projState{\weS\times\xeS}{\Behf(\Qsyse)}$ to $\Zb$ as discussed in \REFsec{sec:prelim}. Then the set of \emph{corresponding external strings w.r.t. $\Ilm$} is defined for every state $\xe\in\xeS$ by
\begin{equation}\label{equ:Xxr}
\EnabWl{}{\Ilm}{\xe}\hspace{-1mm}\deff\SetCompX{\zeG}{
\ExQ{\Tuple{\weG,\xeG}\in\BeheSQ,k\in\Nbn\hspace{-1mm}}{\hspace{-1mm}\begin{propConjA}
                                         \xeG(k)=\xe\\
					\zeG=\weG|_{[k,k]+\Ilm}
                                         \end{propConjA}
}\hspace{-1mm}}\hspace{-1mm}.
\end{equation}
Furthermore, if 
\begin{equation}\label{equ:future_unique}
 \AllQ{\xe\in\xeS,\zeG,\zeG'\in\EnabWl{}{\Ilm}{\xe}}{\zeG\ll{l-m,l-1}=\zeG'\ll{l-m,l-1}}
\end{equation}
$\Qsyse$ is called \emph{future unique} w.r.t. $\Ilm$.
\end{definition}


Observe, that $\zeG,\zeG'\in\EnabWl{}{\Ilm}{\xe}$ in \eqref{equ:future_unique} are obtained from two trajectories $\Tuple{\weG,\xeG},\Tuple{\weG',\xeG'}\in\BeheSQ$ passing $\xe$ at time $k\in\Nbn$ and $k'\in\Nbn$, respectively, (i.e., $\xeG(k)=\xeG'(k')=\xe$) using \eqref{equ:Xxr}. During this restriction of $\weG$ (resp. $\weG'$) to $\zeG$ (resp. $\zeG'$) absolute time information is disregarded (see \REFsec{sec:notation}), implying $\zeG\ll{l-m,l-1}=\weG\ll{k,k+m-1}$ and $\zeG'\ll{l-m,l-1}=\weG\ll{k',k'+m-1}$. Therefore, 
$\Qsyse$ is future unique w.r.t. $\Ilm$ if for all states $\xe\in\xeS$ all trajectories passing $\xe$ have the same $m$-long (non-strict) future of external symbols, i.e. $\weG\ll{k,k+m-1}=\weG'\ll{k',k'+m-1}$. Using this intuition it is easy to see that $\Qsyse$ is always \emph{future unique} w.r.t. $\Interval^l_0=[-l,-1]$, as this interval has no future.\\
We now proceed by constructing $m$ finite state machines using the outlined correspondence between $\xeS$ and $\Pi_{l}(\BeheQ)$.

\begin{definition}\label{def:QsysalW}
Given \eqref{equ:Prelim} and \eqref{equ:Prelim:I}, define 
\begin{subequations}\label{equ:QsysalW}
\begin{align}
\xalS:=&\SetCompX{\zeta}{\ExQ{\xe\in\xeS}{\zeta\in\EnabWl{}{\Ilm}{\xe}}},\label{equ:xalS}\\
\xalSo{}:=&\SetCompX{\zeta}{\ExQ{\xe\in\xeSo{}}{\zeta\in\EnabWl{}{\Ilm}{\xe}}},~\text{and}\label{equ:xalSo}\\
 %
\tral\hspace{-1mm}:=&\SetCompX{\Tuple{\xa,\ue,\ye,\xa'}}{
\begin{propConjA}
\xa'\ll{0,l\mips m\mips1}=\BR{\xa\ll{0,l\mips m\mips1}\sconc\projState{\weS}{\ue,\ye}}\ll{1,l\mips m}\hspace{-1mm}\\[0.1cm]
~\xa\ll{l\mips m,l\mips 1}=\BR{\projState{\weS}{\ue,\ye}\sconc\xa'\ll{l\mips m,l\mips 2}}\ll{0,m-1}\\[0.1cm]
 \ExQ{\xe,\xe'\in\xeS}{
\begin{propConjA}
 \xa\in \EnabWl{}{\Ilm}{\xe}\\
 \xa'\in \EnabWl{}{\Ilm}{\xe'}\\
\Tuple{\xe,\ue,\ye,\xe'}\in\tre
\end{propConjA}}
\end{propConjA}
}\hspace{-1mm}.\label{equ:tral}
\end{align}
\end{subequations}
 Then $\Qsysal=\Tuple{\xalS,\ueS,\yeS,\tral,\xalSo{}}$ is called the $\Ilm$-abstract state machine of $\Qsyse$. 
 \end{definition}
 
 
The construction of the abstract state machines in \REFdef{def:QsysalW} can be interpreted as follows.
Using \eqref{equ:xalS} instead of $\xalS=\Set{\diamond}^l\cup\Pi_{l}(\BeheQ)$ ensures that $\Qsysal$ is live and reachable, which is purely cosmetic but allows to simplify subsequent proofs. The last line in the conjunction of \eqref{equ:tral} simply says that we have a transition in $\Qsysal$ from $\xa$ to $\xa'$ if there is a transition in $\Qsyse$ between any two states compatible with $\xa$ and $\xa'$, respectively. However, the first two lines in the conjunction of \eqref{equ:tral} additionally ensure that $\xa$ and $\xa'$ obey the rules of the domino game, i.e.,
\begin{equation*}
 \xa\ll{1,l-1}=\xa'\ll{0,l-2}
\end{equation*}
as depicted in \REFfig{fig:DominoGame} (right) and the current external symbol $\we=\projState{\weS}{\ue,\ye}$ is contained in either $\xa$ or $\xa'$ or both, at the position corresponding to the current time point, i.e.,
\begin{align*}
 \we&=\xa'(l-1)~\text{if}~m=0,\\
 \we&=\xa(l-m)=\xa'(l-1-m)~\text{if}~0<m<l~\text{and}\\
 \we&=\xa(0)~\text{if}~m=l.
\end{align*}
As we are interested in state machine realizations of \SAlA, we show that $\Qsysal$ realizes $\Eal$ for all choices of $l$ and $m$.

\begin{theorem}\label{thm:behequ}
Given \eqref{equ:Prelim} and \eqref{equ:Prelim:I}, let $\Qsysal$ be defined as in \REFdef{def:QsysalW} and let $\Eal=\Tuple{\Nbn,\weS,\Behal}$ be the unique \SAlA of $\Ee=\Tuple{\Nbn,\weS,\BeheQ}$. Then $\Qsysal$ realizes $\Eal$.
\end{theorem}
%
%
%
%
%
\begin{proof}
See Appendix~\ref{proof:thm:behequ}.
%
\end{proof}


As an intuitive consequence of \REFthm{thm:behequ}, choosing $m=0$ and the full external symbol set $\weS=\ueS\times\yeS$ when constructing $\Qsysal$ in \REFdef{def:QsysalW} yields the standard realization $\Qsysa^l$ of \SAlA.

 \begin{theorem}\label{thm:Qsysalo_equ}
 Given \eqref{equ:Prelim} and \eqref{equ:Prelim:I} with $\weS=\ueS\times\yeS$, let $\Qsysa^l$ and $\Qsysal$ as in \REFprop{prop:SforLcomplete} and \REFdef{def:QsysalW}, respectively. Then $\Qsysa^l=\Qsysa^{\Interval^l_0}$.
 \end{theorem}
 
 \begin{proof}
 See Appendix~\ref{proof:thm:Qsysalo_equ}.
 \end{proof}

\subsection{Ordering $\Qsysal$ based on Simulation Relations}

Before we discuss the ordering between abstract state machines based on changing $l$ and $m$, we show under which conditions the obtained abstraction $\Qsysal$ simulates the original state machine $\Qsyse$ and when both state machines are bisimilar. This investigation is interesting for the comparison to \QBA, as the latter always simulates the original state machine $\Qsyse$. Furthermore, the framework of \QBA allows to construct a bisimilar abstraction whenever the employed repartitioning algorithm terminates. Hence, it is interesting to know if the latter is also true for \SAlA.\\
The investigation of similarity between $\Qsysal$ and $\Qsyse$ requires the construction of a relation between the original state space $\xeS$ and the abstract state space $\xalS$. As $\xalS$ defines a cover for $\xeS$ where each cell is given by all states $\xe$ corresponding to a string $\zeG\in\xalS$ via $\EnabWl{}{\Ilm}{}$, the latter is a natural choice for a relation between $\xeS$ and $\xalS$.\\
Recall from \REFthm{thm:behequ} that the behaviors of $\Qsyse$ and $\Qsysal$ coincide if $\BeheQ$ is asynchronously $l$-complete. Behavioral equivalence is always necessary for a relation $\R$ to be a bisimulation relation but usually not sufficient. We therefore introduce a stronger condition, called \emph{state-based asynchronous $l$-completeness}, to serve the latter purpose.

\begin{definition}\label{def:SBalc}
Given \eqref{equ:Prelim}, $\Qsyse$ is \emph{state-based asynchronously $l$-complete w.r.t. $\Ilm$} if
 \begin{equation}\label{equ:dominolcomplete}
  \AllQ{\xe\in\xeS,\zeG\in\Pi_{l+1}(\BeheQ)}{\propImp{\zeG\ll{0,l-1}\in\EnabWl{}{\Ilm}{\xe}}{\zeG\in\EnabWl{}{[m-l,m]}{\xe}}}.
 \end{equation}
\end{definition}

\begin{remark}\label{rem:SBalc}
Recall from the beginning of this section that the dynamical system $\Ee=\Tuple{\Nbn,\weS,\BeheQ}$ is asynchronously $l$-complete, as  defined in \cite[Def.6]{SchmuckRaisch2014_ControlLetters}, if $\BeheQ$ is the largest behavior satisfying \eqref{equ:BehalW} itself. Intuitively, the latter is true if \emph{for all} $\zeG\in\Pi_{l+1}(\BeheQ)$ there \emph{exists} an $\xe\in\xeS$ s.t. the second part of \eqref{equ:dominolcomplete} holds. Therefore, asynchronous $l$-completeness of $\Ee$ is always implied by \eqref{equ:dominolcomplete}, but not vice-versa. 
\end{remark}

\begin{theorem}\label{thm:SimRel_QsyseQsysal}
Given \eqref{equ:Prelim}, \eqref{equ:Prelim:I} and $\Qsysal$ as in \REFdef{def:QsysalW}, let
\begin{equation}
\R=\SetCompX{\Tuple{\xe,\xa}\in\xeS\times\xalS}{
\xa\in\EnabWl{}{\Ilm}{\xe}
}\label{equ:R0}.
\end{equation}
Then it holds that\footnote{Using $\mathfrak{R}_{\ueS\times\yeS}$ instead of $\mathfrak{R}_{\weS}$ in (i) is done on purpose and indicates that this relation holds for $\ueS\times\yeS$ independent from the choice of $\weS$.} 
\begin{compactenum}[(i)]
 \item $\propAequ{\R\in\SR{\ueS\times\yeS}{}{\Qsyse}{\Qsysal}}{\text{$\Qsyse$ is future unique w.r.t. $\Ilm$}}$ and
 \item $\propAequ{{\R^{-1}\in\SR{\weS}{}{\Qsysal}{\Qsyse}}}{\text{$\Qsyse$ is state-based asych. $l$-complete w.r.t. $\Ilm$}}$.
\end{compactenum}
\end{theorem}

\begin{proof}
See Appendix~\ref{proof:thm:SimRel_QsyseQsysal}. 
\end{proof}

Intuitively, $\Qsysal$ simulates $\Qsyse$ w.r.t. $\weS$ if for every related state pair $\Tuple{\xe,\xa}\in\R$ and every transition $\Tuple{\xe,\ue,\ye,\xe'}\in\tre$ which $\Qsyse$ \enquote{picks}, $\Qsysal$ can \enquote{pick} a transition $\Tuple{\xa,\ue',\ye',\xa'}\in\tral$ s.t. $\we=\projState{\weS}{\ue,\ye}=\projState{\weS}{\ue',\ye'}$. However, if $m>0$, a state $\xa\in\xalS$ has only outgoing transitions s.t. $\we=\xa(l-m)$. Therefore, $\Qsysal$ can only simulate $\Qsyse$ iff in every state $\xe\in\xeS$ all outgoing transitions agree on this $\we$, i.e., $\Qsyse$ is \enquote{output deterministic} w.r.t. $\weS$. For $m>1$ applying this reasoning iteratively gives the (rather restrictive) condition of future uniqueness of $\Qsyse$.
As the outlined problems are absent for $m=0$ (as $\Qsyse$ is always future unique w.r.t. $\Interval^l_0$), $\Qsysa^{\Interval^l_0}$, which we know to coincide with the original realization $\Qsysa^l$ of \SAlA for $\weS=\ueS\times\yeS$, always simulates $\Qsyse$.

\begin{corollary}
 Given \eqref{equ:Prelim}, \eqref{equ:Prelim:I} and $\Qsysa^{\Interval^l_0}$ as in \REFdef{def:QsysalW} it holds that 
  $\Qsyse\kgl{\ueS\times\yeS}{}\Qsysa^{\Interval^l_0}$.
\end{corollary}

 \begin{remark}
 In the context of \SlA a state machine $\Qsysa^{l^+}$ was introduced in \cite{Raisch2010} whose state at time $k$ represents the string of external symbols from time $k-l+1$ to time $k$, i.e., from the interval $k+\Interval^l_1$. While the state sets of  $\Qsysa^{l^+}$ and $\Qsysa^{\Interval^l_1}$ coincide, their transition structure slightly differs. This is a consequence of the fact that $\Qsysa^{l^+}$ was intended to serve as a set-valued observer for the states of $\Qsyse$.
  \end{remark}
  
Recalling the domino game, we know that using longer dominos (i.e., increasing $l$) gives less freedom in composing them and therefore yields a tighter abstraction. This intuition carries over to the state space realizations of $\Eal$, inducing an ordering in terms of simulation relations. 

\begin{theorem}\label{thm:SimRel_QsysalpmQsysalm}
Given \eqref{equ:Prelim}, \eqref{equ:Prelim:I} and $\Qsysal$ as in \REFdef{def:QsysalW}, let
\begin{equation}\label{equ:R_l}
  \R=\SetCompX{\Tuple{\xa_{l+1},\xa_l}\in\xaS^{{\Ilpm}}\times\xaS^{{\Ilm}}}{
  \xa_{l}=\xa_{l+1}\ll{1,l}}.
 \end{equation}
Then it holds that 
 \begin{compactenum}[(i)]
  \item $\R\in\SR{\weS}{}{\Qsysa^{\Ilpm}}{\Qsysa^{\Ilm}}$ and 
  \item $\propAequ{\R^{-1}\in\SR{\weS}{}{\Qsysa^{\Ilm}}{\Qsysa^{\Ilpm}}}{\Behal=\Beha^{l+1} }$.
 \end{compactenum}
\end{theorem}

\begin{proof}
See Appendix~\ref{proof:thm:SimRel_QsysalpmQsysalm}. 
\end{proof}

\REFthm{thm:SimRel_QsysalpmQsysalm} (ii) implies that the accuracy of the abstraction cannot be increased by increasing $l>r$ if $\BeheQ$ is asynchronously $r$-complete and $m$ is fixed, e.g. $m=0$. Therefore, the standard realization $\Qsysa^l$ for \SAlA might never result in a bisimilar abstraction of $\Qsyse$, no matter how large $l$ is chosen, even if $\Ee=\Tuple{\Nbn,\weS,\BeheQ}$ is asynchronously $r$-complete. This is due to the fact that 
\eqref{equ:dominolcomplete} is not implied by asynchronous $l$-completeness of $\Ee$ (see \REFrem{rem:SBalc}).\\
Interestingly, we will show that increasing $m$, i.e., shifting the interval into the future, results in a tighter abstraction w.r.t. simulation relations, i.e. allows to increase the precision of $\Qsysal$ for $l\geq r$ even if $\Ee$ is $r$-complete. 

\begin{theorem}\label{thm:SimRel_QsysalmpQsysalm}
Given \eqref{equ:Prelim}, \eqref{equ:Prelim:I}, and $\Qsysal$ as in \REFdef{def:QsysalW} with $m<l$, let
\begin{equation}\label{equ:R_m}
 \R=\SetCompX{\Tuple{\xa_{m+1},\xa_m}\in\xaS^{{\Ilmp}}\times\xaS^{{\Ilm}}}{
  \begin{propConjA}
  \xa_{m+1}\ll{0,l-2}=\xa_m\ll{1,l-1}\\
  \ExQ{\xe\in\xeS}{
  \begin{propConjA}
   \xa_{m+1}\in\EnabWl{}{{\Ilmp}}{\xe}\\
   \xa_{m}\in\EnabWl{}{{\Ilm}}{\xe}
  \end{propConjA}}
  \end{propConjA}}.
 \end{equation}
Then it holds that
\begin{compactenum}[(i)]
 \item  $\R\in\SR{\weS}{}{\Qsysa^{\Ilmp}}{\Qsysa^{\Ilm}}$ and 
 \item $
  \propAequ{\R^{-1}\in\SR{\weS}{}{\Qsysa^{\Ilm}}{\Qsysa^{\Ilmp}}}{
 \begin{propConjA}
  \text{$\Qsyse$ is future unique w.r.t. $\Ilmp$}\\
  \text{$\Qsyse$ is state-based async. $l$-complete w.r.t. $\Ilm$}
 \end{propConjA}}
 $
\end{compactenum}
\end{theorem}

\begin{proof}
 See Appendix~\ref{proof:thm:SimRel_QsysalmpQsysalm}.
\end{proof}
It is important to note that future uniqueness and state-based asynchronous $l$-completeness are incomparable properties, i.e., none is implied by the other. Therefore, there exist situations where $\Qsysal$ with $m>0$ simulates $\Qsyse$ (i.e., $\Qsyse$ is future unique w.r.t. $\Ilmp$) and $\Qsysal$ is tighter than $\Qsysa^{\Interval^l_0}$ in terms of simulation relations. However, if $\Qsyse$ is both future unique and state-based asynchronously $l$-complete w.r.t. a particular interval $\Interval^r_n$, \REFthm{thm:SimRel_QsysalmpQsysalm} implies that increasing $l>r$ and $m>n$ will not result in a tighter abstraction. Moreover, this is not necessary anyway, as \REFthm{thm:SimRel_QsyseQsysal} implies that in this case $\Qsysa^{\Interval^r_n}$ is bisimilar to $\Qsyse$. 

\subsection{Example}\label{sec:SAlCA_new_exp}
We conclude this section with a detailed example illustrating the construction of $\Ilm$-abstract state machines and the property of future uniqueness and state-based asynchronous $l$-completeness for different choices of $l$ and $m$.
%
 For simplicity, we consider a \emph{finite} state machine 
 \begin{align}
  \Qsyse&=\Tuple{\xeS,\ueS\times\yeS,\tre,\xeSo{}}\quad\SUCHTHAT\quad \weS=\yeS
  \label{equ:example:Qsyse}
 \end{align}
 as the original model, whose transition structure is depicted in \REFfig{fig:exp_Qsyse}.
 \begin{figure}
  \begin{center}
\begin{tikzpicture}[auto]
\def\dh{1.5} \def\dv{1}

\def\h{0} \def\v{0}

\node (name) at (\h-1,\v+0.5) {$\Qsyse:$};

\node[istate] (x1) at (\h,\v) {$\xe_1$};
\node[Sstate] (x2) at (\h+\dh,\v) {$\xe_2$};
\node[Sstate] (x3) at (\h+\dh,\v-\dv) {$\xe_3$};
\node[Sstate] (x4) at (\h+\dh,\v-2*\dv) {$\xe_4$};
\node[istate] (x5) at (\h,\v-2*\dv) {$\xe_5$};
\SFSAutomatEdge{x1}{\EdgeLabelYt{\ue_1}{\ye_1}}{x2}{}{}
\SFSAutomatEdge{x2}{\EdgeLabelYt{\ue_2}{\ye_2}}{x3}{bend left}{}
\SFSAutomatEdge{x3}{\EdgeLabelYt{\ue_3}{\ye_3}}{x2}{bend left}{}
\SFSAutomatEdge{x3}{\EdgeLabelYt{\ue_3}{\ye_3}}{x4}{bend left}{}
\SFSAutomatEdge{x4}{\EdgeLabelYt{\ue_4}{\ye_4}}{x3}{bend left}{}
\SFSAutomatEdge{x5}{\EdgeLabelYt{\ue_1}{\ye_1}}{x4}{}{swap}

\end{tikzpicture}
  \end{center}
  \caption{Transition structure of the state machine $\Qsyse$ in \REFfig{fig:exp_Qsyse}.}\label{fig:exp_Qsyse}
 \end{figure}
 It can be inferred from \REFfig{fig:exp_Qsyse} that the output behavior of $\Qsyse$ is given by
  \begin{equation*}
   \BeheQ=\Set{y_1y_2((y_3y_2)^*(y_3y_4)^*)^\omega,y_1y_4((y_3y_2)^*(y_3y_4)^*)^\omega}
  \end{equation*}
  where $(\cdot)^*$ and $(\cdot)^\omega$ denote, respectively, the finite and infinite repetition of the respective string. Furthermore, the sets of $1$-long and $2$-long dominos obtained from $\BeheQ$ via \eqref{equ:Ds} are
  \begin{align*}
  \Ds{}{1}&=\yeS\quad\text{and}\quad\\
  \Ds{}{2}&=\Set{\diamond y_1,~y_1y_2,~y_1y_4,~y_2y_3,~y_3y_2,~y_3y_4,~y_4y_3}.
  \end{align*}
  To play the domino-game for $l=1$, i.e., with dominos from the set $\Ds{}{2}$, we have to pick $\diamond y_1$ as the initial domino and append dominos such that the last element of the first matches the first element of the second domino. It is easy to see that in this example every such combination of dominos yields a sequence contained in $\BeheQ$. Hence, $\Ee=\Tuple{\Nbn,\yeS,\BeheQ}$ is asynchronously $1$-complete and therefore also asynchronously $2$-complete.
  
 
Using $\Qsyse$ in \REFfig{fig:exp_Qsyse} we can construct the $\Interval^1_0$- and  $\Interval^1_1$-abstract state machines of $\Qsyse$ using \REFdef{def:QsysalW}. Their transition structures are depicted in \REFfig{fig:exp_Qsysa1}. Furthermore, we obtain the following properties of $\Qsyse$ w.r.t $\Interval^1_0$ and $\Interval^1_1$.

 \begin{figure}
  \begin{center}
\begin{tikzpicture}[auto]
\def\dh{1.5} \def\dv{1}
%
%
%

\def\h{0} \def\v{0}

\node (name) at (\h-0.5,\v+0.5) {$\Qsysa^{\Interval^1_0}:$};

\node[istate] (x0) at (\h,\v-\dv) {$\diamond$};
\node[Sstate] (x1) at (\h+0.5*\dh,\v-\dv) {$\StateLabelYo{1}$};
\node[Sstate] (x2) at (\h+1.5*\dh,\v) {$\StateLabelYo{2}$};
\node[Sstate] (x3) at (\h+1.5*\dh,\v-\dv) {$\StateLabelYo{3}$};
\node[Sstate] (x4) at (\h+1.5*\dh,\v-2*\dv) {$\StateLabelYo{4}$};
%
\SFSAutomatEdge{x0}{\EdgeLabelYt{\ue_1}{\ye_1}}{x1}{}{}
\SFSAutomatEdge{x1}{\EdgeLabelYt{\ue_2}{\ye_2}}{x2}{bend left}{}
\SFSAutomatEdge{x2}{\EdgeLabelYt{\ue_3}{\ye_3}}{x3}{bend left}{}
\SFSAutomatEdge{x3}{\EdgeLabelYt{\ue_2}{\ye_2}}{x2}{bend left}{}
\SFSAutomatEdge{x3}{\EdgeLabelYt{\ue_4}{\ye_4}}{x4}{bend left}{}
\SFSAutomatEdge{x4}{\EdgeLabelYt{\ue_3}{\ye_3}}{x3}{bend left}{}
\SFSAutomatEdge{x1}{\EdgeLabelYt{\ue_4}{\ye_4}}{x4}{bend right}{swap}

\def\h{4.5} \def\v{0}

\node (name) at (\h-0.5,\v+0.5) {$\Qsysa^{\Interval^1_1}:$};

\node[istate] (x1) at (\h,\v-\dv) {$\StateLabelYo{1}$};
\node[Sstate] (x2) at (\h+\dh,\v) {$\StateLabelYo{2}$};
\node[Sstate] (x3) at (\h+\dh,\v-\dv) {$\StateLabelYo{3}$};
\node[Sstate] (x4) at (\h+\dh,\v-2*\dv) {$\StateLabelYo{4}$};
%
\SFSAutomatEdge{x1}{\EdgeLabelYt{\ue_1}{\ye_1}}{x2}{bend left}{}
\SFSAutomatEdge{x2}{\EdgeLabelYt{\ue_2}{\ye_2}}{x3}{bend left}{}
\SFSAutomatEdge{x3}{\EdgeLabelYt{\ue_3}{\ye_3}}{x2}{bend left}{}
\SFSAutomatEdge{x3}{\EdgeLabelYt{\ue_3}{\ye_3}}{x4}{bend left}{}
\SFSAutomatEdge{x4}{\EdgeLabelYt{\ue_4}{\ye_4}}{x3}{bend left}{}
\SFSAutomatEdge{x1}{\EdgeLabelYt{\ue_1}{\ye_1}}{x4}{bend right}{swap}

%
%
%
\end{tikzpicture}
  \end{center}
  \caption{$\Interval^1_0$- and $\Interval^1_1$-abstract state machines of $\Qsyse$ in \REFfig{fig:exp_Qsyse}.}\label{fig:exp_Qsysa1}
 \end{figure}

\begin{enumerate}
 \item[(A1)] $\Qsyse$ is \emph{not} state-based asynch. $1$-complete w.r.t. $\Interval^1_0$:\\
 \eqref{equ:dominolcomplete} does not hold as for $\xe_2$ and $y_1y_4\in\Ds{}{2}$ we have $y_1\in\EnabWl{}{[-1,-1]}{\xe_2}$ but $y_1y_4\notin\EnabWl{}{[-1,0]}{\xe_2}$.
 \item[(A2)] $\Qsyse$ is future unique w.r.t. $\Interval^1_0$ (as this always holds).
%
 \item[(B1)] $\Qsyse$ is \emph{not} state-based asynch. $1$-complete w.r.t. $\Interval^1_1$:\\
\eqref{equ:dominolcomplete} does not hold as for $\xe_1$ and $y_1y_4\in\Ds{}{2}$ we have $y_1\in\EnabWl{}{[0,0]}{\xe_1}$ but $y_1y_4\notin\EnabWl{}{[0,1]}{\xe_1}$.
 \item[(B2)] $\Qsyse$ is future-unique w.r.t. $\Interval^1_1$:\\
 It is easy to see that $\Qsyse$ is output deterministic what immediately implies that $\Qsyse$ is future-unique w.r.t. $\Interval^1_1$ as we chose $\weS=\yeS$.
\end{enumerate}

\noindent Using (A2) and (B2), \REFthm{thm:SimRel_QsyseQsysal} (i) implies that
\begin{subequations}\label{equ:example:Rl1}
 \begin{align}
 \R^{\Interval^1_0}:=&\Set{\Tuple{\xe_1,\diamond}}
 \cup\Set{\Tuple{\xe_2,\ye_1},\Tuple{\xe_2,\ye_3}}
 \cup\Set{\Tuple{\xe_3,\ye_2},\Tuple{\xe_3,\ye_4}}\notag\\
 &\cup\Set{\Tuple{\xe_4,\ye_1},\Tuple{\xe_4,\ye_3}}\cup\Set{\Tuple{\xe_5,\diamond}}~\text{and}\\
   \R^{\Interval^1_1}:=&\Set{\Tuple{\xe_1,\ye_1},\Tuple{\xe_2,\ye_2},\Tuple{\xe_3,\ye_3},\Tuple{\xe_4,\ye_4},\Tuple{\xe_5,\ye_1}}
\end{align}
\end{subequations}
are simulation relations from $\Qsyse$ to $\Qsysa^{\Interval^1_0}$ and $\Qsysa^{\Interval^1_1}$, respectively. It should be noted that every state $\xe_i\in\xeS$ is related via $\R^{\Interval^1_1}$ to its unique output $\Set{\ye_j}=\EnabY{\tre}{\xe_i}$, while $\xe_i\in\xeS$ is related via $\R^{\Interval^1_0}$ to all possible output events $\Qsyse$ might produce immediately \emph{before} reaching $\xe_i$, i.e., the set of $y$-labels of all incoming transitions. 

Using (A1) and (B1) we know from \REFthm{thm:SimRel_QsyseQsysal} (ii), that $\R^{\Interval^1_0}$ (resp.$\R^{\Interval^1_1}$) is not a bisimulation relation between $\Qsyse$ and $\Qsysa^{\Interval^1_0}$ (resp. $\Qsysa^{\Interval^1_1}$). This can be observed from \REFfig{fig:exp_Qsysa1} by choosing $\Tuple{\xe_2,\ye_1}\in\R^{\Interval^1_0}$ and $\Tuple{\ye_1,\Tuple{u_4,y_4},y_4}\in\tra^{\Interval^1_0}$ and observing that $\xe_2$ has no outgoing transition labeled by $\ye_4$. Similarly, we can choose  $\Tuple{\xe_1,\ye_1}\in\R^{\Interval^1_1}$ and $\Tuple{\ye_1,\Tuple{u_1,y_1},y_4}\in\tra^{\Interval^1_1}$ and observe that there actually exists an outgoing transition in $\xe_1$ labeled by $\Tuple{u_1,y_1}$ but this transition reaches state $\xe_2$ which is not related to $\ye_4$ via $\R^{\Interval^1_1}$.


Increasing $l$ and constructing the $\Interval^2_0$- and  $\Interval^2_2$-abstract state machines of $\Qsyse$ using \REFdef{def:QsysalW} yields the state machines $\Qsysa^{\Interval^2_0}$ and $\Qsysa^{\Interval^2_2}$ whose transition structure is depicted in \REFfig{fig:exp_Qsysa2}. It is interesting to note that using more information from the past, i.e., using $\Interval^2_0=[-2,-1]$ instead of $\Interval^1_0=[-1,-1]$, does not render $\Qsyse$ state-based asynchronously $l$-complete.

 \begin{figure}
  \begin{center}
\begin{tikzpicture}[auto]
\def\dh{1.5} \def\dv{1}

\def\h{0} \def\v{0} \def\dh{1.5}

\node (name) at (\h-1,\v+0.5) {$\Qsysa^{\Interval^2_0}:$};

\node[istate] (x0) at (\h,\v-\dv) {$\StateLabelYt{\diamond}{\diamond}$};
\node[istate] (x1) at (\h+0.8*\dh,\v-\dv) {$\StateLabelYt{\diamond}{1}$};
\node[Sstate] (x2a) at (\h+1.2*\dh,\v) {$\StateLabelYt{1}{2}$};
\node[Sstate] (x2b) at (\h+3*\dh,\v) {$\StateLabelYt{3}{2}$};
\node[Sstate] (x3a) at (\h+2*\dh,\v-\dv) {$\StateLabelYt{2}{3}$};
\node[Sstate] (x3b) at (\h+4*\dh,\v-\dv) {$\StateLabelYt{4}{3}$};
\node[Sstate] (x4a) at (\h+1.2*\dh,\v-2*\dv) {$\StateLabelYt{1}{4}$};
\node[Sstate] (x4b) at (\h+3*\dh,\v-2*\dv) {$\StateLabelYt{3}{4}$};

%
\SFSAutomatEdge{x0}{\EdgeLabelYt{\ue_1}{\ye_1}}{x1}{}{}
\SFSAutomatEdge{x1}{\EdgeLabelYt{\ue_2}{\ye_2}}{x2a}{bend left}{xshift=0.2cm}
\SFSAutomatEdge{x1}{\EdgeLabelYt{\ue_4}{\ye_4}}{x4a}{bend right}{swap,xshift=0.2cm}

\SFSAutomatEdge{x2a}{\EdgeLabelYt{\ue_3}{\ye_3}}{x3a}{bend right}{pos=0.3,xshift=-0.2cm}
\SFSAutomatEdge{x2b}{\EdgeLabelYt{\ue_3}{\ye_3}}{x3a}{bend left}{pos=0.5,xshift=-0.2cm}
\SFSAutomatEdge{x3a}{\EdgeLabelYt{\ue_2}{\ye_2}}{x2b}{bend left}{pos=0.8,xshift=0.2cm}
\SFSAutomatEdge{x3a}{\EdgeLabelYt{\ue_4}{\ye_4}}{x4b}{bend right}{swap,pos=0.15,xshift=0.2cm}

\SFSAutomatEdge{x4a.300}{\EdgeLabelYt{\ue_3}{\ye_3}}{x3b.340}{bend right=80,in=320,out=280}{pos=0.05,swap,xshift=0.2cm}
\SFSAutomatEdge{x4b}{\EdgeLabelYt{\ue_3}{\ye_3}}{x3b}{bend left}{pos=0.1,xshift=0.2cm}
\SFSAutomatEdge{x3b}{\EdgeLabelYt{\ue_2}{\ye_2}}{x2b}{bend right}{swap,xshift=-0.2cm}
\SFSAutomatEdge{x3b}{\EdgeLabelYt{\ue_4}{\ye_4}}{x4b}{bend left=10}{xshift=-0.2cm,pos=0.1}


\def\h{0} \def\v{-4} \def\dh{2.3}

\node (name) at (\h-1,\v+0.5) {$\Qsysa^{\Interval^2_2}:$};

\node[istate] (x1) at (\h,\v) {$\StateLabelYt{1}{2}$};
\node[Sstate] (x2) at (\h+\dh,\v) {$\StateLabelYt{2}{3}$};
\node[Sstate] (x3a) at (\h+0.5*\dh,\v-\dv) {$\StateLabelYt{3}{2}$};
\node[Sstate] (x3b) at (\h+1.5*\dh,\v-\dv) {$\StateLabelYt{3}{4}$};
\node[Sstate] (x4) at (\h+\dh,\v-2*\dv) {$\StateLabelYt{4}{3}$};
\node[istate] (x5) at (\h,\v-2*\dv) {$\StateLabelYt{1}{4}$};
\SFSAutomatEdge{x1}{\EdgeLabelYt{\ue_1}{\ye_1}}{x2}{}{}
\SFSAutomatEdge{x2}{\EdgeLabelYt{\ue_2}{\ye_2}}{x3a}{bend left}{pos=0.15,xshift=-0.2cm}
\SFSAutomatEdge{x2}{\EdgeLabelYt{\ue_2}{\ye_2}}{x3b}{bend left}{,xshift=-0.2cm}
\SFSAutomatEdge{x3a}{\EdgeLabelYt{\ue_3}{\ye_3}}{x2}{bend left}{pos=0.15,xshift=0.2cm}
\SFSAutomatEdge{x3b}{\EdgeLabelYt{\ue_3}{\ye_3}}{x4}{bend left}{,xshift=-0.2cm}
\SFSAutomatEdge{x4}{\EdgeLabelYt{\ue_4}{\ye_4}}{x3a}{bend left}{pos=0.85,xshift=0.2cm}
\SFSAutomatEdge{x4}{\EdgeLabelYt{\ue_4}{\ye_4}}{x3b}{bend left=10}{pos=0.6,xshift=0.2cm}
\SFSAutomatEdge{x5}{\EdgeLabelYt{\ue_1}{\ye_1}}{x4}{}{swap}
\end{tikzpicture}
  \end{center}
  \caption{$\Interval^2_0$- and $\Interval^2_2$-abstract state machines of $\Qsyse$ in \REFfig{fig:exp_Qsyse}, where $\langle\kern-1pt{i}\kern-1pt{j}\kern-1pt\rangle:=\ye_i\ye_j$.}\label{fig:exp_Qsysa2}
 \end{figure}

 \begin{enumerate}[(C1)]
 \item $\Qsyse$ is \emph{not} state-based asynch. $2$-complete w.r.t. $\Interval^2_0$:\\
\eqref{equ:dominolcomplete} does not hold as for $\xe_2$ and $\diamond y_1y_4\in\Ds{}{3}$ we have $\diamond y_1\in\EnabWl{}{[-2,-1]}{\xe_2}$ but $\diamond y_1y_4\notin\EnabWl{}{[-2,0]}{\xe_2}$. 
 \item $\Qsyse$ is future unique w.r.t. $\Interval^2_0$ (as this always holds). 
\end{enumerate}

\noindent Contrary, using more information from the future, i.e., using $\Interval^2_2=[0,1]$ instead of $\Interval^1_1=[0,0]$, renders $\Qsyse$ state-based asynchronously $l$-complete. However, in this case the future uniqueness-property is lost.

 \begin{enumerate}[(D1)]
 \item $\Qsyse$ is state-based asynchronously $2$-complete w.r.t. $\Interval^2_2$:\\
 Using more future information actually resolves the ambiguity from $\Interval^1_1$. E.g., choosing $\xe_1$ we can only pick $y_1y_2y_3\in\Ds{}{3}$ to obtain $y_1y_2\in\EnabWl{}{[0,1]}{\xe_1}$, obviously implying $y_1y_2y_3\in\EnabWl{}{[0,2]}{\xe_1}$.
 \item $\Qsyse$ is \emph{not} future-unique w.r.t. $\Interval^2_2$:\\
\eqref{equ:future_unique} does not hold as for $\xe_2$ we have $y_2y_3y_2,~y_2y_3y_4\in\EnabWl{}{\Interval^2_2}{\xe_1}$ but obviously $y_2y_3y_2\neq y_2y_3y_4$. 
\end{enumerate}

\noindent Using \REFthm{thm:SimRel_QsyseQsysal} we can now construct relations $\R^{\Interval^2_0}$ and $\R^{\Interval^2_2}$ analogously to the ones for $l=1$ in \eqref{equ:example:Rl1}. However, now (C1)-(D2) imply that 
\begin{align*}
 \R^{\Interval^2_0}\in\SR{\weS}{}{\Qsyse}{\Qsysa^{\Interval^2_0}}~&\text{but}~
 \BR{\R^{\Interval^2_0}}^{-1}\notin\SR{\weS}{}{\Qsysa^{\Interval^2_0}}{\Qsyse}~\text{and}\\
 \BR{\R^{\Interval^2_2}}^{-1}\in\SR{\weS}{}{\Qsysa^{\Interval^2_2}}{\Qsyse}~&\text{but}~
 \R^{\Interval^2_2}\notin\SR{\weS}{}{\Qsyse}{\Qsysa^{\Interval^2_2}}.
\end{align*}

\noindent To see that $\R^{\Interval^2_2}$ is not a simulation relation from $\Qsyse$ to $\Qsysa^{\Interval^2_2}$ pick $\Tuple{\xe_3,\ye_3\ye_4}\in\R^{\Interval^2_2}$ and $\Tuple{\xe_3,\ue_3,\ye_3,\xe_2}\in\tre$ and observe that $\ye_3\ye_4$ does not have an outgoing transition labeled by $\Tuple{\ue_3,\ye_3}$.

Recall from (D2) that $\Qsyse$ is not future unique for $\Interval^2_2$. Using \eqref{equ:future_unique} this implies that for any interval $\Ilm$ with $m-1\geq 2$ (i.e., any interval with two or more future values) the property of future uniqueness does not hold. 

In terms of state-based asynchronous $l$-completeness the problem is inverted. If we use $m-1<2$ (implying future uniqueness of $\Qsyse$ w.r.t. $\Ilm$ from (A2) and (C2)) $\Qsyse$ cannot be state-based asynchronously $l$-complete for any $l$ as the ambiguity for attaching dominos cannot be resolved by further knowledge about the past. In this case the counterexamples in (A1) and (C1) can be reused by pre-appending the considered strings by an appropriate number of diamonds. It is rather necessary to look at least two steps into the future, i.e., pick $m-1\geq 2$, to resolve this ambiguity as shown in (D1). 

Concluding the above discussion there obviously exists no $l$ and $m$ s.t. $\Qsyse$ in \REFfig{fig:exp_Qsyse} is both state-based asynchronously $l$-complete and future unique w.r.t. $\Ilm$. Therefore, increasing $l$ and $m$ will never result in a bisimilar abstraction of $\Qsyse$.

\section{Quotient-based Abstractions (\QBA)}\label{sec:3:QuotientbasedAbstractions}

The idea of quotient based abstractions (QBA) is to partition the state space $\xeS$ into a finite set of equivalence classes $\yaS$ which is used to define the discrete outputs of the original system as well as states of the abstraction. The set $\yaS$ is usually constructed iteratively, by choosing an initial partition $\Phi^1$ and using the refinement algorithm in \cite{Fernandez1990} which terminates if the partition allows to construct a quotient state machine $\Qsysaq$ which is bisimilar to $\Qsyse$.

\subsection{Incorporating the Partition Refinement Algorithm}
To draw the connection to the setting of \SAlA, we assume that the original system is modelled by \eqref{equ:Prelim:Q} with finite, predefined output set $\weS=\yeS$, and initialize the re-partitioning algorithm in \cite{Fernandez1990} with the partition induced by $\EnabY{\tre}{}$. Using state machines instead of transition systems, we restate this algorithm with slightly modified notation. Therefore, some necessary properties of the resulting partitions are restated from \cite{Fernandez1990} in \REFlem{lem:Phil_Prop}.

\newcommand{\EnabYl}[2]{\ON{H}^{l}_{#1}\ifthenelse{\isempty{#2}}{}{(#2)}} 
\newcommand{\EnabYlr}[2]{\BR{\ON{H}^{l}_{#1}}^{-1}\ifthenelse{\isempty{#2}}{}{(#2)}} 
\newcommand{\yeSl}{\tilde{\yaS}^l} 

\begin{definition}\label{def:Phil}
Given \eqref{equ:Prelim}  and $l\in\Nb$, then\footnote{In \eqref{equ:Phil:l:a} the operator $\textstyle\circ_{a\in A} f_a$ composes all functions $f_a$ with $a\in A$ in any order.}
  \begin{subequations}\label{equ:Phil}
  \begin{align}
   \Phi^1:=&\SetCompX{\EnabYr{\tre}{V}}{V\in\twoup{\yeS}}\label{equ:Phil:0}\\
   \text{and}~\Phi^l:=&~
   \textstyle\circ_{\zeS\in\Phi^{l-1}}\Phi^l_{\zeS}\label{equ:Phil:l:a}\\
   \text{\SUCHTHAT} ~
     \Phi^l_{\zeS}:=&\SetCompX{\INTERSECT{\zeS'}{\EnabTr{\tre}{\zeS}}}{\zeS'\inps\Phi^{l-1}}
     \cup\SetCompX{\SETMINUS{\zeS'}{\EnabTr{\tre}{\zeS}}}{\zeS'\inps\Phi^{l-1}}\label{equ:Phil:l:b}.
 \end{align}
  \end{subequations}
    iteratively defines the $l^{th}$ partition $\Phi^l$ of $\xeS$ w.r.t. $\yeS$.
\end{definition}

\begin{lemma}\label{lem:Phil_Prop}
Given \eqref{equ:Prelim} and $\Phi^l$ as in \REFdef{def:Phil} it holds that 
\begin{subequations}\label{equ:Phil_Prop}
\begin{align}
 &\AllQ{\zeS\in\Phi^{l},\xe,\xe'\in\zeS}{\EnabY{\delta}{\xe}=\EnabY{\delta}{\xe'}},\label{equ:Phil_Prop:a}\\
 &\Phi^l\eqps\SetCompX{\zeS\inps\twoup{\xeS}}{
 \AllQ{\zeS'\inps\Phi^{l-1}\hspace{-0.1cm}}{\hspace{-0.1cm}
  \propImp{\BR{
   \INTERSECT{\zeS}{\EnabTr{\tre}{\zeS'}}\neq\emptyset}\hspace{-0.1cm}}{\hspace{-0.1cm}
   \BR{\zeS\subseteq\EnabTr{\tre}{\zeS'}}}
 }}\hspace{-0.1cm},\label{equ:Phil_Prop:b}
\end{align}
and $\Phi^l$ is a fixed point of \eqref{equ:Phil} if
\begin{align}
 \AllQ{\zeS,\zeS'\inps\Phi^{l}\hspace{-0.1cm}}{\hspace{-0.1cm}
  \propImp{\BR{
   \INTERSECT{\zeS}{\EnabTr{\tre}{\zeS'}}\neq\emptyset}}{
   \BR{\zeS\subseteq\EnabTr{\tre}{\zeS'}}}}.\label{equ:Phil_Prop:c}
\end{align}
  \end{subequations}
\end{lemma}

\begin{proof}
\begin{inparaenum}[(a)]
\item 
 It follows from \cite{Fernandez1990}, Prop.3.9 (i) that 
 \begin{equation}\label{equ:proof:Phil_Prop:a}
 \AllQ{\zeS\in\Phi^{l}}{\ExQ{\zeS'\in\Phi^{l-1}}{\zeS\subseteq\zeS'}}
\end{equation}
Now recall from \eqref{equ:Phil:0} that for all $\zeS^{0}\in\Phi^{0}$ there exists $V\in\twoup{\yeS}$ s.t. $\zeS^{0}=\EnabYr{\tre}{V}$ and $\Phi^{0}$ is a partition. Using \eqref{equ:proof:Phil_Prop:a} we obtain $\zeS\subseteq\EnabYr{\tre}{V}$, what proves the statement.\\
\item It follows from \cite{Fernandez1990}, Prop.3.9(v) that 
 \begin{equation*}
  \AllQ{\zeS'\in\Phi^{l-1},\zeS\in\Phi^{l}}{
  \begin{propDisjA}
   \INTERSECT{\zeS}{\EnabTr{\tre}{\zeS'}}=\emptyset\\
   \zeS\subseteq\EnabTr{\tre}{\zeS'}
  \end{propDisjA}}
 \end{equation*}
 Using that ${\propImp{A}{B}}$ is logical equivalent to ${\propDisj{\neg A}{B}}$ and rewriting the previous statement into set-notation gives \eqref{equ:Phil_Prop:b}.\\
 \item It follows from \cite{Fernandez1990}, Prop.3.10 (iii) that  $\Phi^l$ is a fixed point of \eqref{equ:Phil} if $\Phi^{l+1}=\Phi^{l}$. With this \eqref{equ:Phil_Prop:c} follows from \eqref{equ:Phil_Prop:b}.
\end{inparaenum}
\end{proof}

\begin{proposition}\label{prop:PhilvsEl}
Given \eqref{equ:Prelim}, $\weS\eqps\yeS$ and $\Phi^l$ in \eqref{equ:Phil} it holds that 
 \begin{equation}\label{equ:PhilvsEl}
  \Phi^{l}=\SetCompX{\EnabWlr{}{\Ill}{V}}{V\in\twoup{\BR{\yeS}^{l}}}.
 \end{equation}
\end{proposition}

\begin{proof}
See Appendix~\ref{app:proof:prop:PhilvsEl}.
\end{proof}


Observe that \REFprop{prop:PhilvsEl} implies that the equivalence classes of $\Phi^{l}$ are given by all the sets $V\in\twoup{\BR{\yeS}^{l}}$ of $l$-long dominos which are consistent with the behavior of $\Qsyse$ and the map $\EnabWl{}{\Ill}{}$ is the natural projection map of $\Phi^{l}$ taking a state $\xe\in\xeS$ to its (unique) equivalence class.

\subsection{QBA with Increasing Precision}
Constructing quotient state machines $\Qsysaql$ from every obtained partition $\Phi^{l}$ results in a chain of abstractions with increasing precision, similar to increasing $l$ when constructing \SAlA. Precisely following the construction of \QBA one would first construct an output determinized version of $\Qsys$ with output space $\yqlS=\twoup{\BR{\yeS}^{l}}$ for every $l$ and its \QBA $\Qsysaql$, also having $\yqlS$ as its output space. However, to formally compare the resulting state machines to the realizations of \SAlA using simulation relations or behavioral inclusion requires identical output spaces. We therefore slighly change the definition of \QBA to output values in the set $\yeS$ rather than in $\yqlS$.

\begin{definition}\label{def:Tsysaqbl}
Given \eqref{equ:Prelim} and $\yqlS=\twoup{\BR{\yeS}^{l}}$, define 
\begin{subequations}\label{equ:Qsysaqbl}
 \begin{align}
\xaqlS&\eqps\SetCompX{\ya\in\yqlS}{\ExQ{\xe\in\xeS}{\ya=\EnabWl{}{\Ill}{\xe}}},\label{equ:xaqlS}\\
\xaqlSo{}&\eqps\SetCompX{\ya\in\yqlS}{\ExQ{\xe\in\xeSo{}}{\ya=\EnabWl{}{\Ill}{\xe}}},~\text{and}\label{equ:xaqlSo}\\
\traqbl&\eqps\SetCompX{\hspace{-0,1cm}\Tuple{\xa,\ue,\ye,\xa'}}{
\ExQ{\xe,\xe'\inps\xeS\hspace{-0,1cm}}{\hspace{-0,2cm}
\begin{propConjA}
 \xa\eqps\EnabWl{}{\Ill}{\xe}\\
 \xa'\eqps\EnabWl{}{\Ill}{\xe'}\\
\Tuple{\xe,\ue,\ye,\xe'}\inps\tre
\end{propConjA}}
\hspace{-0,1cm}}\hspace{-0,1cm}.\label{equ:traql}
 \end{align}
\end{subequations}
Then $\Qsysaqbl=(\xaqlS,\ueS\times\yeS,\traqbl,\xaqlSo{})$ is the $l$-th quotient state machine of $\Qsyse$.
\end{definition}

Changing the definition of the output space of \QBA from $\yaS$ to $\yeS$, allows us to show (bi)-similarity of $\Qsyse$ and $\Qsysaqbl$ using the usual relation as, e.g., in \cite{TabuadaBook}, Thm.~4.18.

\begin{theorem}\label{thm:SimRel_QsyseQsysaqbl}
Given \eqref{equ:Prelim} and $\Qsysaqbl$ as in \REFdef{def:Tsysaqbl}, let  
 \begin{equation}
  \R=\SetCompX{\Tuple{\xe,\xa}\in\CARTPROD*{\xeS}{\xaqlS}}{\xa=\EnabWl{}{\Ill}{\xe}}\label{equ:QsyseQsysaqbl:Rq}
\end{equation}
be a relation. Then 
\begin{compactenum}[(i)]
 \item $\R\in\SR{\ueS\times\yeS}{}{\Qsyse}{\Qsysaqbl}$ and 
 \item $\propAequ{\R^{-1}\in\SR{\yeS}{}{\Qsysaqbl}{\Qsyse}}{\text{$\Phi^l$ is a fixed-point of \eqref{equ:Phil}.}}$
\end{compactenum}
\end{theorem}

\begin{proof}
See Appendix~\ref{app:proof:thm:SimRel_QsyseQsysaqbl}.
\end{proof}

It is easy to see, that increasing $l$ gives a tighter abstraction as long as no fixed-point of \eqref{equ:Phil} is reached and 
%
%
%
%
whenever a fixed-point $\Phi^r$ exists, the tightest possible abstraction $\abst{\mathcal{Q}}^{r\qsup}$ will be bisimilar to $\Qsyse$.

\subsection{Example}\label{sec:QBA_new_exp}
We conclude this section by revisiting the example in \REFsec{sec:SAlCA_new_exp}. In particular, we discuss the construction of the quotient state machine of  $\Qsyse$ in  \REFfig{fig:exp_Qsyse}. Recall that $\Qsyse$ is output deterministic, implying that \eqref{equ:TransStruct} holds.
Using \eqref{equ:Phil} the first and second partition of $\xeS$ w.r.t. $\yeS$ are given by
\begin{subequations}\label{equ:exp:Phi}
 \begin{align}
 \Phi^1&=\Set{\Set{\xe_1,\xe_5},\Set{\xe_2},\Set{\xe_3},\Set{\xe_4}}~\text{and}\label{equ:exp:Phi:1}\\
 \Phi^2&=\Set{\Set{\xe_1},\Set{\xe_2},\Set{\xe_3},\Set{\xe_4},\Set{\xe_5}}\label{equ:exp:Phi:2}
\end{align}
\end{subequations}
with 
\begin{subequations}\label{equ:exp:yaS}
 \begin{align}
 \yaS^1\eqps&\Set{\Set{\ye_1},\Set{\ye_2},\Set{\ye_3},\Set{\ye_4}}~\text{and}\label{equ:exp:yaS:1}\\
 \yaS^2\eqps&\left\{\Set{\ye_1\ye_2},\Set{\ye_2\ye_3},\Set{\ye_3\ye_2,\ye_3\ye_4},\Set{\ye_4\ye_3},\Set{\ye_1\ye_4}\right\}\hspace{-0.1cm}.\label{equ:exp:yaS:2}
\end{align}
\end{subequations}
Using \eqref{equ:Phil_Prop:c} we have the following observations.
\begin{enumerate}[(E1)]
 \item $\Phi^1$ is \emph{not} a fixed-point of \eqref{equ:Phil}:\\
 To see that \eqref{equ:Phil_Prop:c} does not hold, pick $\xe_1,\xe_5\in\EnabYr{}{\Set{\ye_1}}$ and observe that $\xe_1\in\EnabTr{\tre}{\Set{\ye_2}}$ and $\xe_5\notin\EnabTr{\tre}{\Set{\ye_2}}$, hence $\INTERSECT{\Set{\ye_1}}{\EnabTr{\tre}{\Set{\ye_2}}}\neq\emptyset$ but $\Set{\ye_1}\not\subseteq\EnabTr{\tre}{\Set{\ye_2}}$.
 \item $\Phi^2$ is a fixed-point of \eqref{equ:Phil}:\\
 As all cells of $\Phi^2$ are singletons \eqref{equ:Phil_Prop:c} trivially holds.
\end{enumerate}

Now constructing the first and second quotient state machine of $\Qsyse$ using \REFdef{def:Tsysaqbl} yields the state machines depicted in \REFfig{fig:exp_Qsysaql}. Using \REFthm{thm:SimRel_QsyseQsysaqbl} (i) we know that 
\begin{subequations}\label{equ:example:Rql}
 \begin{align}
 \R^{1\qsup}:=&\left\{\Tuple{\xe_1,\Set{\ye_1}},\Tuple{\xe_2,\Set{\ye_2}},\Tuple{\xe_3,\Set{\ye_3}},\Tuple{\xe_4,\Set{\ye_4}},\Tuple{\xe_5,\Set{\ye_1}}\right\}~\text{and}\\
 \R^{2\qsup}:=&\left\{\Tuple{\xe_1,\Set{\ye_1\ye_2}},\Tuple{\xe_2,\Set{\ye_2\ye_3}},\Tuple{\xe_3,\Set{\ye_3\ye_2,\ye_3\ye_4}},\Tuple{\xe_4,\Set{\ye_4\ye_3}},\Tuple{\xe_5,\Set{\ye_1\ye_4}}\right\}
\end{align}
\end{subequations}
are simulation relations from $\Qsyse$ to $\Qsysa^{1\qsup}$ and $\Qsysa^{2\qsup}$, respectively. However, \REFthm{thm:SimRel_QsyseQsysaqbl} (ii) implies that only $\R^{2\qsup}$ is a bisimulation relation between $\Qsyse$ and $\Qsysa^{2\qsup}$. To see that this is not true for $\R^{1\qsup}$, we choose  $\Tuple{\xe_1,\Set{\ye_1}}\in\R^{1\qsup}$ and $\Tuple{\Set{\ye_1},\Tuple{u_1,y_1},\Set{y_4}}\in\tra^{1\qsup}$ and observe that there exists only one outgoing transition in $\xe_1$ labeled by $\Tuple{u_1,y_1}$ reaching $\xe_2$, which is not related to $\Set{\ye_4}$ via $\R^{1\qsup}$.

  \begin{figure}
  \begin{center}
\begin{tikzpicture}[auto]
\def\h{0} \def\v{0}

\node (name) at (\h-0.5,\v+0.5) {$\abst{\mathcal{Q}}^{1\qsupb}:$};

\node[istate] (x1) at (\h,\v-1) {$\StateLabelYob{1}$};
\node[Sstate] (x2) at (\h+1.5,\v) {$\StateLabelYob{2}$};
\node[Sstate] (x3) at (\h+1.5,\v-1) {$\StateLabelYob{3}$};
\node[Sstate] (x4) at (\h+1.5,\v-2) {$\StateLabelYob{4}$};
%
\SFSAutomatEdge{x1}{\EdgeLabelYt{\ue_1}{\ye_1}}{x2}{bend left}{}
\SFSAutomatEdge{x2}{\EdgeLabelYt{\ue_2}{\ye_2}}{x3}{bend left}{}
\SFSAutomatEdge{x3}{\EdgeLabelYt{\ue_3}{\ye_3}}{x2}{bend left}{}
\SFSAutomatEdge{x3}{\EdgeLabelYt{\ue_3}{\ye_3}}{x4}{bend left}{}
\SFSAutomatEdge{x4}{\EdgeLabelYt{\ue_4}{\ye_4}}{x3}{bend left}{}
\SFSAutomatEdge{x1}{\EdgeLabelYt{\ue_1}{\ye_1}}{x4}{bend right}{swap}

\def\h{3.5} \def\v{0}

\node (name) at (\h-1,\v+0.5) {$\abst{\mathcal{Q}}^{2\qsupb}:$};

\node[istate] (x1) at (\h,\v) {$\StateLabelYtb{1}{2}$};
\node[Sstate] (x2) at (\h+1.5,\v) {$\StateLabelYtb{2}{3}$};
\node[Sstate] (x3) at (\h+3,\v-1) {$\scriptstyle\{\kern-1pt\StateLabelYt{3}{2}\kern-1pt{,}\StateLabelYt{3}{4}\kern-1pt\}$};
\node[Sstate] (x4) at (\h+1.5,\v-2) {$\StateLabelYtb{4}{3}$};
\node[istate] (x5) at (\h,\v-2) {$\StateLabelYtb{1}{4}$};
\SFSAutomatEdge{x1}{\EdgeLabelYt{\ue_1}{\ye_1}}{x2}{}{}
\SFSAutomatEdge{x2}{\EdgeLabelYt{\ue_2}{\ye_2}}{x3}{bend left}{}
\SFSAutomatEdge{x3}{\EdgeLabelYt{\ue_4}{\ye_4}}{x2}{bend left}{pos=0.9}
\SFSAutomatEdge{x3}{\EdgeLabelYt{\ue_4}{\ye_4}}{x4}{bend left}{}
\SFSAutomatEdge{x4}{\EdgeLabelYt{\ue_3}{\ye_3}}{x3}{bend left}{pos=0.1}
\SFSAutomatEdge{x5}{\EdgeLabelYt{\ue_1}{\ye_1}}{x4}{}{swap}

\end{tikzpicture}
  \end{center}
  \caption{First and second quotient state machines of $\Qsyse$ in \eqref{equ:example:Qsyse}, where $\langle\kern-1pt{i}\kern-1pt{j}\kern-1pt\rangle:=\ye_i\ye_j$.
  }\label{fig:exp_Qsysaql}
 \end{figure}


\section{Comparison between \SAlA and \QBA}\label{sec:4:Comparison}
When it comes to comparing \QBA and \SAlA there are two interesting questions to be asked.
\begin{compactenum}[(i)]
 \item Does $\Qsysaqbl$ realize the unique \SAlA $\Eal=\Tuple{\Nbn,\yeS,\Behal}$ of $\Ee=\Tuple{\Nbn,\yeS,\BeheQ}$?
 \item Can we order the realizations $\Qsysaqbl$,  $\Qsysa^{\Ill}$, and $\Qsysa^{\Interval^l_0}$ in terms of simulation relations for $\weS=\yeS$?
\end{compactenum}

Unfortunately, none of the above statements is true in general. We will therefore derive necessary and sufficient conditions on the structure of $\Qsyse$ for those statements to hold.

\subsection{Comparing $\Qsysaqbl$ and $\Qsysa^{\Ill}$}
We start by giving the only comparing result that holds in general.

\begin{theorem}\label{thm:QsysaqblSAlA}
 Given \eqref{equ:Prelim} and \eqref{equ:BehalW} s.t. $\weS=\yeS$ and $\Qsysaqbl$ as in \REFdef{def:Tsysaqbl}, it holds that\footnote{As before,  $\Behe(\Qsysaqbl)$ denotes the extension of $\projState{\yeS}{\Behf(\Qsysaqbl)}$ to $\Zb$ as discussed in \REFsec{sec:prelim}.} $\Behe(\Qsysaqbl)\subseteq\Behal$.
\end{theorem}

\begin{proof}
 See Appendix~\ref{app:proof:thm:QsysaqblSAlA}.
\end{proof}

As behavioral inclusion is a necessary condition for the existence of a simulation relation from $\Qsysaqbl$ to $\Qsysa^{\Ill}$ (where the latter behavior is given by $\Behal$ from \REFthm{thm:behequ}) the natural next step is to try to find such a relation. However, thinking back to the results in \REFthm{thm:SimRel_QsyseQsysal} (i) and \REFthm{thm:SimRel_QsyseQsysaqbl} (i) there is not much hope for success, as the existence of such a relation would imply that we can also find a simulation relation from $\Qsyse$ to $\Qsysa^{\Ill}$ without the need for future-uniqueness of $\Qsyse$ w.r.t. $\Ill$. Not surprisingly, the latter condition will turn out to be necessary and sufficient for the naturally chosen relation from  $\Qsysaqbl$ to $\Qsysa^{\Ill}$ to be a simulation relation.\\
For the inverse relation to be a simulation relation from $\Qsysa^{\Ill}$ to $\Qsysaqbl$ the following property will turn out to be necessary and sufficient.
\begin{definition}\label{def:domconsist}
 Given \eqref{equ:Prelim} and \eqref{equ:Ds} s.t. $\weS=\yeS$, if 
  \begin{equation}\label{equ:domconsist}
  \AllQ{\zeG\inps\Ds{}{l+1},\ya\inps\yqlS}{
  \propImp{\zeG\ll{0,l-1}\inps\ya}{\ExQ{\xe\inps\EnabWlr{}{\Ill}{\ya}}{\zeG\inps\EnabWl{}{[0,l]}{\xe}}},}
 \end{equation}
 $\Qsyse$ is said to be \emph{domino consistent}.
\end{definition}
Intuitively, domino consistency of $\Qsyse$ implies that whenever a string $\zeG$ is part of an abstract state $\ya$, i.e., $\zeG\in\ya$, any domino $\zeG'\in\Ds{}{l+1}$ that can be attached to $\zeG$ in the domino game, i.e., $\zeG'\ll{0,l-1}=\zeG$, can be attached for this particular abstract state $\ya$, i.e., there exists a transition from $\ya$ to $\ya'$ s.t. $\zeG'\ll{1,l}\in\ya'$. As  $\Qsysa^{\Ill}$ can do all moves of the domino game, it becomes intuitively clear why the condition in \REFdef{def:domconsist} is needed to prove that $\Qsysaqbl$ can simulate $\Qsysa^{\Ill}$.

\begin{theorem}\label{thm:QsysalQsysaqbl}
 Given \eqref{equ:Prelim} s.t. $\weS=\yeS$ and 
 $\Qsysa^{\Ill}$ and $\Qsysaqbl$ as in \REFdef{def:QsysalW} and \REFdef{def:Tsysaqbl}, respectively, let
\begin{equation}
\R=\SetCompX{\Tuple{\zeG,\ya}\in\xaS^{\Ill}\times\xaqlS}{
\zeG\in\ya
}\label{equ:thm:QsysalQsysaqbl:R0}.
\end{equation}
Then
\begin{compactenum}[(i)]
 \item $ \propAequ{\R\in\SR{\yeS}{}{\Qsysa^{\Ill}}{\Qsysaqbl}}{\text{$\Qsyse$ is domino consistent}}$ and
 \item $ \propAequ{\R^{-1}\inps\SR{\yeS}{}{\Qsysaqbl}{\Qsysa^{\Ill}}}{\text{$\Qsyse$ is future unique w.r.t. $\Ill$.}}$
\end{compactenum}
\end{theorem}

\begin{proof}
See Appendix~\ref{app:proof:thm:QsysalQsysaqbl}.
\end{proof}

Combining the results from \REFthm{thm:QsysaqblSAlA} and \REFthm{thm:QsysalQsysaqbl} (i) we have the following answer to our first question.

\begin{corollary}\label{cor:QsysaqblSAlA}
Given \eqref{equ:Prelim} and \eqref{equ:BehalW} s.t. $\weS=\yeS$ and $\Qsysaqbl$ as in \REFdef{def:Tsysaqbl}, 
$\Qsysaqbl$ realizes $\Eal=\Tuple{\Nbn,\yeS,\Behal}$ if $\Qsyse$ is domino consistent.
\end{corollary}
%

Even though, we have only given a sufficient condition in \REFthm{thm:QsysaqblSAlA} it should be noted that this condition is \enquote{almost} necessary in the following sense. The only reason for domino consistency to not be necessary for behavioral equivalence is that for any string $\yeG\in\Behal$ domino consistency is only required for all cells this string passes through. Since, in general, not every string passes though all cells that contain any of is $l$-long pieces, domino consistency is only necessary for the cells which are actually passed, i.e., for \enquote{almost all} cells.\\
To wrap up the comparison, it is interesting to note that future uniqueness of $\Qsyse$ w.r.t. $\Ill$ implies domino consistency and therefore also bisimilarity of $\Qsysa^{\Ill}$ and $\Qsysaqbl$. 

\begin{lemma}\label{lem:fuImpliesDc}
 Let $\Qsyse$ be a state machine satisfying \eqref{equ:TransStruct}, \eqref{equ:ReachLive}, and $\weS=\yeS$. Then 
 \begin{equation*}
  \propImp{\Qsyse~\text{is future unique w.r.t. $\Ill$}}{\Qsyse~\text{is domino consistent.}}
 \end{equation*}
\end{lemma}

\begin{proof}
Using \eqref{equ:future_unique}, future uniqueness of $\Qsyse$ w.r.t. $\Ill$ implies that that for all $\xe\in\xeS$ holds
\begin{equation}\label{equ:proof:EnabWlg1}
 \propImp{\EnabWl{}{\Ill}{\xe}\neq\emptyset}{\length{\EnabWl{}{\Ill}{\xe}}=1}.
\end{equation}
Using \eqref{equ:xaqlS} this immediately implies $\length{\ya}=1$ for all $\ya\in\xaqlS$. 
Now pick $\zeG\in\Ds{}{l+1}$ and $\ya\in\xaqlS$ s.t. $\zeG'=\zeG\ll{0,l-1}\in\ya$, implying $\ya=\Set{\zeG'}$.
As $\zeG\in\Ds{}{l+1}$ we know that there exists $\Tuple{\weG,\xeG}\in\BehS{}$ and $k\in\Nbn$ s.t. $\zeG=\weG\ll{k,k+l}$ and therefore $\zeG\in\EnabWl{}{[0,l]}{\xeG(k)}$ and  $\zeG'\in\EnabWl{}{\Ill}{\xeG(k)}$.
Observe, that this immediately implies $\xeG(k)\in\EnabWlr{}{\Ill}{\ya}$, what proves the statement.
\end{proof}

It is interesting to note that the inverse implication does generally not hold, i.e., domino consistency is a weaker condition. Hence, $\Qsysa^{\Ill}$ might actually be a tighter abstraction than $\Qsysaqbl$ if $\Qsyse$ is not future unique w.r.t. $\Ill$. However, recall from \REFthm{thm:SimRel_QsyseQsysal} that in this case, $\Qsysa^{\Ill}$ does not simulate $\Qsyse$, i.e., might be \enquote{too tight} to suitably abstract $\Qsyse$.
However, if $\Qsyse$ is future unique w.r.t. $\Ill$, $\Qsysa^{\Ill}$ and $\Qsysaqbl$ are actually equivalent up to a trivial renaming of states and the following connections can be drawn between both settings. 


\begin{proposition}\label{prop:Bisim_QsysaqblQsysall}
Given \eqref{equ:Prelim} s.t. $\Qsyse$ is future unique w.r.t. $\Ill$, $\weS=\yeS$, and $\Qsysa^{\Ill}$ and $\Qsysaqbl$ as in \REFdef{def:QsysalW} and \REFdef{def:Tsysaqbl}, respectively, let 
 \begin{equation}\label{equ:cor:Bisim_QsysaqblQsysall}
 \R=\SetCompX{\Tuple{\zeG,V}\in\xaS^{\Ill}\times\xaqlS}{V=\Set{\zeG}}.
 \end{equation}
Furthermore, let $\R^\lsup$ and $\R^{l\qsupb}$ denote the relations defined in \eqref{equ:R0} and  \eqref{equ:QsyseQsysaqbl:Rq}, respectively. Then it holds that 
 \begin{enumerate}[(i)]
 \item $\Qsysaqbl$ realizes $\Ea^l$ w.r.t. $\yeS$,
 \item $\R\in\SR{\yeS}{}{\Qsysa^{\Ill}}{\Qsysaqbl}$ and $\R^{-1}\in\SR{\yeS}{}{\Qsysaqbl}{\Qsysa^{\Ill}}$,
  \item $\propAequ{
  \BR{\R^\lsup}^{-1}\in\SR{\yeS}{}{\Qsysa^{\Ill}}{\Qsyse}
  }{
  \BR{\R^{l\qsupb}}^{-1}\in\SR{\yeS}{}{\Qsysaqbl}{\Qsyse}
  }$, and 
  \item $\propAequ{\text{$\Qsyse$ is state-based asychronoulsy $l$-complete w.r.t. $\Ill$}}{\text{$\Phi^l$ is a fixed-point of \eqref{equ:Phil}}}$.
 \end{enumerate} 
\end{proposition}

\begin{proof}
See Appendix~\ref{app:proof:prop:Bisim_QsysaqblQsysall}.
\end{proof}

\subsection{Comparing $\Qsysaqbl$ and $\Qsysa^{\Interval^l_0}$}
Up until now we have investigated when $\Qsysaqbl$ realizes the \SAlA $\Eal$ and how $\Qsysaqbl$ compares to $\Qsysa^{\Ill}$. However, recall from \REFthm{thm:behequ} that choosing $m=0$, i.e., constructing $\Qsysa^{\Interval^l_0}$ instead of $\Qsysa^{\Ill}$, results in the standard realization of \SAlA. Therefore, we want to conclude our comparison by investigating the connection between $\Qsysaqbl$ and $\Qsysa^{\Interval^l_0}$. For this setting, it is essential to note, that $\Qsyse$ being state-based asychronoulsy $l$-compl. w.r.t. $\Interval^l_l$ does not imply that the latter also holds for $m=0$. Hence, we obtain the following ordering of abstractions by combining the results from Prop.~\ref{prop:Bisim_QsysaqblQsysall} and \REFthm{thm:SimRel_QsysalmpQsysalm}.\smalllb


\begin{corollary}\label{cor:last}
Given the premises of \REFthm{thm:QsysalQsysaqbl} s.t. $\Qsyse$ is future unique w.r.t. $\Ill$, then
 ${\Qsysaqbl}\hgl{\yeS}{}{\Qsysa^{\Ill}}\kgl{\yeS}{}\Qsysa^{\Interval^l_0}$.
\end{corollary}

Even though future uniqueness of  $\Qsyse$ is a very strict requirement, it holds whenever $\Qsyse$ is output deterministic and $l=1$ is chosen. In particular, taking the viewpoint of \QBA and assuming that $\yeS$ can be arbitrarily chosen implies that we can always run the refinement algorithm in \REFdef{def:Phil} first, before applying \QBA and \SAlA. In this case, $\Qsyse$ is obviously output deterministic and choosing $l=1$ is sufficient, leading to bisimilar state machines $\Qsysa^{\Interval^1_1}$ and $\abst{\mathcal{Q}}^{1\qsup}$. However, it should be kept in mind that in this scenario the standard \QBA  $\abst{\mathcal{Q}}^{1\qsup}$ is usually tighter than the standard realization $\Qsysa^{\Interval^1_0}$ of \SAlA in terms of similarity.


\subsection{Example}
We conclude this section by revisiting the example in \REFsec{sec:SAlCA_new_exp} and \REFsec{sec:QBA_new_exp} to compare the abstractions constructed therein. Future uniqueness of $\Qsyse$ w.r.t. $\Ill$ was already investigated in \REFsec{sec:SAlCA_new_exp} and is given by the properties (B2) and (D2) for $l=1$ and $l=2$, respectively. Hence, $\Qsyse$ is future unique w.r.t. $\Interval^1_1$ but not w.r.t. $\Interval^2_2$. Concerning domino consistency, we have the following observations.
\begin{enumerate}[(F1)]
 \item $\Qsyse$ is domino consistent for $l=1$:\\
 Follows from (D1) and \REFlem{lem:fuImpliesDc}.
  \item $\Qsyse$ is domino consistent for $l=2$:\\
  Follows from the fact that every $2$-long string $\zeG\in\Ds{}{2}$ is only contained in one abstract state $\xa\in\xaqlS$. Therefore, \eqref{equ:domconsist} trivially holds.
\end{enumerate}

Now observe that (B2) and \REFprop{prop:Bisim_QsysaqblQsysall} (ii) implies that $\Qsysa^{\Interval^1_1}$ and $\Qsysa^{1\qsup}$ are identical up to the trivial renaming of states given by $\R$ in \eqref{equ:cor:Bisim_QsysaqblQsysall}. This is also obvious by investigating \REFfig{fig:exp_Qsysa1} (right) and \REFfig{fig:exp_Qsysaql} (left). It can be furthermore observed from \eqref{equ:example:Rl1} and \eqref{equ:example:Rql} that
\begin{equation*}
\R^{1\qsup}=\R^{\Interval^1_1}\circ\R
\end{equation*}
with $\R$ from \eqref{equ:cor:Bisim_QsysaqblQsysall}. This is actually always true if $\Qsyse$ is future unique w.r.t. $\Ill$ and was used to prove \REFprop{prop:Bisim_QsysaqblQsysall} (iii)-(iv).

As $\Qsyse$ is not future unique w.r.t. $\Interval^2_2$ from (D2) we cannot apply \REFprop{prop:Bisim_QsysaqblQsysall} for $l=2$. However, it follows from (F2) and Cor.~\ref{cor:QsysaqblSAlA} that $\Qsysa^{2\qsup}$ realizes the \SAlA $\Ea^l$ w.r.t. $\yeS$. This implies that for this particular example using \QBA yields a bisimilar abstraction of $\Qsyse$ (from (E2) and \REFthm{thm:SimRel_QsyseQsysaqbl}) which is a realization of the \SAlA $\ElaMax$ w.r.t. $\yeS$. However, this realization cannot coincide with any abstract state machine $\Qsysal$ as one of its states is given by a set $ V\in\twoup{\BR{\yeS}^l}$ with $|V|>1$. 
In particular, as $\Qsysa^{\Interval^2_0}$ only simulates $\Qsyse$ but is not bisimilar to the latter (from (C1) and \REFthm{thm:SimRel_QsyseQsysal}) and $\Qsysa^{\Interval^2_2}$ is only simulated by $\Qsyse$ but not vice versa (from (D2) and \REFthm{thm:SimRel_QsyseQsysal}) we have the following (strict) ordering of abstractions
\begin{equation*}
 \Qsysa^{\Interval^2_2}\kg{\yeS}\Qsysa^{2\qsup}\kg{\yeS}\Qsysa^{\Interval^2_0}\quad\text{with}\quad\Qsysa^{2\qsup}\hg{\yeS}\Qsyse.
\end{equation*}

%

\subsection{Some Comments on Control and Future Research}\label{sec:Compare:Control}
In this section we have compared finite state machine abstractions resulting from \SAlA and \QBA using the notion of simulation relations. The construction of those abstractions is usually motivated by a control problem involving a finite set of output symbols. To use the obtained comparison results it would therefore be interesting to investigate the usability of the constructed abstractions for control purposes.
Unfortunately, given the different settings of \SAlA and \QBA, the controller synthesis techniques applied in the literature also differ significantly, as they are usually tailored to the respective setting. Due to space limitations we are therefore not aiming at a profound comparison of the latter, but rather emphasize some observations from the example.

In the literature on \QBA so called \emph{alternating simulation relations} are used to evaluate if an abstraction is suitable for control (see e.g. \cite[Def.4.19]{TabuadaBook} for a formal definition). 
It is interesting to note that for any choice of $l$ and $m$ the inverse relation $\R^{-1}$ of $\R$ in \REFthm{thm:SimRel_QsyseQsysal} (resp. \REFthm{thm:SimRel_QsyseQsysaqbl}) is an alternating simulation relation from $\Qsysal$ (resp. $\Qsysaqbl$) to $\Qsyse$ iff $\R$ is a simulation relation from $\Qsyse$ to $\Qsysal$ (resp. $\Qsysaqbl$) and\footnote{$\EnabU{\tre}{\xe}:=\SetCompX{\ue\inps\ueS}{\ExQ{\ye\inps\yeS,\xe'\inps\xeS}{\Tuple{\xe,\ue,\ye,\xe'}\inps\tre}}$} 
\begin{equation}\label{equ:CondEnabU}
 \AllQ{\Tuple{\xe,\xa}\in\R}{\EnabU{\tra}{\xa}\subseteq\EnabU{\tre}{\xe}}.
\end{equation}
Hence, the abstraction must simulate $\Qsyse$ to be suitable for controller synthesis in the setting of \QBA.

In our example we have shown in \REFsec{sec:SAlCA_new_exp} that $\R^{\Interval^2_2}$ is not a simulation relation from $\Qsyse$ to $\Qsysa^{\Interval^2_2}$. Intuitively, this results from the observation that the abstraction has to \enquote{guess} non-deterministically when observing $\ye_2$ to which state to move to be able to \enquote{follow} the future evolution of $\Qsyse$. Interestingly, it can be observed in \REFfig{fig:exp_Qsysa1} (right) that the abstraction $\Qsysa^{\Interval^1_1}$ also needs to decide to either move to $\ye_2$ or $\ye_4$ from $\ye_1$ when observing the output $\ye_1$. 
However, it was shown in \REFsec{sec:SAlCA_new_exp} that $\R^{\Interval^1_1}$ is a simulation relation from $\Qsyse$ to $\Qsysa^{\Interval^1_1}$ and it can be easily observed that \eqref{equ:CondEnabU} holds for $\R^{\Interval^1_1}$ in \eqref{equ:example:Rl1}. Hence, $\Qsysa^{\Interval^1_1}$ is suitable for control in terms of alternating simulation relations, while $\Qsysa^{\Interval^2_2}$ is not. 
Intuitively, this is due to the fact that, using simulation relations, it is implicitly assumed that the abstraction \enquote{knows} to which state the original system moves. Therefore, $\Qsysa^{\Interval^1_1}$ can observe if $\Qsyse$ moves to $\xe_2$ or $\xe_4$ and can then pick the \enquote{right} state, i.e., the related one. Contrary, the states $\ye_3\ye_2$ and $\ye_3\ye_4$ of $\Qsysa^{\Interval^2_2}$ are related to the same state $\xe_3$. Therefore, knowing that $\Qsyse$ moves to $\xe_3$ does not help to decide which state to pick in $\Qsysa^{\Interval^2_2}$ when observing $\ye_2$.

The previous argument obviously
only works if the abstraction has full state information from the original system when \enquote{simulating} its moves. However, in the setting for \SAlA the controller (which is designed based on the abstraction and therefore usually given as a sub-machine of the latter) can only interact with the system through the (predefined) set of output symbols $\yeS$. As the state space of $\Qsyse$ is usually infinite while $\yeS$ is finite, this usually implies that no full state feedback is available. Intuitively, one would therefore need to require that non-determinism in $\Qsysa$ can be resolved without full state information in the setting of \SAlA. 

This issue was recently discussed in \cite{ReissigRungger2014} where it is shown that alternating simulation relations are not sufficient for abstraction based control if no full state feedback is available. To overcome this issue \cite{ReissigRungger2014} suggests \emph{feedback refinement relations} for a particular class of transition systems which allow for abstraction based controller synthesis using a predefined set of output events. As applying these ideas to the abstractions constructed in this paper would require a non-trivial extension of the relations in \cite{ReissigRungger2014}, we postpone this idea to future work.

However, even without this formal extension, we can draw the following conclusions from the construction of $\Ilm$-abstract state machines in \REFdef{def:QsysalW}. Observe, that choosing $m=0$, i.e., considering the original realization of \SAlA, will always result in a deterministic state machine, i.e., observing an output $\ye\in\yeS$ fully determines the next state of the abstraction $\Qsysa^{\Interval^l_0}$. Hence, the issue of unresolved non-determinism discussed above cannot occur in $\Qsysa^{\Interval^l_0}$ and $\R$ in \eqref{equ:example:Rl1} is always a simulation relation from $\Qsyse$ to $\Qsysa^{\Interval^l_0}$. Nevertheless, \eqref{equ:CondEnabU} still needs to hold to allow for an alternating simulation relation. For $\Qsysa^{\Interval^1_0}$ and $\Qsysa^{\Interval^2_0}$ constructed in
\REFsec{sec:SAlCA_new_exp} the latter is unfortunately not true as, e.g., $\EnabU{\tre}{\xe_2}=\Set{\ue_2}\subseteq\EnabU{\tra^{\Interval^1_0}}{\ye_1}=\Set{\ue_2,\ue_4}$ and $\Tuple{\xe_2,\ye_1}\in\R^{\Interval^l_0}$ (from \eqref{equ:example:Rl1}).

Interestingly, this observation draws a nice connection to the conditions for controller synthesis using \SlA in \cite{MoorRaischYoung2002}. Therein, the original system is required to have a \emph{free input}, i.e., 
\begin{equation}\label{equ:freeinput}
 \AllQ{\xe\in\xeS}{\EnabU{\tre}{\xe}=\ueS}.
\end{equation}
As \eqref{equ:freeinput} always implies \eqref{equ:CondEnabU}, assuming a free input implies that $\R^{\Interval^l_0}$ is an alternating simulation relation and no state information is needed for the abstraction to simulate the moves of the original system.

Using these insights it would be interesting to investigate which conditions on $\Qsyse$ allow for control based on a predefined set of input and output symbols using abstract state machines $\Qsysal$ with $m>0$ and quotient state machines $\Qsysaqbl$. As we have shown that increasing $m$ results in tighter abstractions this could be beneficial if $\Qsysa^{\Interval^l_0}$ is not tight enough for a particular controller synthesis problem and increasing $l$ does not refine the abstraction sufficiently.

\section{Conclusion}
In this paper we have compared finite state machine abstractions resulting from \SAlA and \QBA. For this purpose we have introduced a new parameter $m\in[0,l]$ to realize \SAlA by different state machines. We have shown that the choice $m=0$ corresponds to relating states in the original state machine $\Qsyse$ to their \emph{strict $l$-long past} of external symbols, reproducing the standard realization of \SAlA. On the other hand, choosing $m=l$ corresponds to relating states in the original state machine $\Qsyse$ to their \emph{$l$-long future} of external symbols. We have shown that this construction of realizations for \SAlA is closely related to the construction of \QBA, if the latter is obtained from a partition resulting from $l$ steps of the usual repartitioning algorithm.\\
Even if the latter observation renders both methods conceptually similar, we could show that they are generally incomparable. Only in the special case where the original system is future unique both abstractions are identical up to a renaming of states.


\appendices
\section{Proofs}
In this appendix we provide all remaining proofs.

\subsection{Preliminaries}
Given two signals $\wG_1,\wG_2\in\BR{\wS}^{\Nbn}$ and two time instants $t_1,t_2\in\Nbn$, their \textit{concatenation}  $\wG_3=\CONCAT{\wG_1}{t_1}{t_2}{\wG_2}$ is defined by
\begin{equation}\label{equ:concat}
\AllQ{t\in\Nbn}{\wG_3(t)=\DiCases{\wG_1(t)}{t< t_1}{\wG_2(t-t_1+t_2)}{t\geq t_1}}.
\end{equation}
Using this definition of concatenation, it can be shown (see e.g. \cite{SchmuckRaisch2014_ControlLetters}, Prop.2) that for the full behavior $\Behf(\Qsyse)$ of a state machine $\Qsyse$ defined in \eqref{equ:Behtr} the state property holds, i.e.,
 \begin{equation}\label{equ:StateSpaceDynamicalSystem:2}
 \AllQSplit{\Tuple{\ueG,\yeG,\xeG},\Tuple{\ueG',\yeG',\xeG'}\in\Behf(\Qsyse), k,k'\in\Nbn}{\propImp{\xeG(k)=\xeG'(k')}{\CONCAT{\Tuple{\ueG,\yeG,\xeG}}{k}{k'}{\Tuple{\ueG',\yeG',\xeG'}}\in\Behf(\Qsyse)}.}
\end{equation}
It is easy to see that \eqref{equ:StateSpaceDynamicalSystem:2} equivalently holds for the extension of $\Behf(\Qsyse)$ to $\Zb$ as discussed in \REFsec{sec:prelim}.\\
To simplify the subsequent proofs we now translate the conditions for a transition in $\Qsysal$ which were given in \REFdef{def:QsysalW} in terms of transitions of $\Qsyse$ into conditions of the domino-game.
 \begin{lemma}\label{lem:tralversusDs}
Given \eqref{equ:Prelim} and $\Qsysal$ as in \REFdef{def:QsysalW} it holds for all $\xa,\xa'\in\xaS$, $\ue\in\ueS$ and $\ye\in\yeS$ that
\begin{subequations}\allowdisplaybreaks
 \begin{align}
  &\Tuple{\xa,\ue,\ye,\xa'}\inps\tral\notag\\
  &\Leftrightarrow
   \begin{propConjA}
  \xa'\ll{0,l-m-1}\eqps\BR{\xa\ll{0,l-m-1}\sconcps\projState{\weS}{\ue,\ye}}\hspace{-0.1cm}\ll{1,l-m}\\
~\xa\ll{l\mips m,l\mips 1}=\BR{\projState{\weS}{\ue,\ye}\sconcps\xa'\ll{l\mips m,l\mips 2}}\ll{0,m-1}\\
  \xa|_{[0,l-m-1]}\sconcps\we\sconcps\xa'|_{[l-m,l-1]}\in\Ds{}{l+1}
     \end{propConjA}\label{equ:tralversusDs}\\
  &\Leftrightarrow
  \ExQ{\zeG\in\Ds{}{l+1}}{
   \begin{propConjA}
  \zeG\ll{0,l-1}=\xa\\
  \zeG\ll{1,l}=\xa'\\
  \we=\zeG(l-m)
   \end{propConjA}
   },\label{equ:tralversusDs:b}
 \end{align}
 \end{subequations}
 where $\we=\projState{\weS}{\ue,\ye}$.
\end{lemma}

\begin{proof}
 \begin{inparaitem}
  \item[\enquote{$\Rightarrow$}:]\\
   \item Observe from \eqref{equ:tral} that $\Tuple{\xa,\ue,\ye,\xa'}\in\tral$ iff the first two lines of the conjunction in \eqref{equ:tral} (resp. \eqref{equ:tralversusDs}) are fulfilled and 
 there exist $\xe,\xe' \in\xeS$ s.t. 
\begin{subequations}\label{equ:proof:behequ}
  \begin{align}\allowdisplaybreaks
   &\ExQ{\Tuple{\weG,\xeG}\inps\BeheSQ,k}{
   \begin{propConjA}
    \xeG(k)=\xe\\
    \xa=\weG|_{[k+m-l,k+m-1]}
   \end{propConjA}},\label{equ:proof:behequ:a}\\
    &\ExQ{\Tuple{\weG',\xeG'}\inps\BeheSQ,k'}{
   \begin{propConjA}
    \xeG'(k')=\xe'\\
    \xa'=\weG'|_{[k'+m-l,k'+m-1]}
   \end{propConjA}},\label{equ:proof:behequ:b}\\
   &\ExQ{\Tuple{\weG'',\xeG''}\inps\BeheSQ,k''}{
   \begin{propConjA}
    \xeG''(k'')=\xe\\
    \weG''(k'')=\we\\
    \xeG''(k''+1)=\xe'
   \end{propConjA}}.\label{equ:proof:behequ:c}
  \end{align}
 \item Now observe from \eqref{equ:proof:behequ} that $\xeG(k)=\xe=\xeG''(k'')$ and $\xeG'(k')=\xe'=\xeG''(k''+1)$. Using \eqref{equ:StateSpaceDynamicalSystem:2} we therefore obtain
 \begin{equation}\label{equ:proof:behequ:2:a}
  \Tuple{\tilde{\weG},\tilde{\xeG}}=\CONCAT{\Tuple{\weG,\xeG}}{k}{k''}{\CONCAT{\Tuple{\weG'',\xeG''}}{k''+1}{k'}{\Tuple{\weG',\xeG'}}}\in\BeheSQ
 \end{equation}
  \end{subequations}
giving
 \begin{align*}
  \tilde{\weG}|_{[k+m-l,k+m]}
  &=\weG|_{[k+m-l,k-1]}\sconc\weG''(k'')\sconc\weG'|_{[k',k'+m-1]}\\
  &=\xa|_{[0,l-m-1]}\sconc\we\sconc\xa'|_{[l-m,l-1]}\\
  &\in\BeheQ|_{[k+m-l,k+m]}\subseteq\Ds{}{l+1},
 \end{align*}
 hence \eqref{equ:tralversusDs} holds.\\
\item Now let $\zeta=\tilde{\weG}|_{[k+m-l,k+m]}\in\Ds{}{l+1}$ and observe that $\we=\zeG(l-m)$. With this choice of $\zeta$ the first two lines of the conjunction in \eqref{equ:tral} (resp. \eqref{equ:tralversusDs}) immediately imply $\zeG\ll{0,l-1}=\xa$ and $\zeG\ll{1,l}=\xa'$, hence \eqref{equ:tralversusDs:b} holds.\\
\item[\enquote{$\Leftarrow$}:]\\
\item Pick $\zeG\in\Ds{}{l+1}$ and $\ue,\ye,\xa$ and $\xa'$ s.t. the right side of \eqref{equ:tralversusDs:b} holds.\\
\item It is easy to see that the first two lines of the conjunction in \eqref{equ:tral} (resp. \eqref{equ:tralversusDs})  hold with this choice and $\zeta=\xa|_{[0,l-m-1]}\sconc\we\sconc\xa'|_{[l-m,l-1]}\in\Ds{}{l+1}$, hence \eqref{equ:tralversusDs} holds.\\
\item Using \eqref{equ:Ds} there exist $\Tuple{\tilde{\weG},\tilde{\xeG}}$ and $\tilde{k}\in\Nbn$ s.t. $\zeG=\tilde{\weG}\ll{\tilde{k}+m-l,\tilde{k}+m-1}$.
We can therefore choose all signals in \eqref{equ:proof:behequ:a} and \eqref{equ:proof:behequ:b} equivalent to $\Tuple{\tilde{\weG},\tilde{\xeG}}$ and $\xe=\tilde{\xeG}(\tilde{k})$ as well as $\xe'=\tilde{\xeG}(\tilde{k}+1)$, giving $\xa\in\EnabWl{}{\Ilm}{\xe}$, $\xa'\in\EnabWl{}{\Ilm}{\xe}$ and $\Tuple{\xe,\ue,\ye,\xe'}\in\tre$. With this the last line of the conjunction in \eqref{equ:tral} holds, hence $\Tuple{\xa,\ue,\ye,\xa'}\in\tral$.
 \end{inparaitem}
\end{proof}

\subsection{Proof of \REFthm{thm:behequ}}\label{proof:thm:behequ}
%
\begin{inparaitem}
 \item[1.)] Show\footnote{As before,  $\Behe(\Qsysal)$ denotes the extension of $\projState{\weS}{\Behf(\Qsysal)}$ to $\Zb$.}  ${\Behal\subseteq\Behe(\Qsysal)}$:\\
 Pick $\waG\in\Behal$, $\uaG,\yaG$ s.t. $\projState{\weS}{\uaG,\yaG}=\waG\llr{0,\infty}$ and $\xaG$ s.t. $\AllQ{k\in\Nbn}{\xaG(k)=\waG\ll{k-l+m,k+m-1}}$. 
 To show the first line of \eqref{equ:Behtr}, recall that for all $k<0$ and $\weG\in\BeheQ$ we have $\weG(k)=\diamond$. Therefore, \eqref{equ:BehalW} and \eqref{equ:Ds} imply $\waG|_{[m-l,m]}\in\BeheQ\ll{m-l,m}$. Hence, there exists $\Tuple{\weG',\xeG'}\inps\BeheSQ$ s.t. $\weG'|_{[m-l,m]}=\waG|_{[m-l,m]}$ and therefore $\waG|_{[m-l,m-1]}\in\EnabWl{}{\Ilm}{\xeG'(0)}$ with $\xeG'(0)\in\xeSo{}$ (from \eqref{equ:Xxr}), hence $\xaG(0)\in\xalSo{}$ (from \eqref{equ:xalSo}).
  The second line of \eqref{equ:Behtr} follows from the choice of $\xaG$ and \eqref{equ:tralversusDs:b}, as \eqref{equ:BehalW} implies $\AllQ{k\in\Nbn}{\waG\ll{k-l+m,k+m}\in \Ds{}{l+1}}$.\\
\item[2.)] Show $\Behe(\Qsysal)\subseteq\Behal$:\\
Pick $\Tuple{\uaG,\yaG,\xaG}\in\Behf(\Qsysal)$ and $\waG$ s.t. $\projState{\weS}{\uaG,\yaG}=\waG\llr{0,\infty}$ and $\AllQ{k<0}{\waG(k)=\diamond}$.
To show $\waG\in\Behal$, observe that the second line of \eqref{equ:BehalW} follows directly from \eqref{equ:tralversusDs:b} and the second line of \eqref{equ:Behtr}, if we pick $\xaG$ accordingly. Using $\AllQ{k<0}{\waG(k)=\diamond}$ the second line of \eqref{equ:BehalW} immediately implies the first. 
\end{inparaitem}

\subsection{Proof of \REFthm{thm:Qsysalo_equ}}\label{proof:thm:Qsysalo_equ}

First observe from \eqref{equ:Xxr} and \eqref{equ:xalS} that
 \begin{equation}\label{equ:proof:xalSDs}
  \xalS=\bigcup_{k\in\Nbn}\BeheQ\ll{k-l+m,k+m-1}\subseteq\UNION{\Set{\diamond}^l}{\Ds{}{l}}
 \end{equation}
where equality only holds for $m=0$. Using \eqref{equ:BehalW:b} we therefore have $\xaS^{\Interval^l_0}=\UNION{\Set{\diamond}^l}{\Pi_{l}(\Behal)}=\xaS^l$. Furthermore, observe from \eqref{equ:Xxr} that $\xaSo{}^{\Interval^l_0}=\BeheQ\ll{-l,-1}=\Set{\diamond}^l=\xaSo{}^l$. Finally, it follows from \REFlem{lem:tralversusDs} that $\Tuple{\xa,\ue,\ye,\xa'}\in\tra^{\Interval^l_0}$ iff
$\xa'=\BR{\xa\ll{0,l-1}\sconc\we}\ll{1,l}=\BR{\xa\sconc\we}\ll{1,l}$ and $\xa|_{[0,l-1]}\sconc\we=\xa\sconc\we\in\Ds{}{l+1}=\Pi_{l+1}(\Behal)$, what proves the statement.

\subsection{Proof of \REFthm{thm:SimRel_QsyseQsysal}}\label{proof:thm:SimRel_QsyseQsysal}



\begin{inparaitem}
 \item[(i)] To see that \eqref{equ:SimRel_EL:a} always holds for $\R$, observe that for all  $\xe\in\xeSo{}$ there exists $\zeG\in\EnabWl{}{\Ilm}{\xe}$ (from \eqref{equ:ReachLive:0} and \eqref{equ:Xxr}), hence $\zeG\in\xalSo{}$ (from \eqref{equ:xalSo}). It remains to show that \eqref{equ:SimRel_EL:b} holds for $\R$ and $\ueS\times\yeS$ iff $\Qsyse$ is future unique w.r.t. $\Ilm$. We show both directions separately:\\
 \item[\enquote{$\Leftarrow$}]:\\ 
 Pick $\Tuple{\xe,\xa}\in\Rn$, i.e., $\xa\in\EnabWl{}{\Ilm}{\xe}$ and $\ue,\ye,\xe'$ s.t. $\Tuple{\xe,\ue,\ye,\xe'}\in\tre$. Then it follows from \eqref{equ:Xxr} and \eqref{equ:Behtr} that \eqref{equ:proof:behequ:a} and \eqref{equ:proof:behequ:c} holds, hence 
   \begin{align}\label{equ:proof:wxtilde}
  &\Tuple{\tilde{\weG},\tilde{\xeG}}=\CONCAT{\Tuple{\weG,\xeG}}{k}{k''}{\Tuple{\weG'',\xeG''}}\in\BeheSQ.
   \end{align}
Now pick $\xa'=\tilde{\weG}|_{[k+m-l+1,k+m]}$ and observe that $\tilde{\xeG}(k+1)=\xeG''(k+1)=\x'$ implies $\xa'\in\EnabWl{}{\Ilm}{\xe'}$. Furthermore, $\tilde{\xeG}(k)=\xeG(k)=\x$ implies $\tilde{\weG}|_{[k+m-l,k+m-1]}\in\EnabWl{}{\Ilm}{\xe}$. Now observe that $\tilde{\weG}|_{[k+m-l,k-1]}=\weG|_{[k+m-l,k-1]}$. Therefore using \eqref{equ:future_unique} implies $\xa=\weG|_{[k+m-l,k+m-1]}=\tilde{\weG}|_{[k+m-l,k+m-1]}$, hence $\xa\ll{1,l-1}=\xa'\ll{0,l-2}$. With this and $\tilde{\weG}(k)=\projState{\weS}{\ue,\ye}$ follows that \eqref{equ:tral} holds, hence $\Tuple{\xa,\ue,\ye,\xa'}\in\tral$. \\
%
 \item[\enquote{$\Rightarrow$}]:\\
  \begin{inparaitem}[-]
 \item Pick $\xe\in\xeS$, $m>0$ (as  $\Qsyse$ is always future unique w.r.t. $\Interval^l_0$) and $\zeta,\zeta'\in\EnabWl{}{\Ilm}{\xe}$. Using \eqref{equ:Xxr} there exist $\Tuple{\weG,\xeG}\inps\BeheSQ$ and $k$ s.t. $\xe=\xeG(k)$ and $\zeta'=\weG\ll{k+m-l,k+m-1}$. \\
 \item Now pick $\xe'=\xeG'(k+1)$ and $\ue,\ye$ s.t. $\projState{\weS}{\ue,\ye}=\weG(k)=\zeta'(l-m)$ and observe that $\Tuple{\xe,\ue,\ye,\xe'}\in\tre$  and
   $\tilde{\zeta}=\zeta'\ll{1,l-1}\sconc\weG'(k+m)\in\EnabWl{}{\Ilm}{\xe'}$.\\
  \item As $\R\in\SR{\ueS\times\yeS}{}{\Qsyse}{\Qsysal}$, there exists $\zeta''\in\EnabWl{}{\Ilm}{\xe'}$ s.t. $\Tuple{\zeta,\ue,\ye,\zeta''}\in\tral$. Using \eqref{equ:tral} this implies that $\zeta(l-m)=\projState{\weS}{\ue,\ye}$, hence $\zeta(l-m)=\zeta'(l-m)$, and $\zeta\ll{1,l-1}=\zeta''\ll{0,l-2}$.\\
  \item As $\tilde{\zeta},\zeta''\in\EnabWl{}{\Ilm}{\xe'}$ we can apply the same reasoning as before (substituting $\xe$ by $\xe'$ and $\zeta,\zeta'$ by $\tilde{\zeta},\zeta''$) and immediately obtain $\zeta''(l-m)=\tilde{\zeta}(l-m)=\zeta'(l-m+1)$. As $\zeta\ll{1,l-1}=\zeta''\ll{0,l-2}$ we therefore have $\zeta(l-m+1)=\zeta'(l-m+1)$. Applying this process iteratively therefore yields $\zeta\ll{l-m,l-1}=\zeta'\ll{l-m,l-1}$, what proves the statement.\\
\end{inparaitem}

 \item[(ii)] First observe that \eqref{equ:SimRel_EL:a} always holds for $\R^{-1}$ as we can pick $\xa\in\xalSo{}$ and \eqref{equ:xalSo} implies the existence of $\xe\in\xeSo{}$ s.t. $\zeG\in\EnabWl{}{\Ilm}{\xe}$. It remains to show that \eqref{equ:SimRel_EL:b} holds for $\R^{-1}$ and $\yeS$ iff $\Qsyse$ is state-based asynchronously $l$-complete w.r.t. $\Ilm$. We show both directions separately:\\
 \item[\enquote{$\Leftarrow$}]:\\
  \begin{inparaitem}[-]
 \item  Pick $\Tuple{\xa,\xe}\in\R^{-1}$, i.e., $\xa\in\EnabWl{}{\Ilm}{\xe}$ and $\ue,\ye,\xa'$ s.t. $\Tuple{\xa,\ue,\ye,\xa'}\in\tral$. Then it follows from \eqref{equ:tralversusDs:b} that there exists $\zeG\in\Ds{}{l+1}$ s.t. $\zeG\ll{0,l-1}=\xa$, $\zeG\ll{1,l}=\xa'$ and $\projState{\yeS}{\ue,\ye}=\zeG(l-m)$. Using \eqref{equ:dominolcomplete} this implies that $\zeG\in\EnabWl{}{[m-l,m]}{\xe}$, hence there exists $\Tuple{\weG,\xeG}\in\BeheSQ$ and $k\in\Nbn$ s.t. $\xe=\xeG(k)$ and $\zeta=\weG\ll{k+m-l,k+m}$. \\
 \item Now pick $\xe'=\xeG(k+1)$ and observe that $\xa'\in\EnabWl{}{\Ilm}{\xe'}$, hence $\Tuple{\xa',\xe'}\in\R^{-1}$. Furthermore, \eqref{equ:Behtr} implies the existence of $\ue'$ and $\ye'$ s.t. $\Tuple{\xe,\ue',\ye',\xe'}\in\tre$ and $\projState{\weS}{\ue',\ye'}=\we$, what proves the statement.\\
  \end{inparaitem}
 \item[\enquote{$\Rightarrow$}]:\\ 
 \begin{inparaitem}[-]
 \item  Pick $\xe\in\xeS$ and $\zeG\in\Pi_{l+1}(\BeheQ)$ s.t. $\zeG\ll{0,l-1}\in\EnabWl{}{\Ilm}{\xe}$. Furthermore define $\xa=\zeG\ll{0,l-1}$ and $\xa'=\zeG\ll{1,l}$ and $\we=\zeG(l-m)$ and observe from \eqref{equ:tralversusDs:b} that there exists $\ue$ and $\ye$ s.t. $\Tuple{\xa,\ue,\ye,\xa'}\in\tral$ and $\projState{\weS}{\ue,\ye}=\we$ and observe that $\Tuple{\xa,\xe}\in\R^{-1}$.\\
 \item As $\R^{-1}\in\SR{\weS}{}{\Qsysal}{\Qsyse}$ we know that there exist $\xe'$ and $\ue'$ and $\ye'$ s.t. $\Tuple{\xe,\ue',\ye',\xe'}\in\tre$, $\projState{\weS}{\ue',\ye'}=\we$ and $\xa'\in\EnabWl{}{\Ilm}{\xe'}$.\\
 \item With this we know that \eqref{equ:proof:behequ} holds, hence $\zeG=\tilde{\weG}\ll{k+m-l,k+m}$ and $\xe=\tilde{\xeG}(k)$ and therefore $\zeG\in\EnabWl{}{[m-l,m]}{\xe}$.
 \end{inparaitem}
  \end{inparaitem}

\subsection{Proof of \REFthm{thm:SimRel_QsysalpmQsysalm}}\label{proof:thm:SimRel_QsysalpmQsysalm}

\begin{lemma}\label{lem:NewProperty2}
Given \eqref{equ:Prelim} let
 \begin{equation}\label{equ:NewProperty2}
  \AllQ{\zeG\in\BR{\UNION{\Set{\diamond}}{\weS}}^{l+2}}{
  \propImp{
  \begin{propConjA}
   \zeG\ll{0,l}\in\UNION{\Set{\diamond}^l}{\Ds{}{l+1}}\\
   \zeG\ll{1,l+1}\in\Ds{}{l+1}\\
  \end{propConjA}
  }{\zeG\in\Ds{}{l+2}}.}
 \end{equation}
Then
\begin{equation*}
 \propAequ{\text{\eqref{equ:NewProperty2}}}{\Behal=\Beha^{l+1}}.
\end{equation*}
\end{lemma}

\begin{proof}
 \begin{inparaitem}
  \item[\enquote{$\Rightarrow$}]: Recall that $\Beha^{l+1}\subseteq\Beha^{l}$ always holds. To prove $\Behal\subseteq\Beha^{l+1}$ we pick $\weG\in\Behal$ and recall from \eqref{equ:BehalW} that $\AllQ{k\inps\Nbn}{\wG\ll{k-l,k}\inps \Ds{}{l+1}}$. Therefore, \eqref{equ:NewProperty2} implies $\AllQ{k\inps\Nbn}{\wG\ll{k-l-1,k}\inps \Ds{}{l+2}}$, hence $\weG\in{\Beha}^{l+1}$.\\
  \item[\enquote{$\Leftarrow$}]: 
  First observe that \eqref{equ:NewProperty2} always holds if $\zeG\ll{0,l}\in\Set{\diamond}^l$. We therefore pick $\zeG$ s.t. $\zeG\ll{0,l}\in\Ds{}{l+1}$ and $\zeG\ll{1,l+1}\in\Ds{}{l+1}$. Using \eqref{equ:BehalW:b}, this implies the existence of $\weG,\weG'\in\Behal$ and $k,k'\in\Nbn$ s.t. $\wG\ll{k-l,k}=\zeG\ll{0,l}$ and $\wG'\ll{k'-l,k'}=\zeG\ll{1,l+1}$. Picking $\weG''=\CONCAT{\weG}{k}{k'}{\weG}$ it is easily verified that $\weG''\in\Behal$ and $\wG''\ll{k-l,k+1}=\zeG$. As $\Behal=\Beha^{l+1}$ we obtain $\zeG=\wG''\ll{k-l,k+1}\in\Beha^{l+1}\ll{k-l,k+1}\subseteq\pi_{l+2}\BR{\Beha^{l+1}}=\Ds{}{l+2}$, what proves the statement
 \end{inparaitem}.
\end{proof}

\subsubsection*{Proof of \REFthm{thm:SimRel_QsysalpmQsysalm}}
We show both statements separately.\\
%
 \begin{inparaitem}
\item[(i)] To show that \eqref{equ:SimRel_EL:a} holds for $\R$, let $\xa_{l+1}\in\xaSo{}^{{\Ilpm}}$ and pick $\xa_l=\xa_{l+1}\ll{1,l}$, i.e., $\Tuple{\xa_{l+1},\xa_l}\in\R$. Now it follows from \eqref{equ:xalSo} that there exists $\xe\in\xeSo{}$ s.t. $\xa_{l+1}\in\EnabWl{}{{\Ilpm}}{\xe}$. Now \eqref{equ:Xxr} implies that $\xa_{l}\in\EnabWl{}{{\Ilm}}{\xe}$, hence $\xa_{l}\in\xaSo{}^{{\Ilm}}$ (from \eqref{equ:xalSo}).\\
To show that \eqref{equ:SimRel_EL:b} holds for $\R$, we pick $\Tuple{\xa_{l+1},\xa_l}\in\R$, $\ue$, $\ye$, $\we$ and $\xa_{l+1}'$, $\xa_{l}'$ s.t. $\Tuple{\xa_{l+1},\ue,\ye,\xa_{l+1}'}\in\tra^{\Ilpm}$, $\xa_{l}'=\xa_{l+1}'\ll{1,l}$ (i.e., $\Tuple{\xa_{l+1}',\xa_l'}\in\R$) and $\we=\projState{\weS}{\ue,\ye}$. Now using \eqref{equ:tralversusDs:b} for $l+1$ gives
 \begin{align*}
  &\xa_{l+1}'\ll{0,l-m}=\BR{\xa_{l+1}\ll{0,l-m}\sconc\we}\ll{1,l+1-m}\\
  &\xa_{l+1}\ll{l+1-m,l}=\BR{\we\sconc\xa_{l+1}'\ll{l+1-m,l-1}}\ll{0,m-1}\\
  &\xa_{l+1}|_{[0,l-m]}\sconc\we\sconc\xa_{l+1}'|_{[l+1-m,l]}\in\Ds{}{l+2}.
 \end{align*}
Now using $\xa_{l}=\xa_{l+1}\ll{1,l}$ and $\xa_{l}'=\xa_{l+1}'\ll{1,l}$ yields
 \begin{align*}
  &\xa_{l}'\ll{0,l-m-1}=\BR{\xa_{l}\ll{0,l-m-1}\sconc\we}\ll{1,l-m}\\
&\xa_{l}\ll{l-m,l-1}=\BR{\we\sconc\xa_{l}'\ll{l-m,l-2}}\ll{0,m-1}\\
  &\xa_{l}|_{[0,l-m-1]}\sconc\we\sconc\xa_{l}'|_{[l-m,l-1]}\in\Ds{}{l+1}.
 \end{align*}
 By using \REFlem{lem:tralversusDs} again this proves the statement.\\

\item[(ii)] To see that \eqref{equ:SimRel_EL:a} always holds for $\R^{-1}$ we pick $\xa_{l}\in\xaSo{}^{{\Ilm}}$ and observe from \eqref{equ:xalSo} that there exists an $\xe\in\xeSo{}$ s.t. $\xa_{l}\in\EnabWl{}{{\Ilm}}{\xe}$. Using \eqref{equ:Xxr} this implies the existence of $\Tuple{\weG,\xeG}\inps\BeheSQ$ s.t. $\xeG(0)=\xe$ and $\xa_l=\weG\ll{m-l,m-1}$. Now pick $\xa_{l+1}=\weG\ll{m-l-1,m-1}$ and observe that $\xa_{l+1}\in\EnabWl{}{{\Ilpm}}{\xe}$, hence $\xa_{m+1}\in\xaSo{}^{{\Ilpm}}$ (from \eqref{equ:xalSo}) and $\Tuple{\xa_l,\xa_{l+1}}\in\R^{-1}$ (from \eqref{equ:R_l}). It remains to show that \eqref{equ:SimRel_EL:b} holds for $\R^{-1}$ iff \eqref{equ:NewProperty2} holds. We show both directions separately:\\
\item[\enquote{$\Leftarrow$}]\\
 \begin{inparaitem}[-]
 \item Pick $\Tuple{\xa_l,\xa_{l+1}}\in\R^{-1}$, $\ue$, $\ye$, $\we$ and $\xa_{l+1}'$, $\xa_{l}'$ s.t. $\Tuple{\xa_{l},\ue,\ye,\xa_{l}'}\in\tra^{\Ilm}$, $\xa_{l+1}'=\xa_l(0)\sconc\xa_{l}'$ (i.e., $\Tuple{\xa_{l+1}',\xa_l'}\in\R$) and $\we=\projState{\weS}{\ue,\ye}$.\\
 \item Now using \eqref{equ:tralversusDs:b} implies the existence of $\zeG\in\Ds{}{l+1}$ s.t. $\zeG\ll{0,l-1}=\xa_l$, $\zeG\ll{1,l}=\xa_l'$ and $\we=\zeG(l-m)$, hence $\xa_{l+1}'=\zeG\in\Ds{}{l+1}$. Furthermore, recall that $\xa_{l+1}\in\xaS^{\Ilpm}\subseteq\UNION{\Set{\diamond}^l}{\Ds{}{l+1}}$ (from \eqref{equ:proof:xalSDs}). Therefore, we can apply \REFlem{lem:NewProperty2} and obtain $\zeG'\in\Ds{}{l+2}$ s.t. $\zeG'\ll{0,l}=\xa_{l+1}$, $\zeG'\ll{1,l+1}=\xa_{l+1}'$ and $\we=\zeG'(l+1-m)$.\\
 \item Using  \eqref{equ:tralversusDs:b} again this implies the existence of $\ue'$, $\ye'$ s.t. $\we=\projState{\weS}{\ue',\ye'}$ and $\Tuple{\xa_{l+1},\ue',\ye',\xa_{l+1}'}\in\tra^{\Ilpm}$, what proves the statement.\\
 \end{inparaitem}
\item[\enquote{$\Rightarrow$}]\\
 \begin{inparaitem}[-]
 \item Pick $\zeG$ s.t. $\zeG\ll{0,l}\in\UNION{\Set{\diamond}^l}{\Ds{}{l+1}}$ and $\zeG\ll{1,l+1}\in\Ds{}{l+1}$. Furthermore, define $\xa_l=\zeG\ll{1,l}$ and $\xa_l'=\zeG\ll{2,l}$ and $\we=\zeG(l-m+1)$.\\
 \item With this choice \eqref{equ:tralversusDs:b} implies the existence of $\ue$, $\ye$ s.t. $\we=\projState{\weS}{\ue,\ye}$ and $\Tuple{\xa_{l},\ue,\ye,\xa_{l}'}\in\tra^{\Ilm}$. \\
 \item Now observe that $\Tuple{\xa_l,\zeG\ll{0,l}}\in\R^{-1}$ and recall that $\R^{-1}\inps\SR{\weS}{}{\Qsysa^{\Ilm}}{\Qsysa^{\Ilpm}}$, hence we know that there exists $\zeG'$, $\ue'$, $\ye'$ s.t. $\Tuple{\zeG\ll{0,l},\ue',\ye',\zeG'}\in\tra^{\Ilpm}$, $\projState{\weS}{\ue',\ye'}=\we=\zeG(l-m+1)$ and $\xa_{l}'=\zeG'\ll{1,l}$, hence $\zeG'=\zeG\ll{1,l+1}$.\\
 \item Now using \eqref{equ:tralversusDs:b} for $l+1$ implies $\zeG\in\Ds{}{l+2}$, what proves the statement.
 \end{inparaitem}
\end{inparaitem}

\subsection{Proof of \REFthm{thm:SimRel_QsysalmpQsysalm}}\label{proof:thm:SimRel_QsysalmpQsysalm}

\begin{lemma}\label{lem:NewProperty}
Given \eqref{equ:Prelim} let
 \begin{equation}\label{equ:NewProperty}
  \AllQ{\zeG,\zeG'\inps\Ds{}{l+1}}{
  \propImp*{\zeG\ll{0,l-1}\eqps\zeG'\ll{0,l-1}}{\zeG\eqps\zeG'}}.
 \end{equation}
Then
\begin{equation*}
 \propAequ{\text{\eqref{equ:NewProperty}}}{
 \begin{propConjA}
  \text{$\Qsyse$ is future unique w.r.t. $\Ilmp$}\\
  \text{$\Qsyse$ is state-based asynch. $l$-complete w.r.t. $\Ilm$}
 \end{propConjA}}
\end{equation*}
\end{lemma}

\begin{proof} We show all statements separately.\\
 \begin{inparaitem}
  \item Show $\propImp{\text{\eqref{equ:NewProperty}}}{\text{\eqref{equ:future_unique}}}$:\\
   \begin{inparaitem}
  \item Pick $\xe\in\xeS$ and $\zeG,\zeG'\in\EnabWl{}{\Ilmp}{\xe}$ and observe from \eqref{equ:Xxr} that there exist $\Tuple{\weG,\xeG},\Tuple{\weG',\xeG'}\inps\BeheSQ$ and $k,k'\in\Nbn$ s.t. $\xeG(k)=\xeG(k')=\xe$, $\weG\ll{k-l+m+1,k+m}=\zeG$ and $\weG'\ll{k'-l+m+1,k'+m}=\zeG'$, hence 
   $\Tuple{\weG'',\xeG''}=\CONCAT{\Tuple{\weG,\xeG}}{k}{k'}{\Tuple{\weG',\xeG'}}\inps\BeheSQ$.
  \item Now pick $\tilde{\zeG}=\weG\ll{k-l,k}$ and $\tilde{\zeG}'=\weG''\ll{k-l,k}$ and observe that $\tilde{\zeG}'\ll{0,l-1}=\tilde{\zeG}\ll{0,l-1}$.
  Using \eqref{equ:NewProperty} we therefore have $\tilde{\zeG}=\tilde{\zeG}'$, hence $\weG(k)=\zeG(l-m+1)=\zeG'(l-m+1)=\weG''(k)$.\\
  \item Now we can pick  $\tilde{\zeG}=\weG\ll{k-l+1,k+1}$ and $\tilde{\zeG}'=\weG''\ll{k-l+1,k+1}$ and (by reusing the above argument) obtain  $\weG(k)=\zeG(l-m+2)=\zeG'(l-m+2)=\weG''(k)$.
  Iteratively applying the above reasoning therefore yields $\zeG\ll{l-m+1,l}=\zeG'\ll{l-m+1,l}$.\\
  \end{inparaitem}
  \item Show $\propImp{\text{\eqref{equ:NewProperty}}}{\text{\eqref{equ:dominolcomplete}}}$:\\
  \begin{inparaitem}
  \item Pick $\xe\in\xeS$, $\zeG\in\Ds{}{l+1}$ s.t. $\zeG\ll{0,l-1}\in\EnabWl{}{\Ilm}{\xe}$ and observe from \eqref{equ:Xxr} that there exist $\Tuple{\weG,\xeG}\inps\BeheSQ$ and $k\in\Nbn$ s.t. $\xeG(k)=\xeG(k')=\xe$ and $\weG\ll{k-l+m,k+m-1}=\zeG\ll{0,l-1}$.\\
  \item Now pick $\zeG'=\weG\ll{k-l+m,k+m}$ and observe that $\zeG'\in\EnabWl{}{[m-l,m]}{\xe}$, $\zeG\ll{0,l-1}=\zeG'\ll{0,l-1}$ and $\zeG'\in\Ds{}{l+1}$.\\
  \item Using \eqref{equ:NewProperty} we have $\zeG=\zeG'$ and therefore  $\zeG\in\EnabWl{}{[m-l,m]}{\xe}$.\\
 \end{inparaitem}
   \item Show $\propImp{\propConj{\text{\eqref{equ:future_unique}}}{\text{\eqref{equ:dominolcomplete}}}}{\text{\eqref{equ:NewProperty}}}$:\\
   \begin{inparaitem}
   \item Pick $\zeG,\zeG'\in\Ds{}{l+1}$ s.t. $\zeG\ll{0,l-1}=\zeG'\ll{0,l-1}$ and observe that this implies $\zeG\ll{0,l-1}\in\Ds{}{l}$. Hence, there exists $\xe\in\xeS$ s.t. $\zeG\ll{0,l-1}\in\EnabWl{}{\Ilm}{\xe}=\EnabWl{}{[m-l,m-1]}{\xe}$.\\
   \item Using \eqref{equ:dominolcomplete} we know that $\zeG,\zeG'\in\EnabWl{}{[m-l,m]}{\xe}$, which implies $\zeG\ll{1,l},\zeG'\ll{1,l}\in\EnabWl{}{\Ilmp}{\xe}$ (from \eqref{equ:Xxr}). With this \eqref{equ:future_unique} implies $\zeG\ll{1,l}=\zeG'\ll{1,l}$, hence $\zeG=\zeG'$.
   \end{inparaitem}
\end{inparaitem}
\end{proof}

\subsubsection*{Proof of \REFthm{thm:SimRel_QsysalmpQsysalm}}
We show both statements separately.\\
%
%
 \begin{inparaitem}
\item[(i)] To show that \eqref{equ:SimRel_EL:a} holds for $\R$ pick $\xa_{m+1}\in\xaSo{}^{{\Ilmp}}$ and $\xa_m=\diamond\sconc\xa_{m+1}\ll{0,l-2}$, implying $\xa_{m+1}\ll{0,l-2}=\xa_m\ll{1,l-1}$ (i.e., $\Tuple{\xa_{m+1},\xa_m}\in\R$). Now it follows from \eqref{equ:xalSo} that there exists an $\xe\in\xeSo{}$ s.t. $\xa_{m+1}\in\EnabWl{}{{\Ilmp}}{\xe}$ and it can be easily observed from \eqref{equ:Xxr} that $\xa_{m}\in\EnabWl{}{{\Ilm}}{\xe}$, hence $\xa_{m}\in\xaSo{}^{{\Ilm}}$ (from \eqref{equ:xalSo}).
It remains to  show that \eqref{equ:SimRel_EL:b} holds for $\R$.\\
  \begin{inparaitem}
 \item Pick $\Tuple{\xa_{m+1},\xa_m}\in\R$, $\ue$, $\ye$, $\we$ and $\xa_{m+1}'$ s.t. $\Tuple{\xa_{m+1},\ue,\ye,\xa_{m+1}'}\in\tra^{\Ilmp}$ and $\we=\projState{\weS}{\ue,\ye}$. Using \eqref{equ:tral} this implies $\xa_{m+1}\ll{1,l-1}=\xa_{m+1}'\ll{0,l-2}$ and $\xa_{m+1}(l-m-1)=\we$.\\
 \item Now pick $\xa_{m}'=\xa_{m+1}$ implying that the first two lines in \eqref{equ:tralversusDs} hold and $\Tuple{\xa_{m}',\xa_m}\in\R$ . Therefore \eqref{equ:R_m}  and \eqref{equ:Xxr} imply the existence of $\Tuple{\weG,\xeG},\Tuple{\weG',\xeG'}\inps\BeheSQ$, $k,k'\in\Nbn$ s.t. $\xeG(k)=\xeG'(k')=\xe$, $\xa_m\eqps\weG|_{[k+m-l,k+m-1]}$ and $\xa_{m}'\eqps\weG'|_{[k'+m+1-l,k'+m]}$, hence $\Tuple{\tilde{\weG},\tilde{\xeG}}=\CONCAT{\Tuple{\weG,\xeG}}{k}{k'}{\Tuple{\weG',\xeG'}}\in\BeheSQ$ and $\tilde{\weG}\ll{k+m-l,k+m}=\xa_{m}\ll{0,l-m-1}\sconc\xa_{m}'\ll{l-m-1,l-1}$.\\
 \item As $\xa_m'(l-m-1)=\xa_{m+1}(l-m-1)=\we$ this implies
  $\xa_m|_{[0,l-m-1]}\sconc\we\sconc\xa_m'|_{[l-m,l-1]}\in\Ds{}{l+1}$.
Using \eqref{equ:tralversusDs} this implies the existence of $\ue'$ and $\ye'$ s.t. $\projState{\weS}{\ue',\ye'}=\we$ and $\Tuple{\xa_{m},\ue',\ye',\xa_{m}'}\in\tra^{\Ilm}$. It remains to show that $\Tuple{\xa_{m+1}',\xa_m'}\in\R$.\\
\begin{inparaitem}[$\cdot$]
\item Recall that $\xa_{m+1}'\ll{0,l-2}=\xa_{m+1}\ll{1,l-1}=\xa_m'\ll{1,l-1}$, hence the first line in \eqref{equ:R_m} holds.\\
 \item To see that the second line in \eqref{equ:R_m}  also holds, observe that $\Tuple{\xa_{m+1},\ue,\ye,\xa_{m+1}'}\in\tra^{\Ilmp}$ and \eqref{equ:tralversusDs:b} implies the existence of $\zeta\in\Ds{}{l+1}$ s.t. $\xa_{m+1}=\zeG\ll{0,l-1}$ and $\xa_{m+1}'=\zeG\ll{1,l}$. We can therefore pick $\Tuple{\tilde{\weG},\tilde{\xeG}}\in\BeheSQ$, $\tilde{\xe}$ and $\tilde{k}$ s.t. $\tilde{\weG}\ll{\tilde{k}+(m+1)-l,\tilde{k}+(m+1)}=\xa_{m+1}|_{[0,0]}\sconc\xa'_{m+1}$ and $\tilde{\xe}=\tilde{\xeG}(k+1)$ and have $\xa_{m+1}'\in\EnabWl{}{\Ilmp}{\tilde{\xe}}$ and $\xa_{m}'=\xa_{m+1}\in\EnabWl{}{\Ilm}{\tilde{\xe}}$.\\
\end{inparaitem}
 \end{inparaitem}
 \item[(ii)] To see that  \eqref{equ:SimRel_EL:a} always holds for $\R^{-1}$, pick $\xa_{m}\in\xaSo{}^{{\Ilm}}$ and recall from \eqref{equ:xalSo} that there exists an $\xe\in\xeSo{}$ s.t. $\xa_{m}\in\EnabWl{}{{\Ilm}}{\xe}$. Using \eqref{equ:Xxr} this implies the existence of $\Tuple{\weG,\xeG}\inps\BeheSQ$ s.t. $\xeG(0)=\xe$ and $\xa_m=\weG\ll{m-l,m-1}$. Using $\xa_{m+1}=\weG\ll{m+1-l,m}$ therefore yields $\xa_{m+1}\in\EnabWl{}{{\Ilmp}}{\xe}$, hence $\xa_{m+1}\in\xaSo{}^{{\Ilmp}}$ (from \eqref{equ:xalSo}) and $\Tuple{\xa_m,\xa_{m+1}}\in\R^{-1}$ (from \eqref{equ:R_m}). It remains to show that \eqref{equ:SimRel_EL:b} holds for $\R^{-1}$ iff \eqref{equ:NewProperty} holds. We show both statements separately.\\
 \item[\enquote{$\Leftarrow$}:]
   \begin{inparaitem}
   \item Pick $\Tuple{\xa_m,\xa_{m+1}}\in\R^{-1}$, $\ue$, $\ye$, $\we$ and $\xa_{m}'$ s.t. $\Tuple{\xa_{m},\ue,\ye,\xa_{m}'}\in\tra^{\Ilm}$ and $\we=\projState{\weS}{\ue,\ye}$. Using \eqref{equ:tralversusDs:b} this implies the existence of $\zeG\in\Ds{}{l+1}$ s.t. $\xa_m=\zeG\ll{0,l-1}$, $\xa_m'=\zeG\ll{1,l}$ and $\we=\zeG(l-m)$. \\
   \item Now let $\zeG'=\xa_m\ll{0,0}\sconc\xa_{m+1}$ and observe that $\zeG\ll{0,l-1}=\zeG'\ll{0,l-1}$ and therefore $\zeG=\zeG'$ (from \eqref{equ:NewProperty}), hence $\xa_{m+1}=\xa_m'$.\\
   \item As $\zeG'\in\Ds{}{l+1}$ we can pick  $\Tuple{\tilde{\weG},\tilde{\xeG}}\in\BeheSQ$, $\tilde{\xe}$ and $\tilde{k}$ s.t. $\tilde{\weG}\ll{\tilde{k}+m+-l,\tilde{k}+m+}=\zeG$. Now pick $\zeG''=\tilde{\weG}\ll{\tilde{k}+(m+1)-l,\tilde{k}+(m+1)}$ and observe that $\xa_{m+1}=\zeG''\ll{0,l-1}$ and $\zeG''(l-m-1)=\we$. Therefore choosing $\xa_{m+1}'=\zeG''\ll{1,l}$ and using \eqref{equ:tralversusDs:b} implies the existence of $\ue'$ and $\ye'$ s.t. $\Tuple{\xa_{m+1},\ue',\ye',\xa_{m+1}'}\in\tra^{\Ilmp}$ and $\projState{\weS}{\ue',\ye'}=\we$. \\
   \item Furthermore, observe that choosing $\tilde{\xe}=\tilde{\xeG}(\tilde{k}+1)$ implies $\xa_{m+1}'\in\EnabWl{}{\Ilmp}{\tilde{\xe}}$ and $\xa_{m}'=\xa_{m+1}\in\EnabWl{}{\Ilm}{\tilde{\xe}}$, hence $\Tuple{\xa_m',\xa_{m+1}'}\in\R^{-1}$.\\
   \end{inparaitem}
 \item[\enquote{$\Rightarrow$}:]
  \begin{inparaitem}
   \item Pick $\zeG,\zeG'\in\Ds{}{l+1}$ s.t. $\zeG\ll{0,l-1}=\zeG'\ll{0,l-1}$ and pick  $\xa_m=\zeG\ll{0,l-1}=\zeG'\ll{0,l-1}$, $\xa_m'=\zeG\ll{1,l}$, $\xa_{m+1}=\zeG'\ll{1,l}$ and $\we=\zeG(l-m)$. Using \eqref{equ:tralversusDs:b} this implies the existence of $\ue,\ye$ s.t.  $\Tuple{\xa_{m},\ue,\ye,\xa_{m}'}\in\tra^{\Ilm}$ and $\we=\projState{\weS}{\ue,\ye}$. Using the same reasoning as before it furthermore holds that $\Tuple{\xa_m,\xa_{m+1}}\in\R^{-1}$ as $\zeG'\in\Ds{}{l+1}$. \\
   \item As $\R^{-1}\in\SR{\weS}{}{\Qsysa^{\Ilm}}{\Qsysa^{\Ilmp}}$ we know that there exist $\ue',\ye',\xa_{m+1}'$ s.t. $\Tuple{\xa_{m+1},\ue',\ye',\xa_{m+1}'}\in\tra^{\Ilm}$, $\we=\projState{\weS}{\ue',\ye'}$ and $\Tuple{\xa_m',\xa_{m+1}'}\in\R^{-1}$.\\
   \item Now let $l=1$ implying $m=0$ (as $m<l$). Then \eqref{equ:tral} implies $\we=\xa_{m+1}(1)$ and therefore
   $\zeG(1)=\we=\xa_{m+1}(1)=\zeG'(1)$ implying $\zeG=\zeG'$.\\
   \item Now let $l>1$ and observe that $\xa_{m+1}\ll{1,l-1}=\xa_{m+1}'\ll{0,l-2}=\xa_m'\ll{1,l-1}$ (from \eqref{equ:tral} and \eqref{equ:R_m}). Therefore $\zeG(l)=\xa_m'(l-1)=\xa_{m+1}(l-1)=\zeG(l)'$ holds, giving $\zeG=\zeG'$.
 \end{inparaitem}
  \end{inparaitem}

\subsection{Proof of \REFprop{prop:PhilvsEl}}\label{app:proof:prop:PhilvsEl}
\begin{inparaitem}
 \item $l=1$: Recall that $\Interval^1_1=[0,0]$ and observe that \eqref{equ:TransStruct}, \eqref{equ:Behtr}, \eqref{equ:Xxr} and $\weS=\yeS$ implies $\EnabWl{}{\Interval^1_1}{\xe}=\EnabY{\tre}{\xe}$. Hence, \eqref{equ:PhilvsEl} holds for $l=1$ from \eqref{equ:Phil:0}.\\
\item $(l-1)\fun l$: Assume that
 \begin{equation}\label{equ:proof:EH:0}
  \Phi^{l-1}=\SetCompX{\EnabWlr{}{\Interval^{l-1}_{l-1}}{V}}{V\in\twoup{\BR{\yeS}^{l-1}}}
 \end{equation}
holds. 
 Using \eqref{equ:TransStruct}, \eqref{equ:Behtr}, \eqref{equ:Xxr} and $\weS=\yeS$ again we obtain
 {\allowdisplaybreaks
  \begin{align*}
  \zeG\inps\EnabWl{}{[0,l-1]}{\xe}
  %
  &\Leftrightarrow
  \ExQ{\tilde{\xe}\inps\xeS,\ue\inps\ueS}{
  \begin{propConjA}
  \Tuple{\xe,\ue,\zeG(0),\tilde{\xe}}\inps\tre\\
  \zeG\ll{1,l-1}\inps\EnabWl{}{[0,l-2]}{\tilde{\xe}}
  \end{propConjA}}
 \end{align*}}
implying
{\allowdisplaybreaks
 \begin{align*}
  &\EnabWl{}{[0,l-1]}{\xe}=\EnabWl{}{[0,l-1]}{\xe'}\\
  &\Leftrightarrow
  \AllQ{\zeG\in\BR{\yeS}^l}{\propImp*{\zeG\in\EnabWl{}{[0,l-1]}{\xe}}{\zeG\in\EnabWl{}{[0,l-1]}{\xe'}}}\\
   %
   %
  &\Leftrightarrow
  \AllQSplit{\ye\in\yeS,V\in\twoup{\BR{\yeS}^{l-1}}}{
  \propImpSplit{
  \ExQ{\tilde{\xe}\in\EnabWlr{}{[0,l-2]}{V},\ue\in\ueS}{
  \Tuple{\xe,\ue,\ye,\tilde{\xe}}\in\tre}
  }{
  \ExQ{\tilde{\xe}'\in\EnabWlr{}{[0,l-2]}{V},\ue'\in\ueS}{
  \Tuple{\xe',\ue',\ye,\tilde{\xe}'}\in\tre}
  }}\\
 &\Leftrightarrow
   \begin{propConjA}
  \EnabY{\tre}{\xe}=\EnabY{\tre}{\xe'}\\
  \AllQ{\zeS'\in\Phi^{l-1}}{
  \propImpSplit{
 \INTERSECT{\Set{\xe}}{\EnabTr{\tre}{\zeS'}}\neq\emptyset
  }{
  \INTERSECT{\Set{\xe'}}{\EnabTr{\tre}{\zeS'}}\neq\emptyset
  }}  
  \end{propConjA}
 \end{align*}}
 where the last equality follows from \eqref{equ:proof:EH:0} and \eqref{equ:TransStruct}.
 With this we obtain from \eqref{equ:Phil_Prop:b} that
 {\allowdisplaybreaks
 \begin{align*}
 &\zeS\in\SetCompX{\EnabWlr{}{[0,l-1]}{V}}{V\in\twoup{\BR{\yeS}^{l}}}\\
  &\Leftrightarrow
  \AllQ{\xe,\xe'\in\zeS}{\EnabWl{}{[0,l-1]}{\xe}=\EnabWl{}{[0,l-1]}{\xe'}}\\
  &\Leftrightarrow
     \begin{propConjA}
  \AllQ{\xe,\xe'\in\zeS}{\EnabY{\tre}{\xe}=\EnabY{\tre}{\xe'}}\\
  \AllQ{\zeS'\inps\Phi^{l-1}}{
  \propImp*{
   \INTERSECT{\zeS}{\EnabTr{\tre}{\zeS'}}\neq\emptyset}{
   \zeS\subseteq\EnabTr{\tre}{\zeS'}}}
   \end{propConjA}
  \\
  &\Leftrightarrow
   \begin{propConjA}
  \AllQ{\xe,\xe'\in\zeS}{\EnabY{\tre}{\xe}=\EnabY{\tre}{\xe'}}\\
  \zeS\in\Phi^l
  \end{propConjA}\\
  &\Leftrightarrow\zeS\in\Phi^l
 \end{align*}}
 where the last equality follows from \eqref{equ:Phil_Prop:a} .
%
\end{inparaitem}

\subsection{Proof of \REFthm{thm:SimRel_QsyseQsysaqbl}}\label{app:proof:thm:SimRel_QsyseQsysaqbl}

The proof of part (i) follows the same lines as the proof in \cite{TabuadaBook}, Thm.~4.18. and is therefore omitted. For part (ii) first observe that \eqref{equ:SimRel_EL:a} always holds for $\R^{-1}$, as we can pick $\xa\in\xaqlSo{}$ and obtain from \eqref{equ:xaqlSo} that there exists $\xe\in\xeSo{}$ s.t. $\xa=\EnabWl{}{\Ill}{\xe}$. To prove that \eqref{equ:SimRel_EL:b} holds for $\R^{-1}$ iff $\Phi^l$ is a fixed-point of \eqref{equ:Phil}, observe that \eqref{equ:SimRel_EL:b} holds for $\R^{-1}$ and $\yeS$ iff for all $\xa,\xa',\ue,\ye,\xe$ holds
 \begin{equation*}
   \propImp{
   \begin{propConjA}
   \xa=\EnabWl{}{\Ill}{\xe}\\
   \Tuple{\xa,\ue,\ye,\xa'}\inps\traqbl
   \end{propConjA}}{
   \ExQ{\xe',\ue'}{
   \begin{propConjA}
   \xa'=\EnabWl{}{\Ill}{\xe'}\\
   \Tuple{\xe,\ue',\ye,\xe'}\inps\tre
   \end{propConjA}
   }}.
 \end{equation*}
Now let $\zeS=\EnabWlr{}{\Ill}{\xa}$ and $\zeS'=\EnabWlr{}{\Ill}{\xa'}$. As \eqref{equ:TransStruct} holds for $\Qsyse$, using \eqref{equ:traql}, it can be easily verified that the previous statement is equivalent to 
 \begin{align*}
   \AllQ{\zeS,\zeS'}{
  \propImp{
  \INTERSECT{\zeS}{\EnabTr{\tre}{\zeS'}\neq\emptyset}}{
  \zeS\subseteq\EnabTr{\tre}{\zeS'}
  }}.
 \end{align*}
Using \REFprop{prop:PhilvsEl} and \eqref{equ:Phil_Prop:c} this proves the statement.

\subsection{Proof of \REFthm{thm:QsysaqblSAlA}}\label{app:proof:thm:QsysaqblSAlA}

\begin{lemma}
Given \eqref{equ:Prelim}  and $\Qsysaql$ as in \REFdef{def:Tsysaqbl},it holds that
\begin{subequations}\allowdisplaybreaks
 \begin{align}
 \AllQ{V\inps\twoup{\BR{\yeS}^l},r\hspace{-0.1cm}<\hspace{-0.1cm}l}{\EnabWlr{}{\Ill}{V}\subseteq\EnabWlr{}{\Interval^{l\mips r}_{l\mips r}}{V\ll{0,l-r-1}}}\label{equ:project:include}\\
 \text{and}~\propImp{\Tuple{\xa,\ue,\ye,\xa'}\in\traql}{
 \begin{propConjA}
  \xa'\ll{0,l-2}\subseteq \xa\ll{1,l-1}\\
  \ye\in \xa\ll{0,0}
 \end{propConjA}}.\label{equ:traqlversusDs}
\end{align}
\end{subequations}
\end{lemma}

\begin{proof}
\begin{inparaitem}
\item[(a)]  Pick $\xe\in\EnabWlr{}{\Ill}{V}$. Using \eqref{equ:Xxr} this implies that for all $\zeG\in V$ there exists $\Tuple{\weG,\xeG}\in\BeheSQ$, $k\in\Nbn$ s.t. $\xe=\xeG(k)$ and $\zeG=\weG\ll{k,k+l-1}$. Now observe, that for every choice of $\zeG$ it holds that $\zeG\ll{0,l-r-1}\in V\ll{0,l-r-1}$. Using the same choice of signals $\Tuple{\weG,\xeG}$ and $k$ this immediately implies that $\xe\in\EnabWlr{}{\Interval^{l\mips r}_{l\mips r}}{V\ll{0,l-r-1}}$, what proves the statement.\\
\item[(b)] Pick $\xa,\xa',\ue,\ye$ s.t. $\Tuple{\xa,\ue,\ye,\xa'}\in\traqbl$ and define $\zeS=:\EnabWlr{}{\Ill}{\xa}$ and $\zeS'=:\EnabWlr{}{\Ill}{\xa'}$.\\
\begin{inparaitem}[-]
\item Using \eqref{equ:traql} this implies that there exists $\xe\in\zeS$ s.t. $\ye\in\EnabY{\tre}{\xe}$. Furthermore, using \eqref{equ:project:include} we know that $\xe\in\EnabWlr{}{\Interval^1_1}{\xa\ll{0,0}}$. As $\EnabWl{}{\Interval^1_1}{\xe}=\EnabY{\tre}{\xe}$ this implies $\ye\in\xa\ll{0,0}$.\\
\item Let $\tilde{\zeS}'=\EnabWlr{}{\Interval^{l-1}_{l-1}}{\xa'\ll{0,l-2}}$ and recall from \eqref{equ:project:include} that $\zeS'\subseteq\tilde{\zeS}'$ implying $\EnabTr{\tre}{\zeS'}\subseteq\EnabTr{\tre}{\tilde{\zeS}'}$. Using \eqref{equ:traql} and \eqref{equ:EnabT} we know that $\Tuple{\xa,\ue,\ye,\xa'}\in\traqbl$ implies $\INTERSECT{\zeS}{\EnabTr{\tre}{\zeS'}}\neq\emptyset$. Hence,  $\INTERSECT{\zeS}{\EnabTr{\tre}{\tilde{\zeS}'}}\neq\emptyset$ and therefore (using \eqref{equ:Phil_Prop:b}) $\zeS\subseteq\EnabTr{\tre}{\tilde{\zeS}'}$. 
This implies that 
   \begin{equation}\label{equ:proof:subsetsV}
    \AllQ{\xe\in\zeS}{\ExQ{\xe'\in\EnabT{\tre}{\xe}}{\xa'\ll{0,l-2}=\EnabWl{}{[0,l-2]}{\xe'}}}.
   \end{equation}
Now it follows from \eqref{equ:Xxr} that
    $\xa\ll{1,l-1}=\bigcup_{\xe\in\EnabT{\tr}{\zeS}}\EnabWl{}{[0,l-2]}{\xe}$
   implying $\xa'\ll{0,l-2}\subseteq \xa\ll{1,l-1}$.
 \end{inparaitem}
\end{inparaitem}
\end{proof}

\subsubsection*{Proof of \REFthm{thm:QsysaqblSAlA}}
Pick $\yaG\in\Behe(\Qsysaqbl)$ and observe from \eqref{equ:Behtr} that there exist $\Tuple{\uaG,\xaG}$ s.t. $\xaG(0)\in\xaqlSo{}$ and \linebreak$\AllQ{k\in\Nbn}{\Tuple{\xaG(k),\uaG(k),\yaG(k),\xaG(k+1)}\in\traqbl}$. Using \eqref{equ:traqlversusDs} we know that for all $k\in\Nbn$ it holds that $\xaG(k+1)\ll{0,l-2}\subseteq \xaG(k)\ll{1,l-1}$ and $\yaG(k)\in \xaG(k)\ll{0,0}$. Applying these equations iteratively yields $\yaG\ll{k,k+l-1}\in\xaG(k)$.\\
Furthermore, we can use \eqref{equ:traql} and \eqref{equ:Behtr} to pick $\Tuple{\yeG',\xeG'}\inps\BehS{}(\Qsyse)$ and $k'\in\Nbn$ s.t. $\yeG'(k')\eqps\yaG(k)$, $\xaG(k)\eqps\EnabWl{}{\Ill}{\xeG'(k')}$ and $\xaG(k+1)\eqps\EnabWl{}{\Ill}{\xeG'(k'+1)}$.
%
Using these signals, $\yaG\ll{k,k+l-1}\in\xaG(k)$ implies
\begin{align*}
  &\ExQ{\Tuple{\yeG'',\xeG''}\inps\BehS{}(\Qsyse),k''}{
  \begin{propConjA}
   \yeG''\ll{k'',k''+l-1}=\yaG\ll{k,k+l-1}\\
   \xeG''(k'')=\xeG'(k')
   \end{propConjA}}~\text{and}\\
   &\ExQ{\Tuple{\yeG''',\xeG'''}\inps\BehS{}(\Qsyse),k'''}{
  \begin{propConjA}
   \yeG'''\ll{k''',k'''+l-1}=\yaG\ll{k+1,k+l}\\
   \xeG'''(k''')=\xeG'(k'+1)
   \end{propConjA}}
\end{align*}
Using \eqref{equ:StateSpaceDynamicalSystem:2} we now obtain
  $\tilde{\yeG}=\CONCAT{\yeG''}{k''}{k'}{\CONCAT{\yeG'}{k'+1}{k'''}{\yeG'''}}\in\BeheQ$
 where $\yaG\ll{k,k+l}=\tilde{\yeG}\ll{k',k'+l}$, hence $\yaG\ll{k,k+l}\in\Ds{}{l+1}$.\\
 Using \eqref{equ:BehalW} it remains to show that $\yaG\ll{-l,0}\in\BeheQ\ll{-l,0}$. 
Observe from \eqref{equ:xaqlSo} and \eqref{equ:Xxr} that $\xaG(0)\in\xaqlSo{}$ implies $\xaG(0)\subseteq\projState{\yeS}{\Behf(\Qsyse)}\ll{0,l-1}$, hence $\yaG\ll{0,l-1}\in\BeheQ\ll{0,l-1}$. As $\AllQ{k<0}{\yaG(k)=\diamond}$ we therefore have $\yaG\ll{-l,0}\in\BeheQ\ll{-l,0}$. 

\subsection{Proof of \REFthm{thm:QsysalQsysaqbl}}\label{app:proof:thm:QsysalQsysaqbl} 
 \begin{inparaitem}
  \item[(i)] To see that \eqref{equ:SimRel_EL:a} always holds for $\R$ pick $\zeG\in\xaSo{}^{\Ill}$. Then it follows from \eqref{equ:xalSo} that there exists $\xe\in\xeSo{}$ s.t. $\zeG\in\EnabWl{}{\Ill}{\xe}$. Now using $\ya=\EnabWl{}{\Ill}{\xe}$ implies $\ya\in\xaqlSo{}$.
  It remains to show that \eqref{equ:SimRel_EL:b} holds for $\R$ and $\weS=\yeS$ iff $\Qsysaqbl$ is domino consistent. We show both directions separately:\\
  \item[\enquote{$\Rightarrow$}] 
Pick $\tilde{\zeta}\in\Ds{}{l+1}$ and $\ya\in\yaS^l$ s.t. $\zeta=\tilde{\zeta}\ll{0,l-1}\in\ya$ (hence $\Tuple{\zeta,\ya}\in\R$) and pick $\zeta'=\tilde{\zeta}\ll{1,l}$ and $\ye=\tilde{\zeta}(0)$. Then it follows from \eqref{equ:tralversusDs:b} that there exits $\ue$ s.t. $\Tuple{\zeG,\ue,\ye,\zeG'}\inps\tra^{\Ill}$. As $\R\in\SR{\yeS}{}{\Qsysa^{\Ill}}{\Qsysaqbl}$ it follows from \eqref{equ:SimRel_EL:b} and \eqref{equ:traql} that there exists $\xe\inps\EnabWlr{}{\Ill}{\ya}$, $\ue'$ and $\xe'$ s.t. $\Tuple{\xe,\ue,\ye,\xe'}\inps\tre$ and $\zeta'\in\EnabWl{}{\Ill}{\xe'}$. Now it follows immediately from \eqref{equ:Xxr} that $\ye\sconc\zeta'=\tilde{\zeta}\in\EnabWl{}{[0,l]}{\xe}$, what proves the statement.\\
 \item[\enquote{$\Leftarrow$}] Pick $\Tuple{\zeG,\ya}\in\R$, i.e., $\zeG\in\ya$ and $\ue,\ye,\zeG'$ s.t. $\Tuple{\zeG,\ue,\ye,\zeG'}\in\tra^{\Ill}$. Now it follows from \eqref{equ:tralversusDs:b} that there exists  $\zeG''=\ye\sconc\zeG'\in\Ds{}{l+1}$ with $\zeG''\ll{0,1}=\zeG\in\ya$. Using \eqref{equ:domconsist} therefore implies the existence of $\xe\inps\EnabWlr{}{\Ill}{\ya}$ s.t. $\ye\sconc\zeG'\inps\EnabWl{}{[0,l]}{\xe}$. Using \eqref{equ:Xxr} this implies that there exists $\Tuple{\ueG',\yeG',\xeG'}\in\BeheSQ$ and  $k'\in\Nbn$ s.t.
$\ya=\EnabWl{}{\Ill}{\xeG'(k')}$ and $\zeG'\in\EnabWl{}{\Ill}{\xeG'(k'+1)}$. Choosing $\ya'=\EnabWl{}{\Ill}{\xeG'(k'+1)}$ therefore implies $\zeG'\in\ya'$ (hence $\Tuple{\zeta',\ya'}\in\R$). Moreover, using \eqref{equ:traql} with $\xe=\xeG'(k')$, $\xe'=\xeG'(k'+1)$ and $\ue'=\ueG'(k')$ immediately implies $\Tuple{\ya,\ue',\ye,\ya'}\inps\traqbl$, what proves the statement.\\

\item[(ii)] We show both directions separately.\\
\item[\enquote{$\Rightarrow$}] 
Let $\R^{\Ill}$ and $\Rq$ be equivalent to the relations in \eqref{equ:R0} and \eqref{equ:QsyseQsysaqbl:Rq} (with $m=l$), respectively. Using $\R$ as in \eqref{equ:thm:QsysalQsysaqbl:R0} it is easily verified that 
 \begin{align*}
  \Rq\hspace{-0.1cm}\circ\hspace{-0.1cm}\R^{-1}\hspace{-0.1cm}:=\hspace{-0.1cm}\SetCompX{\Tuple{\xe,\xa}\inps\xeS\timesps\xaS^{\Ill}}{
  \ExQ{\ya\inps\xaqlS\hspace{-0.2cm}}{\hspace{-0.2cm}
  \begin{propConjA}
   \Tuple{\xe,\ya}\inps\Rq\\
   \Tuple{\ya,\xa}\inps\R^{-1}
  \end{propConjA}
  }\hspace{-0.1cm}}
  \eqps\R^{\Ill}
 \end{align*}
 Using the transitivity of  simulation relations therefore gives
  \begin{align*}
  \begin{propConjA}
   \Rq\inps\SR{\yeS}{}{\Qsyse}{\Qsysaqbl}\\
   \R^{-1}\inps\SR{\yeS}{}{\Qsysaqbl}{\Qsysa^{\Ill}}
  \end{propConjA}
  &\Rightarrow
\Rq\circ\R^{-1}\eqps\R^{\Ill}\inps\SR{\yeS}{}{\Qsyse}{\Qsysa^{\Ill}}\\
  &\Rightarrow\text{$\Qsyse$ is future unique w.r.t. $\Ill$}
 \end{align*}
 where the last implication follows from \REFthm{thm:SimRel_QsyseQsysaqbl}.\\
 \item[\enquote{$\Rightarrow$}] 
It follows from \eqref{equ:future_unique} and $m=l$ that for all $\xe\in\xeS$ holds $\propImp{\EnabWl{}{\Ill}{\xe}\neq\emptyset}{\length{\EnabWl{}{\Ill}{\xe}}=1}$.
Using \eqref{equ:xaqlS} this immediately implies $\length{\ya}=1$ for all $\ya\in\xaqlS$. Therefore \eqref{equ:thm:QsysalQsysaqbl:R0} becomes
 \begin{equation}\label{equ:Rql:simple}
 \R=\SetCompX{\Tuple{\zeG,\ya}\in\xaS^{\Ill}\times\xaqlS}{\ya=\Set{\zeG}}.
 \end{equation}
 To see that \eqref{equ:SimRel_EL:a} holds for $\R^{-1}$ pick $\ya\in\xaqlSo{}$ and observe from \eqref{equ:xaqlSo} that there exists $\xe\in\xeSo{}$ s.t. $\ya=\EnabWl{}{\Ill}{\xe}$. By choosing $\zeG\in\ya$ we obtain $\zeG\in\EnabWl{}{\Ill}{\xe}$, i.e., $\zeG\in\xaSo{}^{\Ill}$ (from \eqref{equ:xalSo}).\\
To sow that \eqref{equ:SimRel_EL:b} holds for $\R^{-1}$ we pick $\Tuple{\ya,\zeG}\in\R^{-1}$, i.e., $\ya=\Set{\zeG}$ and $\ue,\ye,\ya',\zeG'$ s.t. $\Tuple{\ya,\ue,\ye,\ya'}\in\traqbl$ and $\ya'=\Set{\zeG'}$. Using \eqref{equ:Rql:simple} this immediately implies that $\Tuple{\ya',\zeG'}\in\R^{-1}$. Now \eqref{equ:traql} implies the existence of $\xe,\xe'$ s.t. $\Set{\zeG}=\EnabWl{}{\Ill}{\xe}$, $\Set{\zeG'}=\EnabWl{}{\Ill}{\xe'}$ and $\Tuple{\xe,\ue,\ye,\xe'}\in\tre$. Using \eqref{equ:tral} this immediately implies $\Tuple{\zeG,\ue,\ye,\zeG'}\in\tra^{\Ill}$, what proves the statement.
 \end{inparaitem}

\subsection{Proof of \REFprop{prop:Bisim_QsysaqblQsysall}}\label{app:proof:prop:Bisim_QsysaqblQsysall} 
First observe that (i) follows from Cor.~\ref{cor:QsysaqblSAlA}, (ii) follows from \eqref{equ:Rql:simple} in the proof of \REFthm{thm:QsysalQsysaqbl} and (iv) follows from (iii) using  \REFthm{thm:SimRel_QsyseQsysal} and  \REFthm{thm:SimRel_QsyseQsysaqbl}. Hence, we only prove (iii).
Let $\R$ be defined as in \eqref{equ:cor:Bisim_QsysaqblQsysall} and observe that \eqref{equ:Rql:simple} implies
 \begin{align*}
  \R^{-1}\hspace{-0.1cm}\circ\hspace{-0.1cm}\BR{\R^\lsup}^{-1}
  &\eqps\SetCompX{\Tuple{V,\xe}\inps\xaqlS\timesps\xeS}{\ExQ{\zeG\inps\xaS^{\Ill}}{
  \begin{propConjA}
   \zeG\inps\EnabWl{}{\Ill}{\xe}\\
   V\eqps\Set{\zeG}
  \end{propConjA}}}\\
  &\eqps\SetCompX{\Tuple{V,\xe}\inps\xaqlS\timesps\xeS}{V\eqps\EnabWl{}{\Ill}{\xe}}
  =\BR{\R^{l\qsupb}}^{-1}\\
  \R\hspace{-0.1cm}\circ\hspace{-0.1cm}\BR{\R^{l\qsupb}}^{-1}
  &\eqps\SetCompX{\Tuple{\zeG,\xe}\inps\xaS^{\Ill}\timesps\xeS}{\ExQ{V\inps\xaqlS}{
  \begin{propConjA}
   V\eqps\EnabWl{}{\Ill}{\xe}\\
   V\eqps\Set{\zeG}
  \end{propConjA}}}\\
  &\eqps\SetCompX{\Tuple{\zeG,\xe}\inps\xaS^{\Ill}\timesps\xeS}{\zeG\inps\EnabWl{}{\Ill}{\xe}}
  =\R^\lsup
 \end{align*}
With this observations (iii) follows immediately from the transitivity of simulation relations.

%
%

\end{document}